\documentclass[a4paper,USenglish,cleveref, autoref]{lipics-v2019}

\newcommand{\floor}[1]{\left\lfloor #1 \right\rfloor}
\usepackage{hyphenat}
\usepackage[basic]{complexity}
\usepackage{cite}
\usepackage[utf8]{inputenc}
\usepackage{microtype}
\usepackage{placeins}
\usepackage{thmtools, thm-restate,hyperref}
\DeclareMathOperator*{\argmax}{arg\,max}

\hideLIPIcs

\theoremstyle{plain}
\newtheorem{observation}[theorem]{Observation}

\title{Independent Sets of Dynamic Rectangles: Algorithms and Experiments}

\titlerunning{Independent Set on Dynamic Rectangles: Algorithms and Experiments}%

\author{Sujoy Bhore}{Indian Institute of Science Education and Research, Bhopal, India.}{sujoy.bhore@gmail.com}{0000-0003-0104-1659}{}

\author{Guangping Li}{Algorithms and Complexity Group, TU Wien, Vienna, Austria}{guangping@ac.tuwien.ac.at}{https://orcid.org/0000-0002-7966-076X}{}

\author{Martin N\"ollenburg}{Algorithms and Complexity Group, TU Wien, Vienna, Austria}{noellenburg@ac.tuwien.ac.at}{https://orcid.org/0000-0003-0454-3937}{}

\authorrunning{S. Bhore, G. Li, and M. Nöllenburg}%

\Copyright{Sujoy Bhore, Guangping Li, and Martin Nöllenburg}%

\usepackage{xspace}

\newcommand{\grid}{\emph{grid}\xspace}
\newcommand{\gridk}[1]{\emph{grid-$#1$}\xspace}
\newcommand{\linea}{\emph{line}\xspace}
\newcommand{\misr}{\emph{MIS-ORS}\xspace}
\newcommand{\misg}{\emph{MIS-graph}\xspace}

\newcommand{\ggrid}{\emph{g-grid}\xspace}
\newcommand{\ggridk}[1]{\emph{g-grid-$#1$}\xspace}
\newcommand{\glinea}{\emph{g-line}\xspace}
\newcommand\maxset{\textsc{Max-IS}\xspace}
\newcommand{\indset}{\textsc{MIS}\xspace}
\newcommand{\MAXSAT}{\textsc{MAXSAT}\xspace}

\ccsdesc[500]{Theory of computation~Computational geometry}
\ccsdesc[300]{Theory of computation~Dynamic graph algorithms}

\keywords{Independent Sets, Dynamic Algorithms, Rectangle Intersection Graphs, Approximation Algorithms, Experimental Evaluation}

\category{}

\supplement{Source code and benchmark data at \url{https://dyna-mis.github.io/dynaMIS/}.}

\funding{Research supported by the Austrian Science Fund (FWF), grant P 31119.}
\nolinenumbers %

\EventEditors{}
\EventNoEds{1}
\EventLongTitle{}
\EventShortTitle{}
\EventAcronym{}
\EventYear{}
\EventDate{}
\EventLocation{}
\EventLogo{}
\SeriesVolume{}
\ArticleNo{}

\begin{document}
	
	\maketitle
	\begin{abstract}
{Map labeling} is a classical problem in cartography and  geographic information systems (GIS) that asks to place labels for area, line, and point features, with the goal to select and place the maximum number of independent, i.e., overlap-free, labels.
A practically interesting case is point labeling with axis-parallel rectangular labels of common size. In a fully dynamic setting, at each time step, either a new label appears or an existing label disappears. 
Then, the challenge is to maintain a maximum cardinality subset of pairwise independent labels with sub-linear update time. 
Motivated by this, we study the maximal independent set (\indset) and maximum independent set (\maxset) problems on fully dynamic (insertion/deletion model) sets of axis-parallel rectangles of two types---(i) uniform height and width and (ii) uniform height and arbitrary width; both settings can be modeled as rectangle intersection graphs.

We present the first deterministic algorithm for maintaining a \indset (and thus a $4$-approximate \maxset) of a dynamic set of uniform rectangles with polylogarithmic update time. This breaks the natural barrier of $\Omega(\Delta)$
update time (where $\Delta$ is the maximum degree in the graph) for \emph{vertex updates} presented by Assadi et al.~(STOC 2018). We continue by investigating \maxset and provide a series of deterministic dynamic approximation schemes. For uniform rectangles, we first give an algorithm that maintains a $4$-approximate \maxset with $O(1)$ update time. In a subsequent algorithm, we establish the trade-off between approximation quality $2(1+\frac{1}{k})$ and update time 
$O(k^2\log n)$, for $k\in \mathbb{N}$. We conclude with an algorithm that maintains a $2$-approximate \maxset for dynamic sets of unit-height and arbitrary-width rectangles with $O(\log^2 n + \omega\log n)$ update time, where $\omega$ is the maximum size of an independent set of rectangles stabbed by any horizontal line. We have implemented our algorithms and report the results of an experimental comparison exploring the trade-off between solution quality and update time for synthetic and real-world map labeling instances. 
We made several major observations in our empirical study: 1. The original approximations are well above their respective worst-case ratios. 2. In comparison with the static approaches, the dynamic approaches show a significant speed-up in practice.
3. The approximation algorithms show their predicted relative behavior. The better the solution quality, the worse the update times.
4. A simple greedy augmentation to the approximate solutions of the algorithms boost the solution sizes significantly in practice.
\end{abstract}

\ccsdesc[500]{Theory of computation~Computational geometry}
\ccsdesc[100]{Theory of computation~Data structures design and analysis}
\ccsdesc[300]{Theory of computation~Dynamic graph algorithms}

\keywords{Independent Sets, Dynamic Algorithms, Rectangle Intersection Graphs, Approximation Algorithms, Experimental Evaluation}

\maketitle
\section{Introduction}
\textsc{Map Labeling} is a classical problem in cartography and geographic information systems (GIS), that has received significant attention in the past few decades and is concerned with selecting and positioning labels on a map for area, line, and point features. The focus in the computational geometry community has been on labeling point features~\cite{agarwal1998label,fw-ppwalm-91, ww-pla-95, DBLP:conf/compgeom/KreveldSW98}. The labels are typically modeled as the bounding boxes of short names, which correspond precisely to unit height, but arbitrary width rectangles; alternatively, labels can be standardized icons or symbols, which correspond to rectangles of uniform size. In map labeling, a key task is in fact to select an independent (i.e., overlap-free) set of labels from a given set of candidate labels. Commonly the optimization goal is related to maximizing the number of labels. Given a set $\mathcal{R}$ of rectangular labels, \textsc{Map Labeling} is essentially equivalent to the problem of finding a maximum independent set in the intersection graph %
induced by $\mathcal{R}$.  

The independent set problem  is a fundamental graph problem with a wide range of applications. Given a graph $G=(V,E)$, a set of vertices $M\subset V$ is \emph{independent} if no two vertices in $M$ are adjacent in $G$. A \emph{maximal independent set} (MIS) is an independent set that is not a proper subset of any other independent set. A \emph{maximum independent set} (\maxset) is a maximum cardinality independent set. 
While \maxset is one of~Karp's 21 classical \NP-complete problems~\cite{k-racp-72}, computing an \indset can easily be done by a simple greedy algorithm in $O(|E|)$ time. The \indset problem has been studied in the context of several other prominent problems, e.g., graph coloring~\cite{linial1987distributive}, maximum matching~\cite{hopcroft1973n}, and vertex cover~\cite{nguyen2008constant}. On the other hand, \maxset serves as a natural model for many real-life optimization problems that arise in the fields of cartography, scheduling, computer graphics, information retrieval, etc.; see~\cite{agarwal1998label, vanBevern2015, DBLP:journals/tog/SanderNCH08, pardalos1994maximum}. 

Stronger results for independent set problems in geometric intersection graphs are known in comparison to general graphs. 
For instance, it is known that \maxset on general graphs cannot be approximated better than $|V|^{1-\epsilon}$ in polynomial time for any $\epsilon>0$ unless $\NP= P$~\cite{DBLP:journals/toc/Zuckerman07}.
In contrast, a randomized polynomial-time algorithm exists that computes for rectangle intersection graphs an $O(\log \log n)$-approximate solution to \maxset with high probability~\cite{chalermsook2009maximum}, as well as QPTASs~\cite{aw-asmwir-13,ce-amir-16}. 
Very recently, the constant factor approximation schemes have been developed for the \maxset on rectangle intersection graphs; see~\cite{DBLP:journals/corr/abs-2101-00326, DBLP:journals/corr/abs-2106-00623}.
The \maxset problem is already \NP-hard on unit square intersection graphs~\cite{DBLP:journals/ipl/FowlerPT81}, however, it admits a polynomial-time approximation scheme (\textsf{PTAS}) for unit square intersection graphs~\cite{erlebach2005polynomial} and more generally for pseudo disks~\cite{chan2012approximation}. 
Moreover, for rectangles with either uniform size or at least uniform height and bounded aspect ratio, the size of an \indset is not arbitrarily worse than the size of a \maxset. For instance, any \indset of a set of uniform rectangles is a $4$-approximate solution to the \maxset problem, since each rectangle can have at most four independent neighbors.

Past research has mostly considered static label sets in static  maps~\cite{agarwal1998label,fw-ppwalm-91,ww-pla-95,DBLP:conf/compgeom/KreveldSW98} and in dynamic maps allowing zooming~\cite{bnpw-oarcd-10} or rotations~\cite{gnr-clrm-16}, but not fully dynamic label sets with insertions and deletions of labels. 
Recently, Klute et al.~\cite{kllns-esl-19} proposed a framework for semi-automatic label placement, where domain experts can interactively insert and delete labels. %
In their setting an initially computed large independent set of labels can be interactively modified by a cartographer, who can easily take context information and soft criteria such as interactions with the background map or surrounding labels into account.
Standard map labeling algorithms typically do not handle such aspects well~\cite{DBLP:conf/compgeom/FormannW91, DBLP:journals/cartographica/RylovR14}.
Based on these modifications (such as deletion, forced selection, translation, or resizing), the solution is updated by a dynamic algorithm while adhering to the new constraints.
Another scenario for dynamic labels are maps, in which features and labels (dis-)appear over time, e.g., based on a stream of geotagged, uniform-size photos posted on social media or, more generally, maps with labels of dynamic spatio-temporal point sets~\cite{Gabriel2015}. 
For instance, a geo-located event that happens at time $t$ triggers the availability of a new label for a certain period of time, after which it vanishes again. Examples beyond social media are reports of earthquakes, forest fires, or disease incidences.
While traditional geographic map labeling deals with small and relatively static label sets, labeling of social network data, especially the ones used in anomaly detection and visual analytics usually deal with vast and dynamic label set; see~\cite{DBLP:conf/apvis/ThomBKWE12, DBLP:conf/ieeevast/MacEachrenJRPSMZB11}. 
Furthermore, note that these applications often run on devices with limited computational resources, e.g., mobile devices. Therefore, it is desirable to 
design dynamic algorithms that can handle the changes in an efficient and robust manner. 
Motivated by this, we study the independent set problem for dynamic sets of axis-parallel rectangles of two types:
\begin{itemize}
    \item rectangles of uniform height and width
    \item rectangles of uniform height and arbitrary width
\end{itemize}
We consider fully dynamic algorithms for maintaining independent sets under insertions and deletions of rectangles, i.e., vertex insertions and deletions in the corresponding dynamic rectangle intersection graph.

\emph{Dynamic} graphs %
are subject to discrete changes over time, i.e., insertions or deletions of vertices or edges~\cite{EppGalIta-ATCH-99}. A dynamic graph algorithm solves a computational problem, such as the independent set problem, on a dynamic graph by updating efficiently the previous solution as the graph changes over time, rather than recomputing it from scratch. 
A dynamic graph algorithm is called \emph{fully dynamic} if it allows both insertions and deletions, and \emph{partially dynamic} if only insertions or only deletions are allowed. While general dynamic independent set algorithms can obviously 
also 
be applied to %
rectangle intersection graphs, our goal is to exploit their geometric properties %
to obtain more efficient algorithms. %

\subparagraph*{\textbf{Related Work.}}
There has been a lot of work on dynamic graph algorithms in the last decade and dynamic algorithms still receive considerable attention in theoretical computer science. We point out some of these works, e.g., on spanners~\cite{bernstein2019deamortization}, vertex cover~\cite{bhattacharya2017deterministic}, set cover~\cite{abboud2019dynamic}, graph coloring~\cite{bhattacharya2018dynamic}, and maximal matching~\cite{gamlath2019online}.
In particular, the maximal independent set problem on dynamic graphs with edge updates has attracted significant attention in the last two years~\cite{DBLP:conf/stoc/AssadiOSS18, assadi2019fully, behnezhad2019fully, chechik2019fully, DBLP:conf/icalp/CormodeDK19}.
For vertex insertion/deletion, an {MIS} can be maintained dynamically in $O(\Delta)$ update time by using the recent algorithm of Assadi et al.~\cite{DBLP:conf/stoc/AssadiOSS18}, where $\Delta$ is the maximum degree of the intersection graph.

Recently, Henzinger et al. \cite{henzinger_et_al:LIPIcs:2020:12209} studied the \maxset problem for intervals, hypercubes and hyperrectangles in $d$ dimensions, with special assumptions. They assumed that the objects are axis-parallel and contained in the space $[0, N]^d$; the value of $N$ is given in advance, and each edge of an input object has length at least $1$ and at most $N$. They have presented dynamic approximation schemes with the update time $polylog(n,N)$, where $n$ is the instance size. We note that in general, $N$ might
be exponential in $n$ or even unbounded, thus those bounds are not sublinear in $n$ in the
general case. Subsequently, Bhore et al.~\cite{DBLP:journals/corr/abs-2007-08643} 
designed a dynamic approximation scheme for dynamic intervals that maintains a $(1+\epsilon)$-approximate maximum independent set in $O_{\epsilon}(\log n)$ update time, where
$\epsilon>0$ is any positive constant and the notation $O_\epsilon$ hides terms depending only on $\epsilon$. Gavruskin et al.~\cite{gavruskin2015dynamic} studied the \maxset problem for dynamic proper intervals (intervals cannot contain one another), and showed how to maintain a \maxset with polylogarithmic update time.

There is a long history of the empirical study of map-labeling problems. This chain of research started with the work of Christensen et al.~\cite{DBLP:journals/tog/ChristensenMS95}. They proposed two methods: One based on a discrete form of gradient descent and the other on simulated annealing. An alternative approach was presented by Wagner and Wolff~\cite{DBLP:journals/comgeo/WelzlWW97} for the labeling problem, who used the sample data in the experimental evaluation that consists of three different classes of random problems and a selection of problems arising in the production of groundwater quality maps by the authorities of the City of Munich.
Nascimento and Eades~\cite{DBLP:journals/vlc/NascimentoE08} proposed a practically motivated framework, called \emph{user hints}, and proposed an interactive map-labeling system based on this along with its evaluation. This type of user-interactive approach was empirically 
studied by Klute et al.~\cite{kllns-esl-19}. Moreover, other aspects of dynamic map labeling, e.g., rotation, zooming, have been studied over the years; see~\cite{DBLP:journals/jea/GemsaNR16, DBLP:journals/tvcg/BeenDY06, DBLP:journals/comgeo/BeenNPW10}. De Berg and Gerrits~\cite{DBLP:conf/esa/BergG13} developed and experimentally evaluated a heuristic for labeling moving points on static maps.

\subparagraph*{\textbf{Results and Organization.}}
We study {MIS} and \maxset problems for dynamic sets of $O(n)$ axis-parallel rectangles of two types: (i) congruent rectangles of uniform height and width %
and (ii) rectangles of uniform height and arbitrary width. 

In this paper we design and implement algorithms for dynamic \indset %
and \maxset that demonstrate the trade-off between update time and approximation factor, both from a theoretical perspective and in an experimental evaluation. In contrast to the recent dynamic \indset algorithms, which are randomized~\cite{DBLP:conf/stoc/AssadiOSS18, assadi2019fully, behnezhad2019fully, chechik2019fully}, our algorithms are deterministic. %

In Section~\ref{mis-log} we
present an algorithm that maintains an \indset of a dynamic set of unit squares in $O(\log n \log\log n)$ update time or, alternatively, with sub-logarithmic amortized update time, improving the best-known update time $\Omega(\Delta)$ by Assadi et al.~\cite{DBLP:conf/stoc/AssadiOSS18}, where $\Delta$ is the maximum degree of the intersection graph. A major, but generally unavoidable bottleneck of that algorithm is that the entire graph is stored explicitly, and thus insertions/deletions of vertices take 
$\Omega(\Delta)$ time. 
We use structural geometric properties of the unit squares along with a dynamic orthogonal range searching data structure to bypass the explicit intersection graph and overcome this bottleneck. 

In Section~\ref{MIS-const}, we study the \maxset problem. %
For dynamic unit squares, we give an algorithm that %
maintains a $4$-approximate \maxset with $O(1)$ update time.
We  generalize this algorithm and improve the approximation factor to $2(1+\frac{1}{k})$, which increases the update time to $O(k^2\log n)$. %
We conclude with an algorithm that %
maintains a $2$-approximate \maxset for a dynamic set of unit-height and arbitrary-width rectangles (in fact, for a dynamic interval graph, which is of independent interest) with  $O(\log^2 n + \omega\log n)$ update time, where $\omega$ is 
the maximum size of an independent set of rectangles stabbed by any horizontal line.

Finally, Section~\ref{exp} provides an experimental evaluation of the proposed \maxset approximation algorithms on synthetic and real-world map labeling data sets. The experiments explore the trade-off between solution size and update time, as well as the speed-up of the dynamic algorithms over their static counterparts. See the supplemental material\footnote{Source code and the benchmark data are available on \url{https://dyna-mis.github.io/dynaMIS/}.} for source code and benchmark data. %

\section{Model and Notation}\label{sec:model}
For every $N \in \mathbb N$, $[N]$ denotes the set $\{1, 2, \ldots, N\}$.
Let $R$ be a dynamic set of axis-parallel, unit-height rectangles in the plane, which is dynamically updated by a sequence of $N \in \mathbb N$ insertions and deletions.
Let $R_i$ denote the set of rectangles at step $i \in [N]$ and let $n=\max \{|R_i| \mid i \in [N] \}$ be the maximum number of rectangles over all steps.
The rectangle intersection graph defined by $R_i$ at time step $i$ is denoted as $G_i=(R_i,E_i)$, where two rectangles $r,r' \in R_i$ are connected by an edge $\{r,r'\} \in E_i$ if and only if $r \cap r' \ne \emptyset$.
We use $M_i$ to denote a maximal independent set in $G_i$, and $OPT_i$ to denote a maximum independent set in $G_i$.
For a graph $G=(V,E)$ and a vertex $v\in V$, let $N(v)$ denote the set of neighbors of $v$ in $G$. %
This notation also extends to any subset $U \subseteq V$ by defining  $N(U) = \bigcup_{v\in U} N(v)$. %
We use $\deg(v)$ to denote the degree of a vertex $v \in V$. %
For any vertex $v\in V$, let $N^{r}(v)$ %
be the $r$-neighborhood %
of $v$, i.e., the set of vertices that are within distance at most $r$ from $v$ (excluding $v$).

We study the independent set problem for dynamic sets of axis-parallel rectangles of two types---(i) unit rectangles and (ii) rectangles of unit height and arbitrary width.
In this work, we may assume that the unit rectangles are unit squares.
If the rectangles of $R$ are of uniform height and width, we can use an affine transformation to map $R$ to a set of unit squares $S$ and map $R_i$ to unit square set $S_i$ for $i \in [N]$.  
We further define the set $C_i$ be the corresponding centers of squares of $S_i$.

\section{Algorithms for Dynamic Maximal Independent Set}
\label{mis-log}
In this section, we study the \indset problem for dynamic uniform rectangles. 
As stated before we can assume w.l.o.g.\ that the rectangles are unit squares. 
We design an algorithm that maintains a {MIS} for a dynamic set of $O(n)$ unit squares in polylogarithmic update time. 
Assadi et al.~\cite{DBLP:conf/stoc/AssadiOSS18} presented an algorithm for maintaining a {MIS} on general dynamic graphs with $O(\Delta)$ update time, where $\Delta$ is the maximum degree in the graph. 
In the worst case, however, that algorithm takes $O(n)$ update time. 
In fact, it seems unavoidable for an algorithm that explicitly maintains the (intersection) graph to perform an \indset update in less than $\Omega(\deg(v))$ time for an insertion/deletion of a vertex $v$.
In contrast, our proposed algorithm in this section does not explicitly maintain the intersection graph $G_i=(\mathcal{S}_i,E_i)$ (for any $i\in [N]$), but rather only the set of squares $\mathcal S_i$ in a suitable dynamic geometric data structure. 
For the ease of explanation, however, we do use graph terms at times. 

Let $i \in [N]$ be any time point in the sequence of updates. 
For each square $s_v\in \mathcal{S}_i$, let $s^a_v$ be a square of side length $a$ concentric with $s_v$. 
Further, let $M_i$ denote the \indset that we compute for $G_i=(\mathcal{S}_i, E_i)$, and let $\mathcal{C}(M_i)\subseteq \mathcal{C}_i$ be their corresponding square centers. 
We maintain two fully dynamic orthogonal range searching data structures, which maintain a set of points dynamically and support efficient deletions and insertions of points, throughout: 
(i) a dynamic range tree $T(\mathcal{C}_i)$ for the entire point set $\mathcal{C}_i$ and (ii) a dynamic range tree $T(\mathcal{C}(M_i))$ for the point set $\mathcal{C}(M_i)$ corresponding to the centers of $M_i$.
They can be implemented with dynamic fractional cascading~\cite{mn-dfc-90}, which yields $O(\log n \log\log n)$ update time and $O(k + \log n \log \log n)$ query time for reporting $k$ points.

We compute the initial \indset $M_1$ for $G_1=(\mathcal{S}_1, E_1)$ by using a simple linear-time greedy algorithm. 
First we initialize the range tree $T(\mathcal{C}_1)$.
Then we iterate through the set $\mathcal S_1$ as long as it is not empty, select a square $s_v$ for $M_1$ and insert its center into $T(\mathcal{C}(M_1))$, find its neighbors $N(s_v)$ by a range query in $T(\mathcal{C}_1)$ with the concentric square $s^2_v$, and delete $N(s_v)$ from $\mathcal S_1$.
It is clear that once this process terminates, $M_1$ is an MIS.%

When we move in the next step from $G_i=(\mathcal{S}_i,E_i)$ to $G_{i+1}=(\mathcal{S}_{i+1},E_{i+1})$, 
either a square is inserted into $\mathcal{S}_i$ or deleted from $\mathcal{S}_i$. 
Let $s_x$ be the square that is inserted or deleted. 
\begin{figure}[tb]
\centering
\includegraphics[page=3]{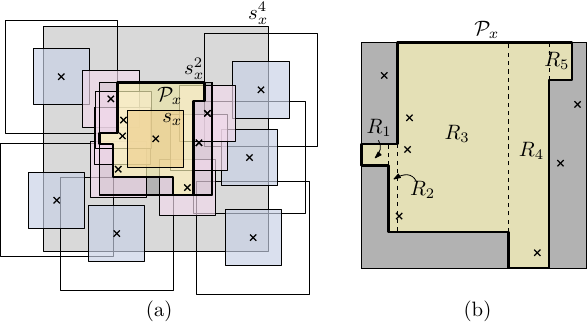}
\caption{Example for the deletion of a square $s_x$. (a) Square~$s_x$, its neighborhood with centers in $s_x^2$, its 2-neighborhood with centers in $s_x^4$, and the polygon $\mathcal P_x$. (b) Vertical slab partition of $\mathcal P_x$.}
\label{fig:rectilinear}
\end{figure}
\begin{lemma} \label{lem:4neighbors}
Given an arbitrary set of pairwise overlap-free unit squares $S$ and an arbitrary square $r$ of side length $2$, $r$ contains at most four centers of unit squares of $S$.
\end{lemma}

\begin{proof}
We split $r$ equally into four parts where each quarter corresponds to a unit square. 
Consider two squares $s_1,s_2\in S$.
If their centers lie in the same quarter, then they must overlap. 
Hence, by a simple packing argument the claim holds. 
\end{proof}

\medskip
\noindent\underline{\textsc{Insertion:}} 
When we insert a square $s_x$ into $\mathcal{S}_i$ to obtain $\mathcal{S}_{i+1}$, we do the following operations. 
First, we obtain $T(\mathcal{C}_{i+1})$ by inserting the center of $s_x$ into $T(\mathcal{C}_i)$. Next, we have to detect whether $s_x$ can be included in $M_{i+1}$.
If there exists a square $s_u$ from $M_i$ intersecting $s_x$, we should not include $s_x$; otherwise we will add it to the MIS.
To check this, we search with the range $s^2_x$ in $T(\mathcal{C}(M_i))$. 
By Lemma~\ref{lem:4neighbors}, we know that no more than four points (the centers of four independent squares) of $\mathcal{C}(M_i)$ can be in the range $s^2_x$.
If the query returns such a point, then $s_x$ would intersect with another square in $M_i$ and we set $M_{i+1} = M_i$.
Otherwise, we add $s_x$ to the current solution $M_i$ and to the tree $T(\mathcal{C}(M_i))$ to obtain 
$M_{i+1}$ and $T(\mathcal{C}(M_{i+1}))$.
%
%

\medskip
\noindent\underline{\textsc{Deletion:}} 
When we delete a square $s_x$ from $\mathcal{S}_i$, %
it is possible that $s_x \in M_i$.
In this case we may have to add squares from $N(s_x)$ into $M_{i+1}$ to keep it maximal. 
Since any square can have at most four independent neighbors, we can add in this step up to four squares to~$M_{i+1}$. %

First, we check if $s_x \in M_i$. 
If not, then we simply delete $s_x$ from $T(\mathcal{C}_i)$ to get $T(\mathcal{C}_{i+1})$ and set $M_{i+1}=M_i$.
Otherwise, we delete again $s_x$ from $T(\mathcal{C}_i)$ and also from $T(\mathcal{C}(M_i))$.
In order to detect which neighbors of $s_x$ can be added to $M_i$, we use suitable queries in the data structures $T(\mathcal{C}(M_i))$ and $T(\mathcal{C}_i)$.
Figure~\ref{fig:rectilinear}a illustrates the next observations.
The centers of all neighbors in $N(s_x)$ must be contained in the square $s^2_x$.
But some of these neighbors may intersect other squares in $M_i$.
In fact, these squares would by definition belong to the 2-neighborhood, i.e., be in the set $Q_x = N^2(s_x) \cap M_i$.
We can obtain $Q_x$ by querying $T(\mathcal{C}(M_i))$ with the range $s_x^4$.
Since $s_x \in M_i$, we know that no center point of squares in $M_i$ lie in $s^2_x$. Hence, the center points of the squares in $Q_x$ lie in the annulus $s^4_x - s^2_x$.
A simple packing argument (similar to the proof of Lemma~\ref{lem:4neighbors}) implies that $|Q_x| \le 12$ and therefore querying $T(\mathcal{C}(M_i))$ will return at most $12$ points.

Next we define the rectilinear polygon $\mathcal{P}_x = s^2_x - \bigcup_{s_y \in Q_x} s_y^2$, which contains all possible center points of squares that are neighbors of $s_x$ but do not intersect any square $s_y \in M_i \setminus \{s_x\}$.  %

\begin{observation}
  \label{rect-poly}
The polygon $\mathcal{P}_x$ has at most $28$ corners.
\end{observation}%
\begin{proof}
	We know that $Q_x$ contains at most 12 squares $s_y$, for each of which we subtract $s_y^2$ from~$s_x^2$. 
	Since all squares have the same side length, at most two new corners can be created in $\mathcal{P}_x$ when subtracting a square $s_y^2$. Initially $\mathcal{P}_x$ had four corners, which yields the claimed bound of at most $28$ corners.
\end{proof}

Next we want to query $T(\mathcal{C}_i)$ with the range $\mathcal{P}_x$, which we do by vertically partitioning $\mathcal{P}_x$ into rectangular slabs $R_1, \dots, R_c$ for some $c \le 28$ (see Figure~\ref{fig:rectilinear}b).
For each slab $R_j$, where $1 \le j \le c$, we perform a range query in $T(\mathcal{C}_i)$. If a center $p$ is returned, we can add the corresponding square $s_p$ into $M_{i+1}$, and $p$ into $T(\mathcal{C}(M_i))$ to obtain $T(\mathcal{C}(M_{i+1}))$.
Moreover, we have to update $\mathcal{P}_x \gets \mathcal{P}_x - s_p$, refine the slab partition and continue querying $T(\mathcal{C}_i)$ with the slabs of $\mathcal{P}_x$.
Observe that after cutting the square $s_p$ from the rectilinear polygon $\mathcal{P}_x$, the number of sides of the remaining region of $\mathcal{P}_x$ can be increased by at most $4$.
We know that the deleted square $s_x$ can have at most four independent neighbors. 
So after adding at most four new squares to $M_{i+1}$ we know that there is no center point in $C_i$ in the range $\mathcal{P}_x $ and we can stop searching.

\begin{lemma}\label{lem:rangei}
	The  set $M_i$ is a maximal independent set of $G_i=(\mathcal{S}_i,E_i)$ for each step $i \in [N]$. 
\end{lemma}
\begin{proof}
	The correctness proof is inductive. 
	By construction the initial set $M_1$ is an \indset for $G_1$.
	Let us consider some step $i>1$ and assume by induction that $M_{i-1}$ is an \indset for $G_{i-1}$.
	If a new square $s_x$ is inserted in step $i$, we add it to $M_i$ if it does not intersect any other square in $M_{i-1}$; otherwise we keep $M_{i-1}$. In either case $M_i$ is an \indset of $G_i$.
	If a square $s_x$ is deleted in step $i$ and $s_x \not\in M_{i-1}$, then $M_i = M_{i-1}$ is an \indset of $G_i$.
	Finally, let $s_x \in M_{i-1}$. Assume for contradiction that $M_i$ is not an MIS, i.e., some square $s_q$ could be added to $M_i$. Since $M_{i-1}$ was an MIS, $s_q \in N(s_x)$ and thus its center must lie in the region $\mathcal{P}_x$. But then we would have found $s_q$ in our range queries with the slabs of $\mathcal{P}_x$. Hence $M_i$ is indeed an \indset of $G_i$.
\end{proof}

\begin{theorem} %
\label{thm:mis-rt}
We can maintain a maximal independent set of a dynamic set of unit squares, deterministically, in $O(\log n \log\log n)$ update time and $O(n)$ space. 
\end{theorem}
\begin{proof}
The correctness follows from Lemma~\ref{lem:rangei}. It remains to show the running   time for the fully dynamic updates. 
 At each step $i$ we perform either an \textsc{Insertion} or a \textsc{Deletion} operation. 
Let us first discuss the update time for the insertion of a square. 
As described above, an insertion performs one or two insertions of the center of the square into the range trees and one range query in $T(\mathcal{C}(M_{i-1}))$, which will return at most four points.
Since we use the data structure of Mehlhorn and Näher~\cite{mn-dfc-90}, the update time for inserting a square is $(\log n \log\log n)$, which corresponds to the time requires for inserting a new point into their range searching data structure and one range query.
The deletion of a square triggers either just a single deletion from the range tree $T(\mathcal C_{i-1})$ or, if it was contained in the \indset $M_{i-1}$, two deletions, up to four insertions, and a sequence of range queries: one query in $T(\mathcal{C}(M_{i-1}))$, which can return at most $12$ points and a constant number of queries in $T(\mathcal C_{i-1})$ with the constant-complexity slab partition of $\mathcal{P}_x$. 
Note that while the number of points in $\mathcal{P}_x$ can be large, for our purpose it is sufficient to return a single point in each query range if it is not empty.
Therefore, the update time for a deletion is again $O(\log n \log \log n)$ with dynamic fractional cascading~\cite{mn-dfc-90}.

In this approach, we maintain two dynamic range trees for the center points of  rectangles. We use the dynamic range tree structure by  Mehlhorn and N\"{a}her~\cite{mn-dfc-90}, whose space requirement is linear in the number of elements stored. Thus, the space requirement of this approach is $O(n)$. 
\end{proof}

For unit square intersection graphs, recall that any square in an \indset can have at most four mutually independent neighbors. Therefore, maintaining a dynamic \indset immediately implies maintaining a dynamic 4-approximate \maxset.

Note that the update time of the dynamic data structure for orthogonal range queries dominates the update time of this algorithm.
Its update time can be improved by using a state-of-the-art dynamic range query structure.
The best-known dynamic data structure for orthogonal range reporting requires $O(\log^{2/3 + \epsilon} n)$ amortized update time, where $\epsilon$ denotes an arbitrarily small positive constant, and $O(k + \frac{\log n}{\log\log n})$ amortized query time, where $k$ is the number of reported points~\cite{ct-dorsr-17}.
From this, we conclude the following corollary.

\begin{corollary}\label{cor:mis-4apx}
We can maintain a 4-approximate maximum independent set of a dynamic set of unit squares, in amortized $O(\frac{\log n}{\log\log n})$ update time. 
\end{corollary}

\section{Approximation Algorithms for Dynamic Maximum Independent Set}\label{MIS-const}
In this section, we study the \maxset problem for dynamic unit squares as well as for unit-height and arbitrary-width rectangles. 
In a series of dynamic schemes proposed in this section, we establish the trade-off between the update time and the solution size, i.e., the approximation factors. 
First, we design a $4$-approximation algorithm with $O(1)$ update time for \maxset of dynamic unit squares (Section~\ref{4-approx}).
We generalize this to an algorithm that maintains a $2(1+\frac{1}{k})$-approximate \maxset with $O(k^2 \log n)$ update time, for any integer $k>1$ (Section~\ref{sub:2kapx}) . 
Finally, we conclude with an algorithm that deterministically maintains a $2$-approximate \maxset with $O(\log^2 n + \omega\log n)$ update time, where $\omega$ is the maximum size of an independent set of the unit-height rectangles stabbed by any horizontal line (Section~\ref{sub:2apx}).

Let $\mathcal{B}$ be a bounding square of the dynamic set of $1 \times 1$-unit squares $\bigcup_{i \in [N]} \mathcal S_i$ of side length $\sigma\times \sigma$.%
Let $H=\{h_1,\ldots,h_{\sigma}\}$ and $L=\{l_1,\ldots,l_{\sigma}\}$ be a set of top-to-bottom and left-to-right ordered equidistant horizontal and vertical lines partitioning $\mathcal B$ into a square grid of side-length-$1$ cells, see Figure~\ref{fig:4-approx}.
Let $E_H = \{h_i \in H \mid i = 0 \pmod{2}\}$ and $O_H = \{h_i \in H \mid i = 1 \pmod{2}\}$ be the set of even and odd horizontal lines, respectively. 

\subsection{4-Approximation Algorithm with Constant Update Time}\label{4-approx}
We design a $4$-approximation algorithm for the \maxset problem on dynamic unit square intersection graphs with constant update time. 
Our algorithm is based on a grid partitioning approach. 
Consider the square grid on $\mathcal B$ induced by the sets $H$ and $L$ of horizontal and vertical lines. 
We denote the grid points as $g_{p,q}$ for $p,q \in [\sigma]$, where $g_{p,q}$ is the intersection point of lines $h_p$ and $l_q$.
We assign each unit square in any set $\mathcal S_i$ to a grid point (denoted by its associated grid point) in the following deterministic way. 
Due to the unit grid construction, each unit square intersects at least one grid point.
If an unit square $s \in S_i$ contains exactly one grid point $g_{p,q}$, we associate $s$ with $g_{p,q}$. 
Otherwise, if $s$ contains multiple grid points, we assign $s$ to the top-leftmost grid point among others. 
For the ease of description, we may assume that the dynamic squares are in general positions, i.e., each square contains exactly one grid point. Moreover, the above assignment does not affect the algorithm description or the analysis.
For each $g_{p,q}$, we store a Boolean \emph{activity value} $1$ or $0$ based on its intersection 
with $\mathcal S_i$ (for any step $i\in [N]$). 
If $g_{p,q}$ intersects at least one square of $S_i$, we say that it is \emph{active} and set the value to $1$; otherwise, we set the value to $0$. 
Observe that for each grid point $g_{p,q}$ and each time step $i$ at most one square of $\mathcal{S}_i$ intersecting $g_{p,q}$ can be chosen in any \maxset. 
This holds because all squares that intersect the same grid point form a clique in $G_i$, and at most one square from a clique can be chosen in any independent set. 
\begin{figure}[t]
\centering
\includegraphics{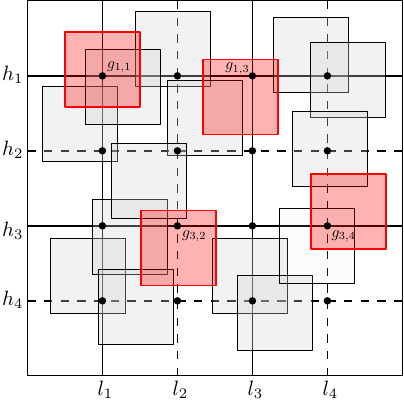}
\caption{Example instance with bounding square $\mathcal{B}$ partitioned into a $5 \times 5$ grid. Red squares represent the computed $4$-approximate solution, which here is $M(O(H))$.}
\label{fig:4-approx}
\end{figure}

For each grid point $g_{p,q}$, for some $p,q \in [\sigma]$, we store the squares in $S$ that intersect $g_{p,q}$ into a list $L_{pq}$. 
Moreover, a counter for the size of $L_{pq}$ is maintained dynamically for each grid point $g_{p,q}$ such that we could detect if there exists at least one square intersecting $g_{p,q}$ efficiently. 
Note that the set $L_{pq}$ should allow constant time insertion and deletion, e.g., be stored in sequence containers like lists.
We first initialize an independent set $M_1$ for %
$G_1=(\mathcal{S}_1,E_1)$ with $|M_1|\ge |OPT_1|/4$.
For each horizontal line $h_j\in H$, we compute two independent sets $M^1_{h_j}$ and $M^2_{h_j}$, where $M^1_{h_j}$ (resp.\ $M^2_{h_j}$) contains an arbitrary square intersecting each odd (resp.\ even) grid point on $h_j$.
Since every other grid point is omitted in these sets, any two selected squares are independent.
Let $M(h_j)=\argmax\{|M^1_{h_j}|, |M^2_{h_j}|\}$ be the larger of the two independent sets. 
We define $p(h_j)=|M^1_{h_j}|$ and $q(h_j)=|M^2_{h_j}|$, as well as $c(h_j) = |M(h_j)| =  \max \{p(h_j), q(h_j)\}$. %

We construct the independent sets $M(E_H) = \bigcup_{j=1}^{\floor{ %
\sigma/2
}} (M(h_{2j}))$ for $E_H$ and %
$M(O_H) = \bigcup_{j=1}^{\floor{ %
\sigma/2}} (M(h_{2j-1})$ for $O_H$. %
We return $M_1=\argmax\{|M(E_H)|, |M(O_H)|\}$ as the independent set for $G_1$. %
See Figure~\ref{fig:4-approx} %
for an illustration. 
The initialization of all $O(\sigma^2)$ variables and the computation of the first set $M_1$ take $O(\sigma^2)$ time. (Alternatively, a hash table would be more space efficient, but could not provide the $O(1)$-update time guarantee.)%
\begin{lemma}
\label{valid-g1}
The set $M_1$ is an independent set of $G_1=(\mathcal{S}_1,E_1)$ with $|M_1|\ge|OPT_1|/4$ and can be computed in $O(\sigma^2)$ time.
\end{lemma}	
\begin{proof}
Partition the squares $S_1$ into $2$ sets $S_{E}$, $S_O$, where $S_E$ (resp. $S_O$) consists of all squares intersecting an even (resp. odd) horizontal line. 
Let  $M_E$ and $M_O$ be the \maxset of the $S_E$ and $O_E$ respectively.
Clearly, the larger one of these two sets contains at least half of as many elements as a \maxset of $S_1$.
We may assume, w.l.o.g., that the $M_E\ge |OPT_1|/2$.
For each even horizontal line $h_j$, $M^1_{h_j}$ (resp. $M^2_{h_j}$) is a \maxset of all rectangles stabbed on $h_j$ and an odd (resp. even) vertical line. 
The larger one of these two sets contains at least half as many elements as a \maxset of the squares stabbed by $h_j$. Overall, This implies that $|M_1| \ge |OPT_1|/4$.
\end{proof}

In the following step, when we move from $G_i$ to $G_{i+1}$, for any $i\in [N]$, a square $s_x$ is inserted into $\mathcal S_i$ or deleted from $\mathcal S_i$.
Let $g_{pq}$ be the grid point contained in $s_x$. We update the list $L_{pq}$ by either inserting the square $s_x$ into $L_{pq}$ or deleting $s_x$ from $L_{pq}$. 
Moreover, we update the counter recording the size of $L_{pq}$ and the activity value of the grid point accordingly.
Intuitively, we check the activity value of the grid point that $s_x$ intersects. If the update has no effect on its activity value, we keep $M_{i+1} = M_i$. Otherwise, we update the activity value, the corresponding cardinality counters, and report the solution accordingly. All of these operations can be performed in $O(1)$-time.

A more detailed description of the \textsc{Insertion} and \textsc{Deletion} operations is given in the following.
When we move in the next step from $G_i$ to $G_{i+1}$ (for some $1 \le i < N$), we either insert a new square into $\mathcal S_i$ or delete one square from $\mathcal S_i$.
Let $s_x$ be the square that is inserted or deleted and let $g_{u,v}$ (for some $u,v\in [\sigma]$) be the grid point that intersects $s_x$. 
We next describe how to maintain a 4-approximate \maxset with constant update time.
We distinguish between the two operations \textsc{Insertion} and \textsc{Deletion}. 

\medskip

\noindent\underline{\textsc{Insertion:}} If $g_{u,v}$ is active for $\mathcal S_i$, %
there is at least one square intersecting $g_{u,v}$ that was considered while computing $M_i$. 
Hence, even if we would include $s_x$ in a modified independent set $M_{i+1}$, it would not make any impact on its cardinality. %
Hence, we simply set  $M_{i+1} \gets M_i$. 
Otherwise, we perform a series of update operations: 
(1) Change the activity value of $g_{u,v}$ from $0$ to $1$. 
(2) Include $s_x$ in $M^1_{h_u}$ (resp. $M^2_{h_u}$) if $v$ is odd (resp. even), and increase the value of $p(h_u)$ (resp. $q(h_u)$) by $1$. 
This lets us reevaluate the cardinality $c(h_u)$ of $M(h_u)$ in constant time. 
(3) Reevaluate $M(E_H)$ and $M(O_H)$ and their cardinalities based on the updated value of $c(h_u)$. Note that none of these operations takes more than $O(1)$ time.

\medskip
\noindent\underline{\textsc{Deletion:}} If there is a square $s_l$ other than $s_x$ intersecting $g_{u,v}$, then $g_{u,v}$ stays \emph{active}. 
We replace $s_x$ by $s_l$ in the maintained independent sets $M^1_{h_u}$, $M^2_{h_u}$,$M(E_H)$ and $M(E_H)$.
Note that this makes no impact on the cardinality of the sets $M(E_H)$ and $M(E_H)$.
If there is no other square intersecting $g_{u,v}$, we reset the activity value of $g_{u,v}$ to \emph{false}.
Moreover, we delete $s_x$ from maintained independent sets of line $h_u$ and reevaluate $M_E$ and $M_O$.

The update procedure described above ensures that the respective cardinality maximization for the affected stabbing line $h_j$ and finally $M_i$ is reevaluated and updated.
In this approach, we maintain an $O(\sigma^2)$ grid. Each of $n$ rectangles is stored in one of the grid points, thus the storage of rectangles is $O(n)$. 
Thereby, we conclude the following Lemma~\ref{valid-gi}.

\begin{lemma}
\label{valid-gi}
		The set $M_i$ is an independent set of $G_{i}=(\mathcal{S}_{i},E_{i})$ for each $i \in [N]$ and $|M_{i}|\ge |OPT_{i}|/4$ and $O(\sigma^2 + n)$ space.
\end{lemma}

\subparagraph*{\textbf{Running Time.}} We perform either an insertion or a deletion operation at every step $i \in [N]$. 
Both of theses operations perform only local operations: 
(i) compute the grid point intersecting the updates square and check its activity value; (ii) reevaluate the values $p(h_j)$ and $q(h_j)$ of the horizontal line $h_j$ intersecting the square---this may or may not flip the independent set $M(h_j)$ and its cardinality from $p(h_j)$ to $q(h_j)$, or vice versa; (iii) finally, if the cardinality of $M(h_j)$ changes, we reevaluate the sets $M(E_H)$ and $M(O_H)$. 
All these operations possibly change one activity value, increase or decrease at most three variables by $1$ and perform at most two comparison operation.
Therefore, the overall update process takes $O(1)$ time in each step.
Recall that the process to initialize the data structures for the set $\mathcal S_1$ and to compute $M_1$ for $G_1$ takes $O(\sigma^2)$ time.

Lemmas~\ref{valid-g1} and~\ref{valid-gi} and the above discussion of the $O(1)$ update time yield: %

\begin{theorem}
We can maintain a $4$-approximate maximum independent set in a dynamic unit square intersection graph, deterministically, in $O(1)$ update time.  
\end{theorem}
\subsection[]{$2(1+\frac{1}{k})$-Approximation Algorithm with $O(k^2\log n)$ Update Time}\label{sub:2kapx}

Next, we improve the approximation factor from $4$ to $2(1+\frac{1}{k})$, for any integer $k>1$, by combining the shifting technique~\cite{hm-avlsi-85} with the insights gained from Section~\ref{4-approx}. 
This comes at the cost of an increase of the update time to $O(k^2\log n)$, which illustrates the trade-off between solution quality and update time.
We reuse the grid partition and some notations from Section~\ref{4-approx}. 
We first describe how to obtain a solution $M_1$ for the initial graph $G_1$  that is of size at least $|OPT_1| / 2(1+\frac{1}{k})$ and then discuss how to maintain this under dynamic updates. 
Let $h_j \in H$ be a horizontal stabbing line and let $\mathcal S(h_j)\subseteq \mathcal S$ be the set of squares stabbed by~$h_j$. 
Since they are all stabbed by $h_j$, the intersection graph of $\mathcal S(h_j)$ is  equivalent to the unit interval intersection graph obtained by projecting each unit square $s_x \in \mathcal S(h_j)$ to a unit interval $i_x$ on the line $h_j$; we denote this set of unit intervals as $I(h_j)$.
First, we sort the intervals in $I(h_j)$ from left to right. 
Next we define $k+1$ \emph{groups} with respect to $h_j$ that are formed by deleting those squares and their corresponding intervals from $S(h_j)$ and $I(h_j)$, respectively, that intersect every 
$k+1$-th grid point on $h_j$, starting from some $g_{j,\alpha}$ with $\alpha \in [k+1]$. %
Now consider the $k$ consecutive grid points on $h_j$ between two deleted grid points in one such group, say, $\{g_{j,\ell},\ldots,g_{j,\ell+k-1}\}$ for some $\ell \in [\sigma]$. 
Let $I^k_{\ell}(h_j) \subseteq I(h_j)$ be the set of unit intervals intersecting the $k$ grid points $g_{j,\ell}$ to $g_{j,\ell+k-1}$. 
We refer to them as \emph{subgroups}.
See Figure~\ref{fig:k-approx-partition} for an illustration. Observe that the maximum size of an independent set of each subgroup is at most $k$, since the width of each subgroup is strictly less than $k+1$ and each interval has unit length.

\begin{figure}[t]
\centering
\includegraphics{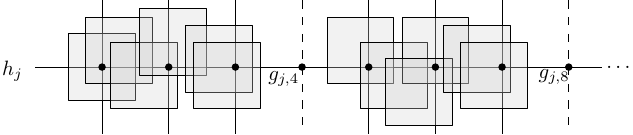}
\caption{Illustration of a group on line $h_j$ for $k=3$ with the two  subgroups  $I^3_{1}(h_j)$ and  $I^3_{5}(h_j)$.}
\label{fig:k-approx-partition}
\end{figure}

We compute $M_1$ for $G_1$ as follows. %
For each stabbing line $h_j\in H$, we form the $k+1$ different {groups} of $I(h_j)$. 
For each {group}, a \maxset is computed optimally and separately inside each {subgroup}. 
Since any two subgroups are horizontally separated and thus independent, we can then take the union of the independent sets of the subgroups to get an independent set for the entire {group}. 
This is done with the linear-time greedy algorithm to compute maximum independent sets for interval graphs \cite{DBLP:journals/networks/GuptaLL82}. %
Let $\{M^1_{h_j},\ldots,M^{k+1}_{h_j}\}$ be $k+1$ maximum independent sets for the $k+1$ different {groups} and let $M(h_j)=\argmax\{|M^1_{h_j}|, |M^2_{h_j}|, \dots, |M^{k+1}_{h_j}|\}$ be one with maximum size. We store its cardinality as $c(h_j)=\max\{|M^i_{h_{j}}| \mid i \in [k+1]\}$. 
Next, we compute an independent set for $E_H$, denoted by $M(E_H)$, by composing it from the best solutions $M(h_j)$ from the even stabbing lines, i.e., $M(E_H) = \bigcup_{j=1}^{\floor{ %
\sigma/2}} M(h_{2j})$ and its cardinality $|M(E_H)| = \sum_{j=1}^{\floor{ %
\sigma/2}} c(h_{2j})$. 
Similarly, we compute an independent set for $O_H$ %
as  $M(O_H) = \bigcup_{j=1}^{\floor{ %
\sigma/2}} M(h_{2j-1})$ and its cardinality $|M(O_H)| = \sum_{j=1}^{\floor{ %
\sigma/2}} c(h_{2j-1})$. 
Finally, we return $M_1=\argmax\{|M(E_H)|, |M(O_H)|\}$ as the solution for $G_1$.

\begin{lemma}	
\label{valid-k-g1}
The independent set $M_1$ of \hspace{0.05cm}$G_1=(\mathcal{S}_1,E_1)$ can be computed in $O(n \log n + kn)$ time and \hspace{0.03cm}$|M_1|\ge {|OPT_1|}/{2 (1+\frac{1}{k})}$.
\end{lemma}
\begin{proof}
	Let us begin with the analysis for one horizontal line, say, $h_j$. 
	The objective is to show that for $h_j$, the size of our solution is least the optimum solution size for $h_j$ divided by $(1+\frac{1}{k})$. 
	Recall that a {group} of $\mathcal S(h_j)$ and $I(h_j)$ is formed by deleting the squares and their corresponding intervals from $S(h_j)$ and $I(h_j)$, respectively, which intersect every $k+1$-th grid point on $h_j$, starting at some index $\alpha \in [k+1]$. 
	Now consider a hypothetical \maxset $OPT(h_j)$ on $h_j$. 
	By the pigeonhole principle, for at least one of the $k+1$ groups of $\mathcal S(h_j)$ we deleted at most $|OPT(h_j)| / (k+1)$ squares from $OPT(h_j)$. 
	Assume that this group corresponds to the independent set $M_{h_j}^\ell$ for some $\ell \in [k+1]$, which is maximum within each subgroup. 
	Then we know that $|M(h_j)| \ge |M_{h_j}^\ell| \ge |OPT(h_j)| - |OPT(h_j)| / (k+1) = |OPT(h_j)| / (1 + \frac{1}{k})$.
	Since this is true for each individual stabbing line $h_j$ and since any two lines in $E_H$ (or $O_H$) are independent, this implies that $|M(E_H)| \ge |OPT(E_H)| / (1 + \frac{1}{k})$ and $|M(O_H)| \ge |OPT(O_H)| / (1 + \frac{1}{k})$.
	Again by pigeonhole principle, if we choose $M_1$ as the larger of the two independent sets $M(E_H)$ and $M(O_H)$, then we lose by at most another factor of $2$, i.e., $|M_1| \ge |OPT_1| / 2 (1 + \frac{1}{k})$.

	The algorithm requires $O(n \log n)$ time to sort all intervals and then computes \maxset for the different subgroups with the linear time greedy algorithm. 
	Since each square belongs to at most $k$ different subgroups, this takes $O(kn)$ time in total. %
\end{proof}

Next, we describe a pre-processing step, which is required for the  dynamic updates.  

\medskip

\noindent\textsc{\underline{Pre-Processing:}} For each horizontal line $h_j\in H$, consider a {group}. 
For each {subgroup} $I^k_{\ell}(h_j)$ (for some $\ell\in [k+1]$), we construct a balanced binary tree $T(I^k_{\ell}(h_j))$ storing the intervals of $I^k_{\ell}(h_j)$ in left-to-right order (indexed by their left endpoints) in the leaves. 
This process is done for each group of every horizontal line $h_j\in H$. 
This preprocessing step takes $O(kn \log n)$ time.

\medskip

When we perform the update step from $G_i=(\mathcal{S}_i,E_i)$ to 
$G_{i+1}=(\mathcal{S}_{i+1},E_{i+1})$, either a square is inserted into $\mathcal{S}_i$ or deleted from $\mathcal{S}_i$. 
Let $s_x$ and $i_x$ be this square and its corresponding interval. %
Let $g_{u,v}$ (for some $u,v\in [\sigma]$) be the grid point that intersects $s_x$. 

\medskip

\noindent\textsc{\underline{Insertion/Deletion}:} 
The insertion or deletion of $i_x$ affects all but one of the {groups} on line $h_u$.
We describe the procedure for one such {group} on $h_u$; it is then repeated for the other groups. 
In each {group}, $i_x$ appears in exactly one {subgroup} and the other subgroups remain unaffected.
This subgroup, say $I^k_{\ell}(h_u)$, is determined by the index $v$ of the grid point $g_{u,v}$ intersecting $i_x$. 
For each affected subgroups, we do the following update.
First, we update the search tree $T(I^k_{\ell}(h_h))$ of $I^k_{\ell}(h_u)$ by inserting or deleting $i_x$, which can be done in $O(\log n)$ time.
Then, we recompute a \maxset of the subgroup with the greedy algorithm.
Since the intervals of $I^k_{\ell}(h_u)$ are sorted, we could locate the left-most interval which is to the right of all the chosen intervals in $O(\log n)$ time. 
Since a maximum independent set in each subgroup contains at most $k$ intervals, the re-computation takes $O(k \log n)$ time in each affected subgroup.

For all groups affected by the insertion or the deletion of $i_x$ we update the corresponding independent sets $M^p_{h_u}$ for $p \in [k+1]$, whenever some updates of selected intervals were necessary.
Then we select the largest independent set of all $k+1$ groups as $M(h_j)$ and update its new cardinality in %
$c(h_j)$.
Finally, we update the independent sets $M(E_H)$ and $M(O_H)$ and their cardinalities and return 
$M_{i+1}=\argmax\{|M(E_H)|, |M(O_H)|\}$ as the solution for~$G_{i+1}$. %

\medskip

Since the intervals computed by the left-to-right greedy algorithm are precisely those intervals that our update procedure selects, we get the following Lemma~\ref{valid-k-gi}.

\begin{lemma}		\label{valid-k-gi}
The set $M_i$ is an independent set of $G_i=(\mathcal{S}_i,E_i)$
for each  $i\in [N]$ and $|M_i|\ge |OPT_i|/ 2(1+\frac{1}{k})$.
\end{lemma}
\begin{proof}
	The fact that $M_i$ is an independent set follows directly from the construction.
    Let $s_x$ be the square added or deleted  and let $h_{u}$ be the grid horizontal line stabbed by $s_x$. 
	In fact, our update algorithm constructs the same set of %
	independent intervals as the one obtained by running from scratch the greedy \maxset algorithm on the set $I^k_{\ell}(h_u)$.
	The remaining arguments for the claimed approximation ratio of $M_{i+1}$ are exactly the same as in the proof of Lemma~\ref{valid-k-g1}.
\end{proof}

\subparagraph*{\textbf{Running Time.}} At every step, we perform either an \textsc{insertion} or a \textsc{deletion} operation.
Recall from the description of these two operations that an update affects a single stabbing line, say $h_u$, for which we have defined $k+1$ groups. 
Of those groups, $k$ are affected by the update, but only inside a single subgroup. 
Updating a subgroup can trigger up to $k$ %
selection updates, 
each taking $O(\log n)$ time. 
In total this yields an update time of $O(k^2 \log n)$.

This approach requires $O(\sigma^2)$ space to maintain the grid. Note that each rectangle can be in the stored solutions of at most $k$ subgroups, thus at most $O(kn)$ storage is used.
With Lemma~\ref{valid-k-gi} and the above update time discussion we obtain:%

\begin{theorem}
We can maintain a $2(1+\frac{1}{k})$-approximate maximum independent set in a dynamic unit square intersection graph, deterministically, in $O(k^2\log n)$ update time and $O(\sigma^2 + kn)$ storage. 
\end{theorem}

\subsection[]{2-Approximation Algorithm with $O(\log^2 n + \omega\log n)$ Update Time}\label{sub:2apx}

We finally design a $2$-approximation algorithm for the \maxset problem on dynamic axis-aligned unit height, but arbitrary width rectangles. Note that the coordinates of the input rectangles might not be integers. %
Let $\mathcal B$ be the bounding box of the dynamic set of rectangles $\widetilde{\mathcal R} = \bigcup_{i \in [N]} \mathcal R_i$.
We begin by dividing $\mathcal B$ into horizontal strips of height $1$ defined by the set $H=\{h_1,\ldots,h_{\sigma}\}$ of $\sigma = O(n)$ horizontal lines. 
We assume, w.l.o.g., that every rectangle in $\widetilde{\mathcal{R}}$ is stabbed by exactly one line in $H$. 
For a set of rectangles $\mathcal R$, we denote the subset stabbed by a line $h_j$ as $\mathcal R(h_j) \subseteq \mathcal R$.

We first describe how to obtain an independent set $M_1$ for the initial graph $G_1=(\mathcal{R}_1,E_1)$ such that $|M_1|\ge |OPT_1|/2$ by using the following algorithm of Agarwal et al.~\cite{agarwal1998label}. 
For each horizontal line $h_j\in H$, we compute a maximum independent set for $\mathcal{R}_1 (h_j)$.
The set $\mathcal{R}_i(h_j)$ (for any $i\in [N]$ and $j \in [\sigma]$) can again be seen as an interval graph. 
For a set of $n$ intervals, a \maxset can be computed by a left-to-right greedy algorithm visiting the intervals in the order of their right endpoints in $O(n \log n)$ time. 
So for each horizontal line $h_j\in H$, let $M(h_j)$ be a \maxset of $\mathcal R_1(h_j)$, and let $c(h_j)=|M(h_j)|$. 
Then we construct the independent set $M(E_H) = \bigcup_{j=1}^{\floor{ %
\sigma/2}} (M(h_{2j}))$ for $E_H$. %
Similarly, we construct the independent set 
$M(O_H) = \bigcup_{j=1}^{\floor{ %
\sigma/2}} (M(h_{2j-1})$ for $O_H$. %
We return $M_1=\argmax\{|M(E_H)|, |M(O_H)|\}$ as the independent set for $G_1=(\mathcal{R}_1,E_1)$.
See Figure~\ref{fig:2-approx} %
for an illustration. 

\begin{lemma}[Theorem 2,\cite{agarwal1998label}]
\label{agarwallabel}
The set $M_1$ is an independent set of $G_1=(\mathcal{R}_1,E_1)$ with $|M_1|\ge |OPT_1|/2$ and can be computed in $O(n \log n)$ time.
\end{lemma}

We describe the following pre-processing step to initialize in $O(n \log n)$ time the data structures that are  required for the subsequent dynamic updates.

\medskip

\noindent\textsc{\underline{Pre-Processing:}} 
Consider a stabbing line $h_j$ and the set of rectangles $\mathcal R_i(h_j)$ stabbed on $h_j$ for some $i \in [N]$.
We denote the corresponding set of intervals as $I(h_j)$. 

We build a balanced binary search tree $T_l(I(h_j))$, storing the intervals in $I(h_j)$ in left-to-right order based on their left endpoints.

This is called the \emph{left tree} of $I(h_j)$.
For each internal node in the left tree, 
we associate it with an augmented balanced binary search tree storing the intervals in its subtree based on their right endpoints.
Thus, we could get the interval with leftmost right endpoint in each subtree in constant time.

Such range-tree like data structure can be constructed in $O(n \log n)$ and can be dynamically maintained in $O(\log^2 n)$ time~\cite{DBLP:journals/jacm/WillardL85}.
Additionally, we compute a \maxset of $I(h_j)$ and store it in left-to-right order in a balanced binary search tree $T_s(I(h_j))$, denoted by the \emph{solution tree} of $h_j$.
Let $\omega_j$ be the cardinality of a maximum independent set of $I(h_j)$ for $j \in [\sigma]$, and let $\omega = \max_j \omega_j$ for $j \in [\sigma]$ be the maximum of these cardinalities over all stabbing lines.

\medskip

When we move from $G_i$ to $G_{i+1}$ (for some $1\le i < N$), either we insert a new rectangle into $\mathcal{R}_i$ or delete one rectangle from $\mathcal{R}_i$. 
Let $r_x$ be the rectangle that is inserted or deleted, let $i_x$ be its corresponding interval,  and let $h_j$ (for some $j\in [\sigma]$) be the horizontal line that intersects $r_x$.
By maintaining the maximum independent set for $I(h_j)$, and then reevaluating the solution set, a $2$-approximate \maxset can be maintained. 
Note that this recomputation costs $O(\omega \log n + \log^2 n)$, including updating the left tree of $h_j$ and recomputing the maximum independent set for $h_j$ by the greedy approach described in the pre-processing phase. 
In what follows, we describe how to maintain the maximum independent set for $h_j$ dynamically in $O(\omega \log n + \log^2 n)$ time.

Note that our update operation is faster in practice while having the same asymptotic worst-case update time $O(n\log n)$ as recomputing the maximum independent set of $h_j$.

\begin{figure}[t]
\centering
\includegraphics{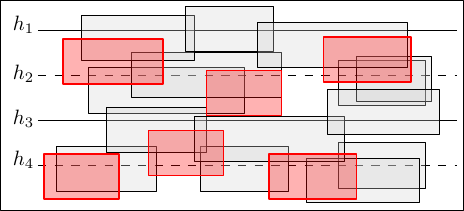}
\caption{Example instance with four horizontal lines. Red rectangles represent the computed $2$-approximate solution, which here is $M(E(H))$.}
\label{fig:2-approx}
\end{figure}

\medskip

Given an interval $i$, we denote the left endpoint and right endpoint of $i$ as $l(i)$ and $r(i)$, respectively.

\noindent\textsc{\underline{Insertion/Deletion}:} 
Because the greedy algorithm for constructing the \maxset visits the intervals in left-to-right order based on the right endpoints, it would make the same decisions for all the intervals with their right endpoint before the right endpoint of $i_x$.
Let $i_y$ be the right-most interval in the current solution such that the right endpoint of $i_y$ is before the right endpoint of $i_x$,i.e., $r(i_y) < r(i_x)$.
We may assume such interval $i_y$ always exist by adding a dummy interval $(0,\epsilon)$ in each solution tree $T_s(I(h_j))$ for an arbitrary small value $\epsilon$.
We could find $i_y$ for $i_x$ by querying the %
solution tree $T_s(I(h_j))$ in $O(\log \omega_j)$ time.
Now we need to identify the next %
selected interval right of $i_y$ that would have been found by the greedy algorithm.
We use the left tree $T_l(I(h_j))$ to search in $O(\log n)$ time for the interval $i'_z$ with leftmost right endpoint, whose left endpoint is right of the right endpoint $r(i_y)$ of $i_y$.
More precisely, we search for $r(i_y)$ in $T_l(I(h_j))$ and whenever the search path branches into the left subtree, we compare whether the leftmost right endpoint stored in the root of the right subtree is left of the right endpoint of the current candidate interval.
If so, we use this interval as the new candidate interval.
Once a leaf is reached, the leftmost found candidate interval is the desired interval $i'_z$.
This interval $i'_z$ is is precisely the first interval considered by the greedy algorithm after $I_y$ and thus must be the next selected interval.
We repeat the update process for $i'_z$ as if it would have been the newly inserted interval until either $i'_z$ is also selected in the previous solution or we reach the end of $I(h_j)$.
We now reevaluate the new \maxset $M(h_j)$ and its cardinality, which possibly affects $M(E_H)$ or $M(O_H)$. 
We obtain the new independent set $M_{i+1}=\argmax\{|M(E_H)|, |M(O_H)|\}$ for $G_{i+1}=(\mathcal{R}_{i+1}, E_{i+1})$.

\subparagraph*{\textbf{Running Time.}} 
An update in the left tree (interval insertion/deletion) costs $O(\log^2 n)$ time. 

To update the solution of $h_j$, we perform at most $\omega_j$ searches in $T_l(I(h_j))$, each of which takes $O(\log n)$ time. %
Finally, we need to delete $O(\omega_j)$ old %
selected intervals from and insert $O(\omega_j)$ new %
selected intervals into the %
solution tree $T_s(I(h_j))$, each of which takes $O(\log \omega_j)$ time.
We now re-evaluate the new \maxset $M(h_j)$ and its cardinality $c(h_j)$, which possibly affects $M(E_H)$ or $M(O_H)$. 
We obtain the new independent set $M_{i+1}=\argmax\{|M(E_H)|, |M(O_H)|\}$ for $G_{i+1}=(\mathcal{R}_{i+1}, E_{i+1})$.
Overall, the total update time is $O(\omega_j \log n + \log^2 n)$.
\medskip%

\begin{lemma}
	\label{lem:2apx}
	The set $M_i$ is an independent set of $G_i=(\mathcal{R}_i,E_i)$ for each $i\in [N]$ and 
	$|M_i|\ge |OPT_i|/2$.
	\end{lemma}

\begin{proof}
	We prove the lemma by induction.
	From Lemma~\ref{agarwallabel} we know that $M_1$ satisfies the claim, and in particular each set $M(h)$ for $h \in H$ is a \maxset of the interval set $I(h)$.
	So let us consider the set $M_i$ for $i \ge 2$ and assume that $M_{i-1}$ satisfies the claim  by the induction hypothesis. 
	Let $r_x$ and $i_x$ be the updated rectangle and its interval, and assume that it belongs to the stabbing line $h_j$.
	Then we know that for each $h_k \in H$ with $k \ne j$ the set $M(h_k)$ is not affected by the update to $r_x$ and thus is a \maxset by the induction hypothesis. 
	It remains to show that the update operations described above restore a \maxset $M(h_j)$ for the set $I(h_j)$.
	But in fact the updates are designed in such a way that the resulting set of %
	selected intervals is identical to the set of %
	intervals that would be found by the greedy \maxset algorithm for $I(h_j)$.
	Therefore $M(h_j)$ is a \maxset for $I(h_j)$ and by the pigeonhole principle $|M_i| \ge |OPT_i|/2$.
\end{proof}

\subparagraph*{\textbf{Running Time.}} 
Each update of a rectangle $r_x$ (and its interval $i_x$) triggers either an \textsc{Insertion} or a \textsc{Deletion} operation on the unique stabbing line of $r_x$.
As we have argued in the description of these two update operations, the insertion or deletion of $i_x$ requires one $O(\log^2 n)$-time update in  the left tree data structure.
If $i_x$ is a %
selected independent interval, the update further triggers a sequence of at most $\omega_j$ selection updates, each of which requires $O(\log n)$ time.
Hence the update time is bounded by $O(\log^2 n + \omega_j \log n) = O(\log^2 n + \omega \log n)$. Recall that $\omega_j$ and $\omega$ are output-sensitive parameters describing the maximum size of an independent set of $I(h)$ for a specific stabbing line $h=h_j$ or any stabbing line $h$.

In this approach, we have to maintain $O(\sigma)$ stabbing lines. 
For each stabbing line $h_l$, let $n_l$ be the number of rectangles stabbed by $h_l$.  
For each stabbing $h_l$, we maintain a dynamic left tree for the  $n_l$ corresponding intervals, which requires $O(n_l \log n_l)$ space \cite{DBLP:journals/jacm/WillardL85}. 
Overall, the total space required is $O(n\log\, n + \sigma)$.

\begin{theorem}
We can maintain a $2$-approximate maximum independent set in a dynamic unit-height arbitrary-width rectangle intersection graph, deterministically, in $O(\log^2 n + \omega\log n)$ time and in $O(n\log\, n + \sigma)$ space, where $\omega$ is the maximum size of an independent set of the unit-height rectangles stabbed by any horizontal line.
\end{theorem}

\begin{remark}
We note that 
Gavruskin et al. \cite{gavruskin2015dynamic} gave a dynamic algorithm for maintaining a \maxset on \textit{proper} interval graphs.
Their algorithm runs in amortized time 
$O(\log^2 n)$ for insertion and deletion, and $O(\log n)$ for element-wise decision queries. The complexity to report a \maxset $J$ is $\Theta(|J|)$. Whether the same result holds for general interval graphs was posed as an open problem~\cite{gavruskin2015dynamic}. 
Our algorithm in fact solves the \maxset problem on arbitrary dynamic interval graphs, which is of independent interest.
Moreover, it explicitly maintains an exact \maxset at every step. Recently, Bhore et al.~\cite{DBLP:journals/corr/abs-2007-08643} showed that for intervals a $(1+\epsilon)$-approximate maximum independent set can be
maintained with logarithmic worst-case update time, where
$\epsilon>0$ is any positive constant.
\end{remark}

\section{Experiments}\label{exp}
We implemented all our \maxset approximation algorithms presented in Sections~\ref{mis-log} and~\ref{MIS-const} in order to empirically evaluate their trade-offs in terms of \emph{solution quality}, i.e., the cardinality of the computed independent sets, and \emph{update time} measured on a set of suitable synthetic and real-world map-labeling benchmark instances of two types of dynamic rectangle sets: (1) unit squares, (2) rectangles of uniform height and bounded width-height integer aspect ratio; see Figure~\ref{fig:unit-height rectangle}. 
We believe these two models are representative models in map labeling applications.
The goal is to identify those algorithms that best balance the two performance criteria.

\begin{figure}[t]
\centering
\includegraphics{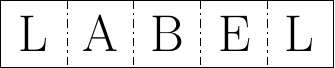}
\caption{Example instance of unit height rectangle with width-height aspect ratio $5$.}
\label{fig:unit-height rectangle}
\end{figure}

Moreover, for smaller benchmark instances with up to 2\,000 squares, we compute exact \maxset solutions using a \MAXSAT model by Klute et al.~\cite{kllns-esl-19} that we solve with %
MaxHS~3.0 (see \url{www.maxhs.org}). 
These exact solutions allow us to evaluate the %
optimality gaps of the different %
algorithms %
in light of their worst-case approximation guarantees.
Finally, we investigate the speed-ups gained by using our dynamic update algorithms compared to the baseline of recomputing new solutions from scratch with their respective static algorithm after each update.

\subsection{Experimental Setup}

\subparagraph*{\textbf{Implemented Algorithms.}}
We have implemented %
the following six algorithms (and their greedy augmentation variants) in C++.
We implemented two \indset algorithm for unit squares, \misg\ and \misr\, where a maximal independent set is maintained dynamically.
In both of these approaches, we use a dynamic orthogonal range searching data structure to check the intersections of squares. The main difference is that in \misg\, we maintain the geometric intersection graph explicitly and thus the update time is affected by the degree of the 
corresponding vertex. 
Note that in our implementation, the \indset algorithms \misg\ and \misr\ compute the same maximal independent set at each round and provide a $4$-approximation.
Moreover, we extended the approach \misg\ for unit-height rectangles.

\begin{description}

\item[\misg for unit squares] A naive graph-based dynamic \indset algorithm for unit squares, explicitly maintaining the square intersection graph and a MIS~\cite[Sec.~3]{DBLP:conf/stoc/AssadiOSS18}. In order to evaluate and compare the performance of our algorithm \misr\ (Section~\ref{mis-log})
for the  \indset problem, we have implemented this alternative %
dynamic algorithm as the baseline approach. 
This algorithm maintains the current instance
in a dynamic geometric data structure and maintains the square intersection graph explicitly.
We use standard adjacency lists to represent the intersection graph, implemented as unordered sets in C++.

In the initialization step, we store the center points of all unit squares in the dynamic point range query structure\footnote{2D Range and Neighbor Search in CGAL
 \url{https://doc.cgal.org/latest/Point_set_2/index.html}} implemented in CGAL (version 5.2.1).
 By performing neighbor searches in this range search structure, we build the initial geometric intersection graph. 
 Precisely, for each square $s_x$, we query the orthogonal range tree with the range $s^2_x$, where $s^2_x$ is the square of side length $2$ concentric with $s_x$. This range contains the center points of all squares that intersect $s_x$.
Now, to obtain a \indset at the first step, we add the first (unmarked) vertex $v$ to the solution and mark $N(v)$ in the corresponding intersection graph.
This process is repeated iteratively until there is no unmarked vertex left in the intersection graph. %
Clearly, by following this greedy method, we obtain a \indset. 

Moreover, for each vertex $v$, we maintain an augmenting \emph{counter} that stores the number of vertices from its neighborhood $N(v)$ that are contained in the current \indset. 
Note that in our implementation, the  approach greedily checks the vertices in their ordering as given in the input file.

This approach handles the updates in a straightforward manner. 
In order to add a new square, we first insert its center point into the dynamic orthogonal range query data structure and add a new vertex in the intersection graph.
When a new vertex is inserted, its corresponding square may introduce new intersections. 
Therefore, when adding a vertex, we also determine the edges that are required to be added to the intersection graph. %
Notice that unlike the canonical vertex update operation defined in the literature, where the adjacencies of the new vertex are part of the dynamic update, here, we actually need to figure out the neighborhood of a vertex.
Let $s_x$ be the square to add and let $v_x$ be its corresponding vertex, which is added to the intersection graph. 
In order to find  all squares that overlap $s_x$, we query the orthogonal range tree with  $s^2_x$. This output-sensitive operation takes $O(\log\,n +\deg(v))$ time, where $\deg(v)$ is the size of neighborhood of this newly added vertex $v$. 
If the newly inserted square %
has no intersection with any square from the current solution, then we simply add its vertex to the solution; otherwise, we ignore it. Finally, we update the counters.
If a vertex is deleted, we update the orthogonal range query structure and the intersection graph by deleting its corresponding center point and vertex, respectively. %
If the deleted vertex was in the  solution, then we decrease the counters of its neighbors by $1$. Once the counter of a vertex is updated to $0$, we add this vertex into the solution. 
Both the insertion (after computing $N(v)$) and deletion operation for a vertex $v$ take $O(\deg(v))$ %
time each to update the intersection graph and the \indset solution.
By maintaining the conflict graph in an  adjacency list, the space requirement of this implemented approach is $O(n^2)$.

\item[MIS-graph for unit-height rectangles] 
We extend the approach \misg\ for axis-parallel rectangles of unit height and with bounded integer width-height aspect ratio.
Let $r$ be an axis-parallel rectangle with unit height and aspect ratio $w$ for an integer $w$.
The rectangle $r$ can be partitioned into $w$ unit squares. 
Each corner of these partitioning unit squares is denoted as a \emph{witness point} of $r$ in the following.
In this  extended \misg\ approach, instead of maintaining the center points of squares in the orthogonal range query data structure as in \misg\, all of the witness points of all rectangles are stored in  the range query data structure. 
Furthermore, we use a hash table, which assigns each witness point to its corresponding rectangle.
We assume a general position property of the dynamic set of rectangles in the input: each pair of rectangles intersects at most twice at their boundaries. This property is denoted as corner intersection property in the following\footnote{In computational geometry, a set of geometric objects with such property is denoted as a family of pseudo-disks}. 
Thus, given a rectangle $r$, we could find all rectangles that overlap with $r$ by making a range query of $r$ in the dynamic orthogonal range tree.
Note that in our input instances, the width height aspect ratio is from $\{1,\dots b\}$ for a positive constant integer $b$. 
Thus, the initialization takes $(bn\log\,bn)$ time.
Let $r_x$ be the rectangle to add or delete, and $v_x$ be its corresponding vertex in the conflict graph.  
Both the insertion and deletion (after updating the range tree and conflict graph) take $O(\deg(v_x))$ where $\deg(v_x)$ is the size of the neighborhood of the vertex $v_x$ in the conflict graph.

\item[MIS-ORS] The dynamic \indset algorithm based on orthogonal range searching (Section~\ref{mis-log}); this algorithm provides a 4-approximation. In the implementation we used the dynamic orthogonal range searching data structure implemented in CGAL (version 5.2.1), which is based on a dynamic Delaunay triangulation~\cite[Chapter 10.6]{mn-lpcgc-99}. 
More precisely, given a  rectilinear area $R$, a range query with $R$ is implemented in CGAL by making a circular range query with its circumscribed circle and then check if the reported points are in $R$.
That means, finding and reporting one point of a rectilinear area take in the worst case $O(n)$.
Hence, this implementation does not provide the polylogarithmic worst-case update time of Theorem~\ref{thm:mis-rt}.
However, we did not observe such behaviour in most of our experiments.
Note that the initial solution is computed greedily in the ordering of the instance file and in each update round the maximal independent set is maintained dynamically.
Overall, the solution computed in each round by this approach is identical as the solution computed by \misg.

We implement \indset algorithms \grid (Section~\ref{4-approx}), \gridk{k} (Section~\ref{sub:2kapx}), and \linea (Section~\ref{sub:2apx}) and their greedy augmentation variants in C++. 
Since all these algorithms are based on partitioning the set of squares and considering only sufficiently segregated subsets, they produce a lot of white space in practice.

For instance, they ignore the squares stabbed by either all the even or all the odd stabbing lines completely in order to create isolated subinstances. 
In practice, it is therefore interesting to augment the computed approximate \maxset by greedily adding independent, but initially discarded squares.
We have also implemented the greedy variants of these algorithms, which are denoted as \ggrid, \ggridk{k}, and \glinea.

\item[\grid] 
Recall that in the grid-based $4$-approximation algorithm for unit squares (Section~\ref{4-approx}), we either omit all squares intersecting even horizontal lines or all squares intersecting odd horizontal lines. 
Then, for each horizontal line $l_i$, we maintain a Boolean value based on whether we omit all rectangles intersecting the odd vertical lines or all rectangles intersecting the even vertical lines.

In our implementation, to obtain the initial solution, we iterate over all grid points that are not omitted, and collect the first square stored in their square lists. Let this be our initial solution. 

\item[g-grid]The grid-based approach with greedy augmentation. 
We describe the greedy augmenting procedure. 
Recall that for each horizontal line $h_j$, we maintain two candidate independent sets $M^1_{h_j}$ and $M^2_{h_j}$ of rectangles of odd grid points and rectangles of even grid points of $h_j$, respectively.
Given a horizontal grid line $h_j$, we check all rectangles intersecting even (resp. odd) grid points on $h_j$ and compute an augmentation for $M^1_{h_j}$ (resp. $M^2_{h_j}$), which is denoted as a \textit{line augmentation set}.
For three consecutive horizontal grid lines $h_j$, $h_{j+1}$ and $h_{j+2}$ in the grid partitioning, we consider the four possible combinations of one candidate set of $h_j$ and one candidate set of $h_{j+2}$ with their corresponding line augmentation sets. 
For each combination, we compute a greedy augmentation of rectangles stabbed by line $h_{j+1}$, denoted as a \textit{combination augmentation set}. 
To achieve a constant-time update, each of these computed augmentation sets is stored in a vector of size $\sigma$ such that the rectangle intersecting the $k$-th vertical grid line is stored as the $k$-th element of the vector. Overall, for each horizontal grid line $h_j$, we maintain six greedy augmentation sets: four combination augmentation sets for the four possible combinations of $h_{j-1}$ and $h_{j+1}$ as well as two line augmentation set for $M^1_{h_j}$ and $M^2_{h_j}$, respectively.
Thus, the initial solution obtained by this approach contains three parts: the initial solution obtained by the \grid\ approach,  the line augmentation sets for the chosen candidate sets of the odd/even (based on the choice by \grid\ approach) horizontal lines, and the combination augmentation sets  from the omitted horizontal lines based on the choices of its two neighboring lines; see Figure~\ref{fig:greedy-augmentation}.

\begin{figure}[t]
\centering
\includegraphics{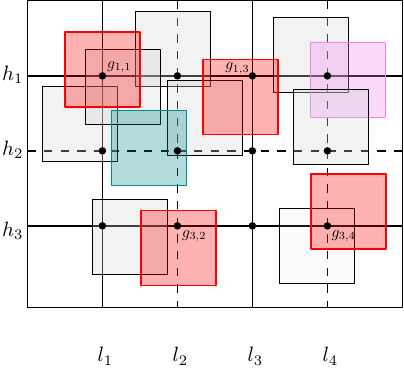}
\caption{Example instance within a $4 \times 4$ grid. 
The solution obtained by the g-\grid\ approach is the union of the solution $M^1(h_1) \cup M^2(h_3)$ (red) by the \grid\ approach, the line augmentation set of $M^1(h_1)$ (violet), the line augmentation set of $M^3(h_3)$ (= $\emptyset$) and the combination augmentation set of  $M^1(h_1) \cup M^2(h_3)$ (green).}
\label{fig:greedy-augmentation}
\end{figure}

Whenever a square $s$ is inserted on the horizontal grid line $h_r$, the two candidate sets $M^1_{h_r}$ and  $M^2_{h_r}$ of $h_r$ might be updated by adding $s$.
Then we update all greedy augmentation sets involving $h_r$
by removing the rectangles intersecting $s$. 
Note that we only need to check the rectangles on a neighboring grid point of the grid point of $s$. 
Thus, to update one greedy augmentation set takes constant time.
When removing a square $s$ from $h_r$, we remove $s$ from all greedy augmentation sets of $h_r$. Note that this greedy augmentation procedure does not affect the approximation 
bound (i.e., $4$-approximation) since the solution of the \grid\ approach is unaffected. 
Moreover, this procedure takes constant update time and $O(\sigma^2)$ space, which is the same as for the \grid\ approach.

\item[\gridk{k}] The shifting-based $2(1 + \frac{1}{k})$-approximation algorithm (Section~\ref{sub:2kapx}). In the experiments we use $k=2$ (i.e., a $3$-approximation) and $k=4$ (i.e., a $2.5$-approximation).

\item[\ggridk{k}]
We describe the greedy augmented version of the \gridk{k}\ approach and
generalize the idea of the \ggrid\ approach.
Recall that for each horizontal grid line $h_j$, 
there is one candidate independent set for each of $k+1$ group.
In the initialization phase, for each candidate set of $h_j$, we compute its augmented set consisting of rectangles of $h_j$, which is denoted as the  \textit{line augmentation set} of the candidate set.
For every three consecutive horizontal grid lines $h_j, h_{j+1}, h_{j+2}$, we consider $k^2$ unions of one candidate set of $h_j$ and one candidate set of $h_{j+1}$ with their corresponding line augmentation sets and compute for each union an augmented set of rectangles on $h_{j+1}$ greedily.
For each computed augmented set, we store it in a vector of size $\sigma$ such that the index of each stored rectangle is identical to the index of its vertical grid point.
Thus, the initial solution obtained by this approach contains three parts: the initial solution obtained by the \gridk{k}\ approach,  greedy augmented rectangles for the chosen candidate sets of the odd/even (based on the choice by \gridk{k}\ approach) horizontal lines, and the greedy augmented rectangles from the omitted horizontal lines based on the choices of its two neighboring lines.
Whenever a square is inserted or deleted in an update phase, the involved $O(k^2)$ greedy augmentation sets need to be updated accordingly.
Overall, this greedy augmented version of the approach \gridk{k} retains the same approximation ratio $2(1+\frac{1}{k})$, the same update time $O(k^2\log n)$ and the same space requirement $O(\sigma^2+kn)$ as the \gridk{k} approach.

\item[\linea] The stabbing-line based $2$-approximation algorithm (Section~\ref{sub:2apx}).

\item[\glinea]
We describe the greedy augmentation procedure of the \linea\ approach. Recall that \linea\ maintains a candidate set, which is a maximum independent set, for each stabbing line and we omit either all rectangles stabbed by even lines or by odd lines.
In the initialization phase of this greedy-augmented approach, we first compute the candidate set for each horizontal line as in the \linea\ approach. 
For every three consecutive horizontal grid lines $h_j, h_{j+1}, h_{j+2}$, we consider the union of the candidate set of $h_j$ and the candidate set of $h_{j+1}$ and compute for this union an augmentation set of rectangles on $h_{j+1}$ greedily.
Each greedy augmentation set is sorted from left to right and stored in an ordered set. 
The solution obtained by this approach consists of the solution obtained by \linea\  and the greedy augmented sets of the omitted lines.
When a square is inserted into or removed from a stabbing line $h_j$, we first update the candidate set of $h_j$ and then update the greedy augmentation sets of the stabbing line above and the stabbing line below $h_j$, if needed.
More precisely, when the candidate set of $h_j$ is updated, we mark the leftmost and the rightmost rectangles $r_s, r_t$ which are the newly added squares in this candidate set. 
Then, when we update a greedy augmentation set, we first find the leftmost rectangle in this greedy augmentation set which is to the left of the left-endpoint of $r_s$ and recompute the greedy augmentation set from this rectangle. 
This recomputing procedure terminates when the chosen rectangle is in the greedy augmentation set from previous round and is right to the right end-point of $r_t$.
That means, to update the greedy augmentation set of one stabbing line $h_j$, it takes in the worst-case $O(n_j)$ time where $n_j$ is the number of rectangles stabbed by $h_j$.
Thus, this greedy augmentation takes worst case $O(n)$ time for an update.
However, we note that such worst-case behaviour was not observed in the experiments.
\end{description}

\subparagraph*{\textbf{System Specifications.}} 
The experiments were run on a server equipped with two Intel Xeon E5-2640 v4 processors (2.4 GHz 10-core) and 160GB RAM. The machine ran the 64-bit version of Ubuntu Bionic (18.04.2 LTS). The code was compiled using g++ 7.5.0 with optimization level O3.  %

\subparagraph*{\textbf{Benchmark Data of Unit Squares.}}
We created three types of benchmark instances. The two synthetic data sets consist of $n$ $30\times30$-pixel squares %
placed inside a bounding rectangle $\mathcal B$ of size $1\,080 \times 720$ pixels, which also creates different densities. The real-world instances use the same square size, but geographic feature distributions. %
For the  updates we consider three models: \emph{insertion-only}, \emph{deletion-only}, and \emph{mixed}, where the latter selects insertion or deletion uniformly at random.
The new squares to insert are generated uniformly.

\begin{description}
\item[Gaussian]  In the Gaussian model, we generate $n$ squares randomly in $\mathcal B$ according to an overlay of three Gaussian distributions, where 70\% of the squares are from the first distribution, 20\% from the second one, and 10\% from the third one. 
Firstly, the three means of the distributions are sampled uniformly at random in $\mathcal B$; for each Gaussian the standard deviation is set to $100$ in both dimensions. 
Next, in each distribution, we sample a point and check whether the unit square centered at this point is inside the bounding box $\mathcal B$. 
If so, we add this point to the input, otherwise we discard it.
This process is repeated until the required number of rectangles of each distribution is collected.

\item[Uniform] In the uniform model, we generate $n$ squares in $\mathcal B$ uniformly at random. %
\item[Real-world] We created six real-world data sets by extracting point features from OpenStreetMap (OSM), see Table~\ref{tab:osm} for their detailed properties.
Note that in our experiment, the input set is from a real-world instance and the new squares to insert are generated by an uniform sampling.
In Figure~\ref{fig:real-distribution}, we illustrate the  distribution of point features in our small real-world instances.
\end{description}

\begin{figure}[t]
	\centering
	\begin{subfigure}[t]{0.5\textwidth}
		\centering
		\includegraphics[width=\textwidth]
		{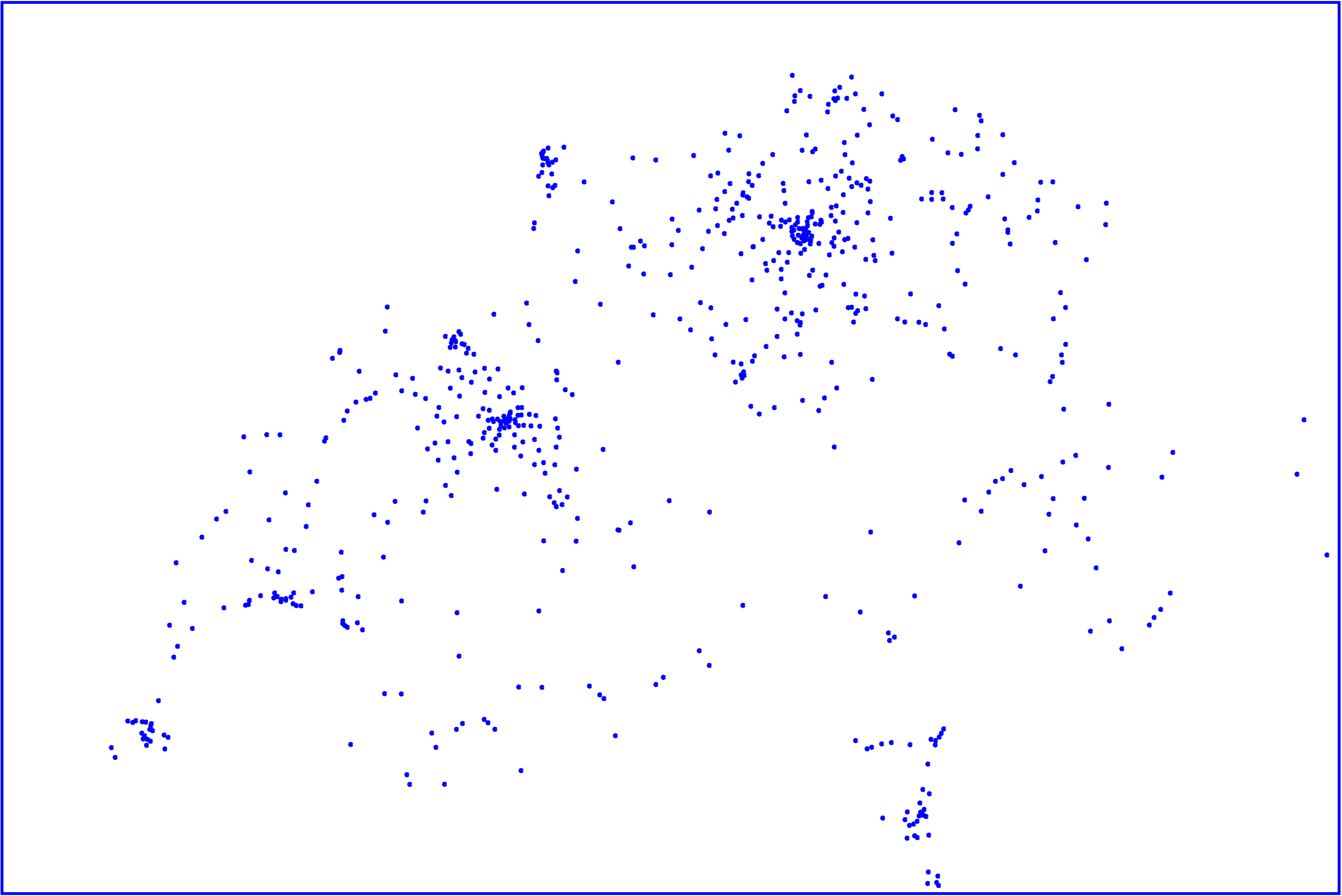}
		\caption{Post-CH}\label{fig:dis1}
	\end{subfigure}%
	\hfill
	\begin{subfigure}[t]{0.5\textwidth}
		\centering
		\includegraphics[width=\textwidth]
		{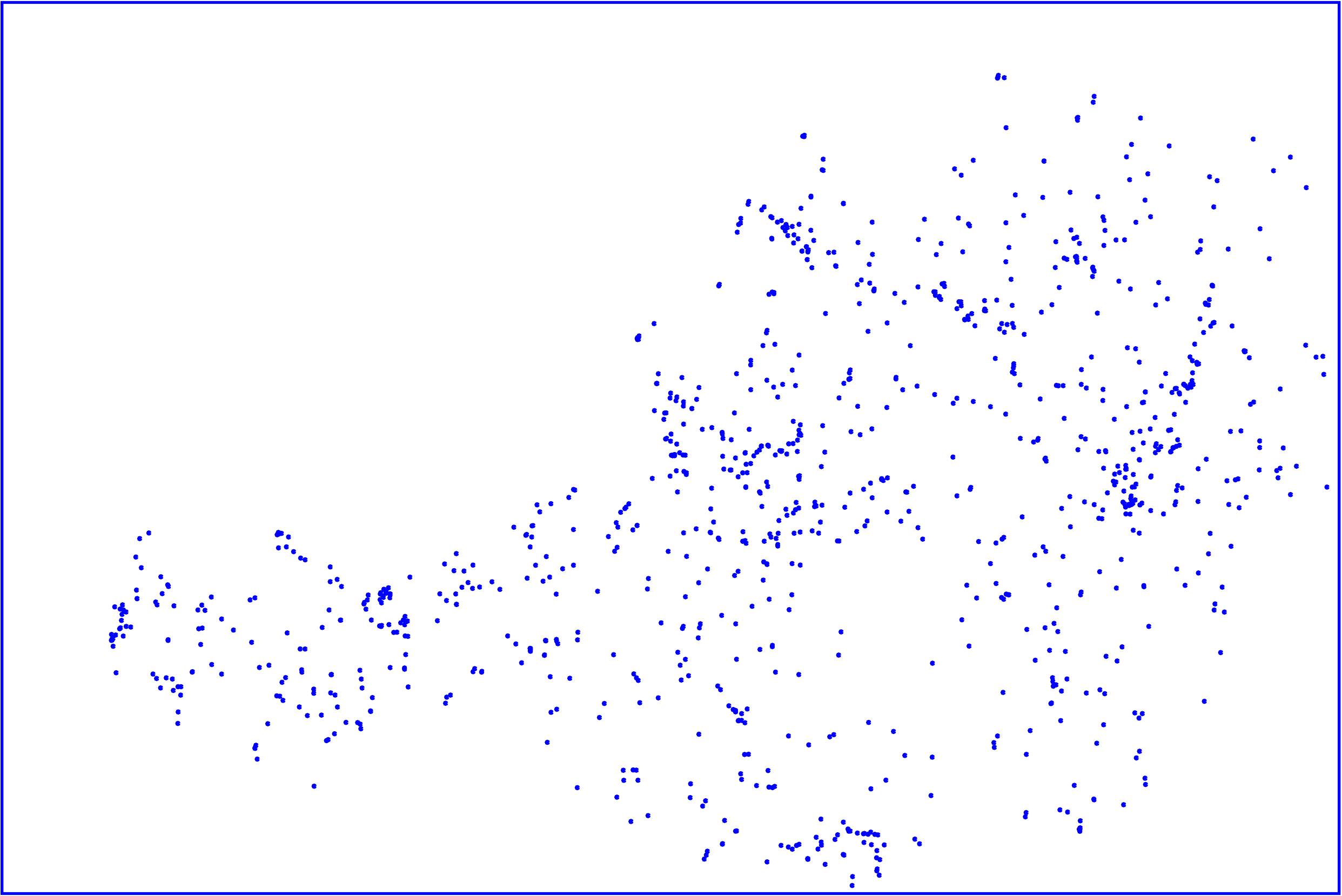}
		\caption{viewpoint-AT}\label{fig:dis2}
	\end{subfigure}
\caption{Example distributions of small real-world instances}
\label{fig:real-distribution}
\end{figure}

\begin{table}[htb]
	\begin{center}
		\begin{tabular}{l*{6}{|r}}
			&  post-CH & viewpoint-AT & hotels-CH & hotels-AT &  peaks-CH  & hamlets-CH\\
			\hline
			features ($n$)  & 646 & 652 & 1\,788 & 2\,209 & 4\,320 & 4\,326\\
			overlaps ($m$) & 5\,376 & 5\,418 & 28\,124 & 68\,985 & 107\,372 & 159\,270\\
			density ($m/n$) & 8.32 & 8.31 & 15.73 & 31.23 & 24.85 & 36.92\\
		\end{tabular}
	\end{center}
	\caption{Specification of the six OSM unit square instances. }
	\label{tab:osm}
\end{table}

\subparagraph*{\textbf{Benchmark Data of Unit-Height Rectangles.}}
Similarly to the unit squares data, we created three types of benchmark instances with unit-height rectangles.
Since the approach MIS-graph requires the corner intersection property of the input set, we created instances such that 
the boundary of every pair of rectangles intersects at most twice.
The synthetic data sets consist of unit-height rectangles with height of $10$ pixels %
placed inside a bounding rectangle $\mathcal B$ of size $1\,080 \times 720$ pixels, which also creates different densities. 
The real-world instances use the geographic feature distributions and the original label texts extracted from the OSM files. Each label of length $l$ is represented as a $10l\times 10$-pixel rectangle placed inside the bounding rectangle $\mathcal B$. For synthetic data sets, we sample a label length for each rectangle between $2$ and $21$ based on the distribution of word lengths in English.\footnote{taken from \url{ http://www.ravi.io/language-word-lengths}}

\begin{description}
\item[Gaussian]  In the Gaussian model, we generate center points of $n$ unit-height rectangles randomly in $\mathcal B$ according to an overlay of three Gaussian distributions, where 70\% of the center points are from the first distribution, 20\% from the second one, and 10\% from the third one. The means that they are sampled uniformly at random in $\mathcal B$ and the standard deviation is $100$ in both dimensions.
Once each center point is sampled, we sample a label length for it based on the distribution of word lengths in English. 
Whenever, a new rectangle candidate is generated, we check if it is inside the bounding box $\mathcal B$.
Furthermore, we verify if the rectangles set keeps the corner intersection property after adding this newly generated rectangle. 
This process is repeated for each distribution until the number of collected rectangles reaches the corresponding required number for this distribution. 
\item[Uniform] In the uniform model, we generate $n$ unit-height rectangles in $\mathcal B$ uniformly at random. 
To guarantee the corner intersection property, we check each newly generated rectangle.
\item[Real-world] We created six real-world data sets by extracting point features from OpenStreetMap (OSM), see Table~\ref{tab:osm} for their detailed properties.
In order to guarantee the corner intersection property of the instances, we filter the the instances in a post-processing step; see Table~\ref{tab:osm2} for their detailed properties.  
\end{description}

\begin{table}[htb]
	\begin{center}
		\begin{tabular}{p{3.6cm}|p{1.1cm}|p{1.25cm}
		|p{1.2cm}|p{1.25cm}|p{1.25cm}|p{1.25cm}}
			&post-CH&viewpoint-AT &hotels-CH&hotels-AT&peaks-CH&hamlets-CH\\
			\hline
			features ($n$)&612 & 918& 1\,476 & 1\,755 & 2\,899&2\,930\\
			overlaps ($m$) & 4\,688  &6\,710& 19\,886 & 46\,621 & 54\,912 & 64\,098\\
			density ($m/n$) & 7.66 & 7.3 & 13.47 & 18.56 & 18.94& 21.87\\
			label lengths(range/aver.)& 3-50/18& 3-63/16  & 2-46/15  &  2-57/17& 3-63/12  & 3-29/10 \\
		\end{tabular}
	\end{center}
	\caption{Specification of the six OSM unit-height rectangle instances. }
	\label{tab:osm2}
\end{table}

The source code of benchmark instance generator is available on \url{https://dyna-mis.github.io/dynaMIS/}.

\FloatBarrier
\subsection{Experimental Results for Unit Squares}\label{sec:results}
\subparagraph*{\textbf{Time-Quality Trade-offs.}}
For our first set of experiments we compare the five implemented algorithms, including their greedy variants, in terms of update time and size of the computed independent sets.
Figure~\ref{fig:plot-unifrom} shows scatter plots of runtime vs.\ solution size on uniform and Gaussian benchmarks,
where algorithms with dots in the top-left corner perform well in both measures.

We first consider the results for the uniform instances with $n=10\,000$ squares in the top row of Figure~\ref{fig:plot-unifrom}. 
Each algorithm performed $N=400$ updates, either insertions (Figure~\ref{fig:uni-ins}) or deletions (Figure~\ref{fig:uni-del}) and each update is shown as one point in the respective color.
\begin{figure}[ht]
	\begin{subfigure}[t]{\textwidth}
		\centering
		\includegraphics[width=\textwidth]{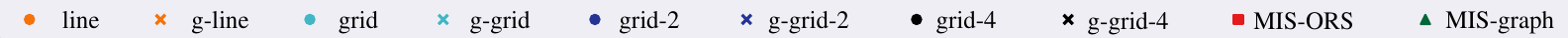}
	\end{subfigure}

	\centering
	\begin{subfigure}[t]{0.5\textwidth}
		\centering
		\includegraphics[width=\textwidth]
		{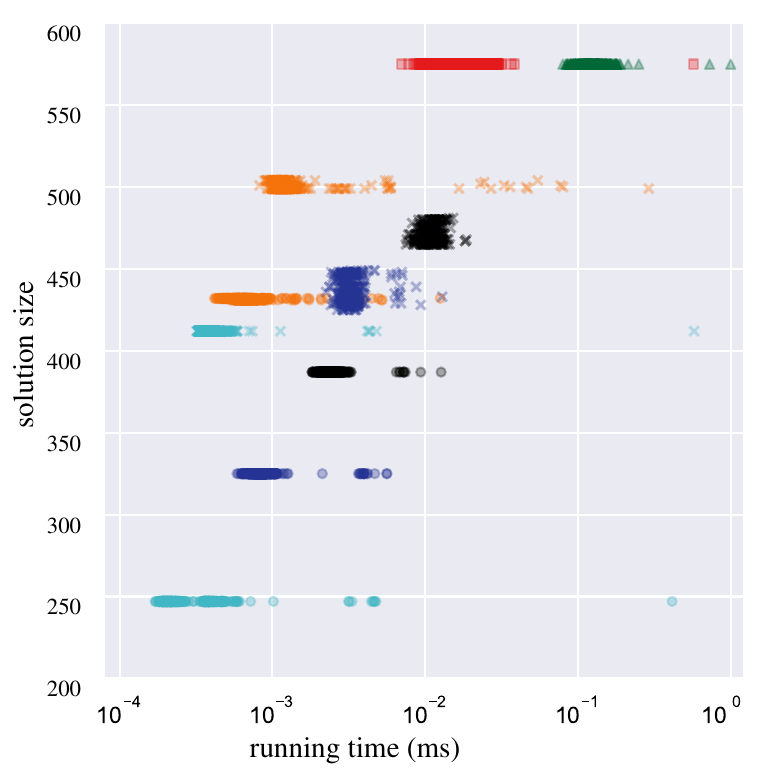}
		\caption{Uniform, $n=10\,000$, $400$ insertions}\label{fig:uni-ins}
	\end{subfigure}%
	\hfill
	\begin{subfigure}[t]{0.5\textwidth}
		\centering
		\includegraphics[width=\textwidth]
		{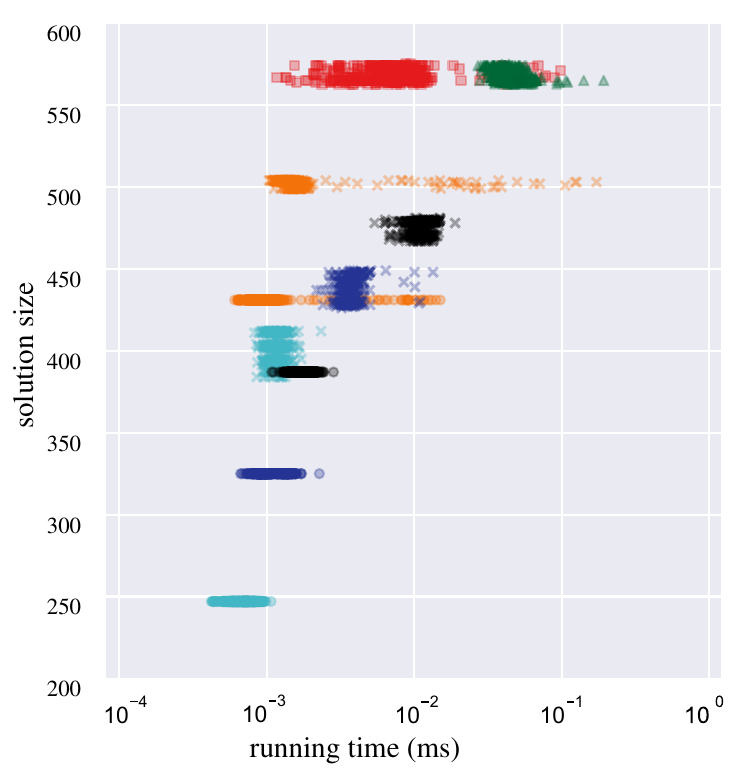}
		\caption{Uniform, $n=10\,000$, $400$ deletions}\label{fig:uni-del}
	\end{subfigure}
	
		\begin{subfigure}[t]{0.5\textwidth}
		\centering
		\includegraphics[width=\textwidth]
		{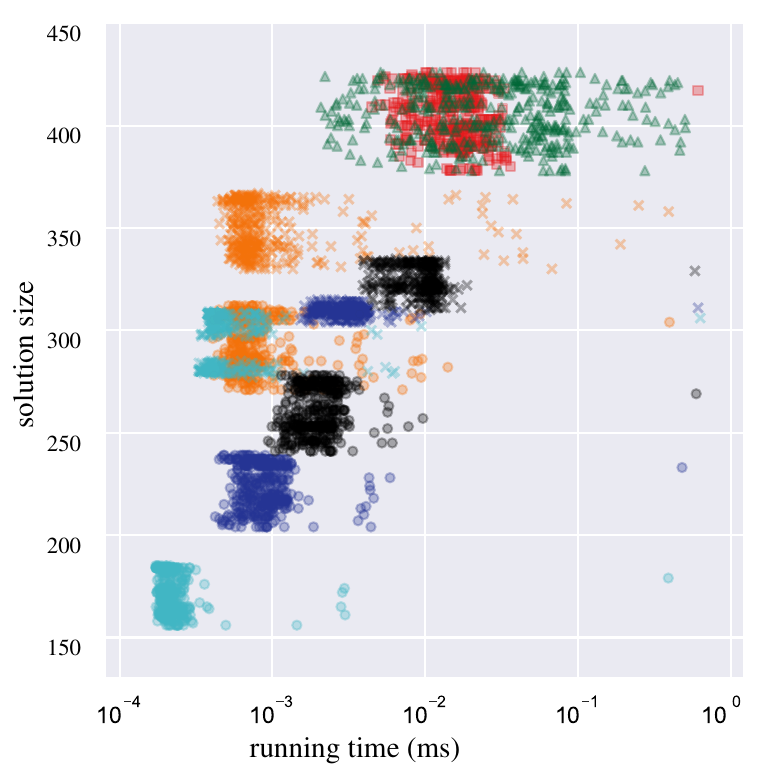}
		\caption{Gaussian, $n=10\,000$, $400$ insertions}\label{fig:gauss-ins:appendix}
	\end{subfigure}%
	\hfill
	\begin{subfigure}[t]{0.5\textwidth}
		\centering
		\includegraphics[width=\textwidth]
		{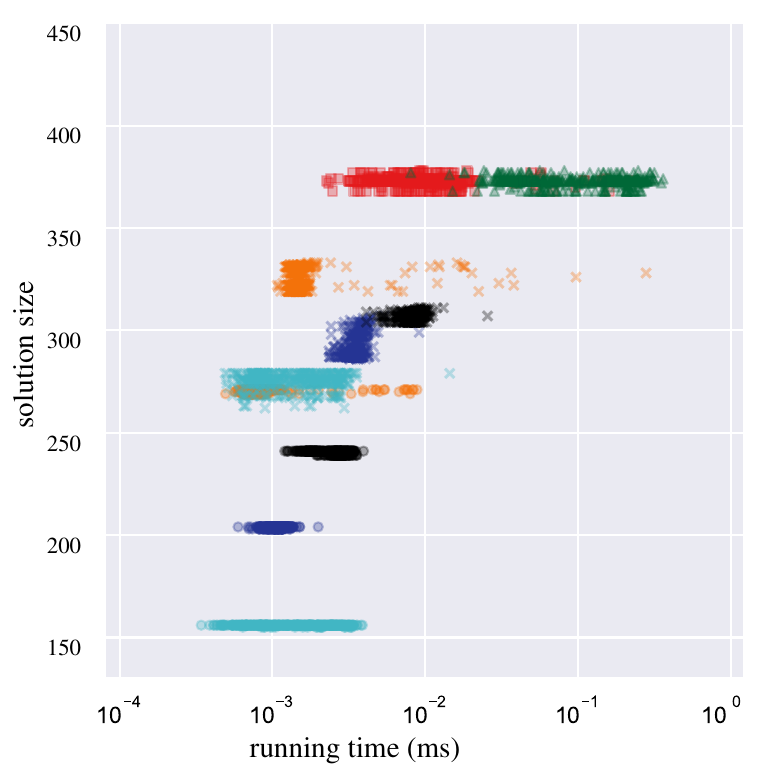}
		\caption{Gaussian, $n=10\,000$, $400$ deletions}\label{fig:gauss-del:appendix}
	\end{subfigure}

\caption{Time-quality scatter plots for synthetic benchmark instances. The x-axis (log-scale) shows runtime, the y-axis shows the solution size. We use semi-transparent markers in the scatter plots.
}
	\label{fig:plot-unifrom}
\end{figure}
Both plots show that the two \indset algorithms compute the best solutions with almost the same size and well ahead of the rest.
The \indset algorithm \misr\ is clearly faster than \misg\ on both insertions and deletions.
The approximation algorithms \grid, \gridk{2}, \gridk{4}, and \linea\ (without the greedy optimizations) show their predicted relative behavior: The better the solution quality, the worse the update times. 
Algorithms \linea\ and \glinea\ show a wide range of update times, spanning almost two orders of magnitude. 
Adding the greedy optimization drastically improves the solution quality in all cases, but typically at the cost of higher runtimes. 
For \ggridk{k} the algorithms get slower by an order of magnitude and increase the solution size by 30--50\%.
For \ggrid, the additional runtime is not as significant (but deletions are slower than insertions), and the solution size almost doubles.
Finally, for \glinea, the additional runtime is not as significant, and reaches the best quality among the approximation algorithms with about 90\% of the \indset solutions, but faster by one order of magnitude.

For the results of the Gaussian instances with $n=10\,000$ squares and $N=400$ updates plotted in Figures~\ref{fig:gauss-ins:appendix} (insertions) and~\ref{fig:gauss-del:appendix} (deletions) we observe the same ranking between the different algorithms.
However, due to the non-uniform distribution of squares, the solution sizes are more varying, especially for the insertions. 
For the deletions it is interesting to see that \grid\ and \misg\ have more strongly varying runtimes, which is in contrast to the deletions in the uniform instance, possibly due to the dependence on the vertex degree.
The best solutions are computed by \misr\ and \misg.
Regarding the runtime, \misr\ has more homogeneous update times ranging between the extrema of \misg, while they are comparably fast for insertions in average.

Algorithm \glinea\ again reaches nearly 90\% of the quality of the \indset algorithms, with a speed-up almost one order of magnitude.

\subparagraph*{\textbf{Optimality Gaps.}} Next, let us look at the results of the real-world instances in Figure~\ref{fig:plot-prac-small} and in Figure~\ref{fig:plot-prac-large}. %
The first four instances in Figure~\ref{fig:plot-prac-small} were small enough so that we could compute each \maxset exactly with MaxHS and compare the solutions of the approximation algorithms with the optimum on the y-axis.
The largest two instances in Figure~\ref{fig:plot-prac-large} plot the solution size on the y-axis.
First, let us consider Figure~\ref{fig:hotels:appendix} as a representative, which is based on a data set of 1\,788 hotels and hostels in Switzerland with mixed updates of 10\%  of the squares ($N=179$). 
Generally speaking, the results of the different algorithms are much more overlapping in terms of quality than for the synthetic instances. 
The plot shows that the \indset algorithms reach consistently between 80\% and 85\% of the optimum, but are sometimes outperformed by the greedy-augmented approaches.
Interestingly, \glinea, the best of the approximation algorithms with greedy augmentation, contributes consistently best solutions.
Regarding the runtime, \misr\ has generally faster update time than \misg\ approach.
The original approximations are well above their respective worst-case ratios, but stay between 45\% and 65\% of the optimum. 
The greedy extensions push this towards larger solutions, at the cost of higher runtimes. 
However, \glinea\ seems to provide a very good balance between quality and speed.

\begin{figure}[tbp]	
	\centering
	\begin{subfigure}[t]{0.5\textwidth}
		\centering
		\includegraphics[width=\textwidth]
		{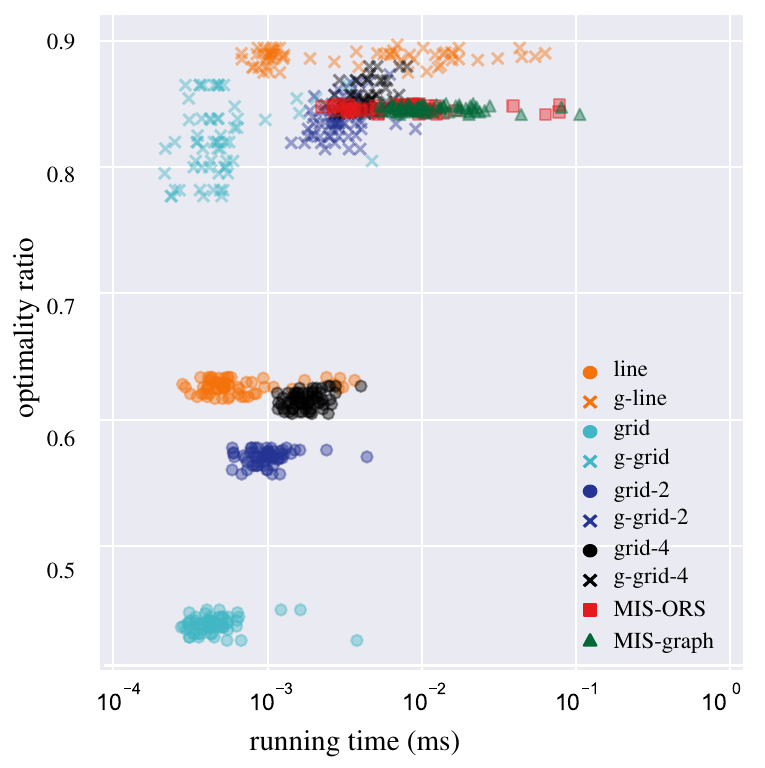}
		\caption{post-CH, 10\% mixed updates}\label{fig:post-ch}
	\end{subfigure}%
	\hfill
	\begin{subfigure}[t]{0.5\textwidth}
		\centering
		\includegraphics[width=\textwidth]
		{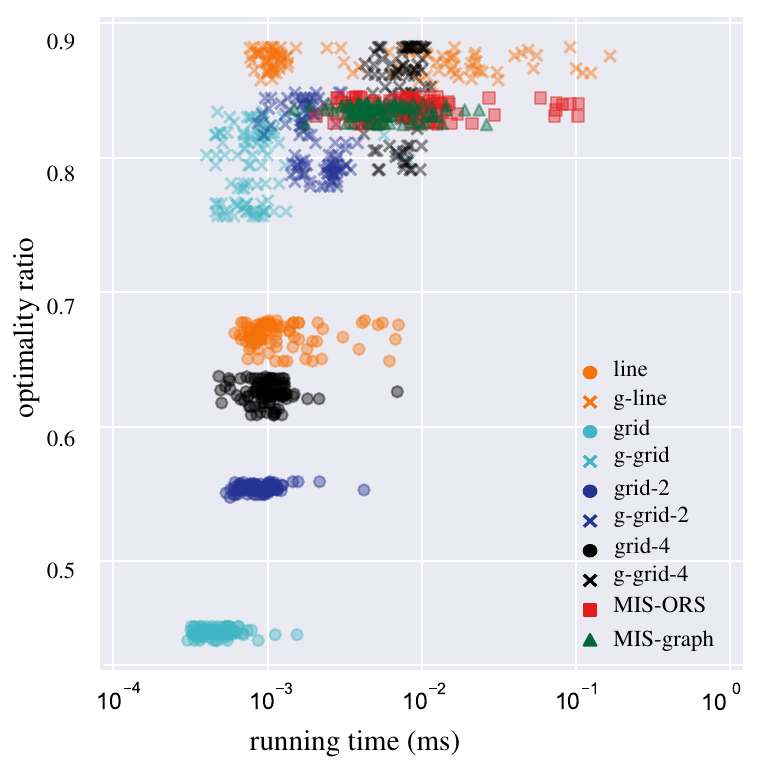}
		\caption{viewpoint-AT, 10\% mixed updates}\label{fig:peaks-at}
	\end{subfigure}	
	\begin{subfigure}[t]{0.5\textwidth}
		\centering
		\includegraphics[width=\textwidth]
		{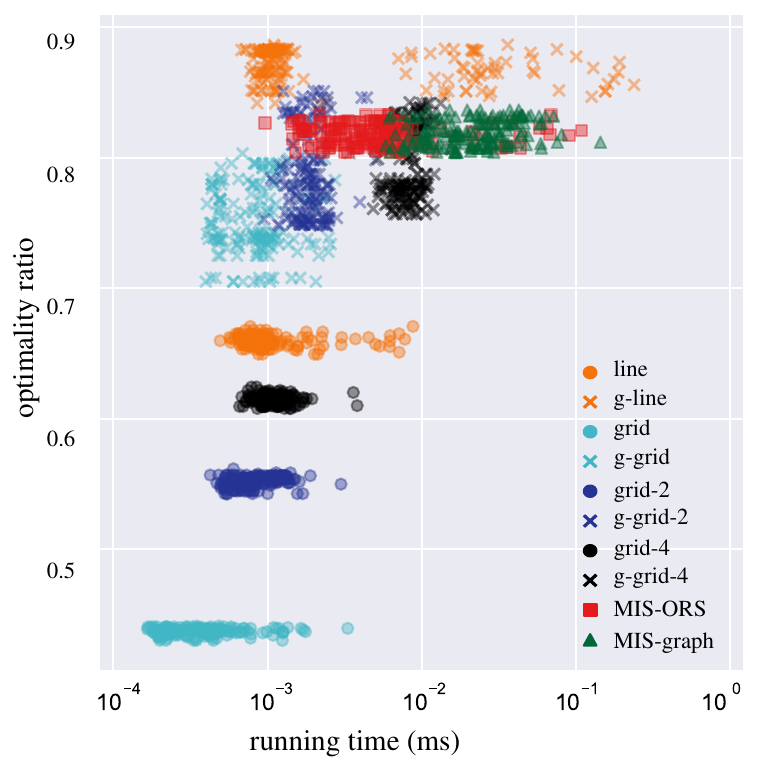}
		\caption{hotels-CH, 10\% mixed updates}
		\label{fig:hotels:appendix}
	\end{subfigure}%
	\hfill
	\begin{subfigure}[t]{0.5\textwidth}
		\centering
		\includegraphics[width=\textwidth]
		{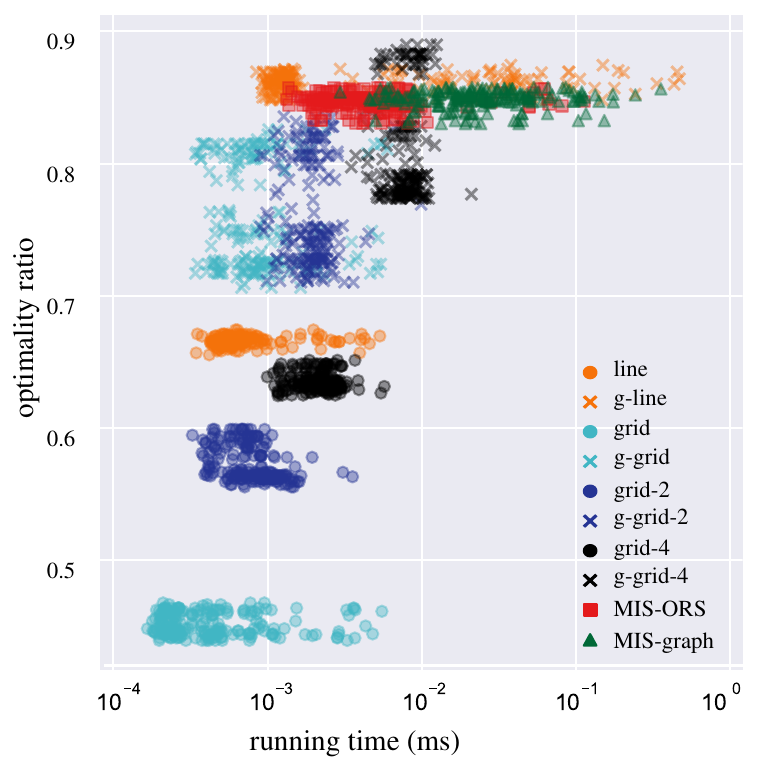}
		\caption{hotels-AT, 10\% mixed updates}\label{fig:hotels-at}
	\end{subfigure}	
	\caption{Time-quality scatter plots for the small OSM instances. The x-axis (log-scale) shows runtime. The y-axis shows the quality ratio compared to an optimal \maxset solution.}
	\label{fig:plot-prac-small}
\end{figure}
\begin{figure}[tbp]	
	\begin{subfigure}[t]{0.5\textwidth}
		\centering
		\includegraphics[width=\textwidth]
		{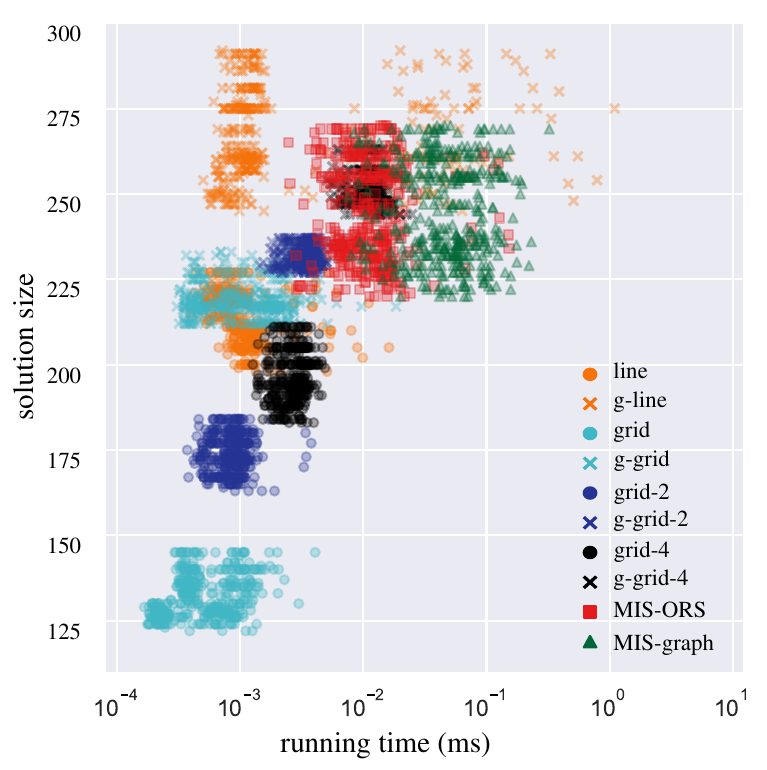}
		\caption{peaks-CH, 10\% mixed updates}\label{fig:peaks-ch}
	\end{subfigure}%
	\hfill
	\begin{subfigure}[t]{0.5\textwidth}
		\centering
		\includegraphics[width=\textwidth]
		{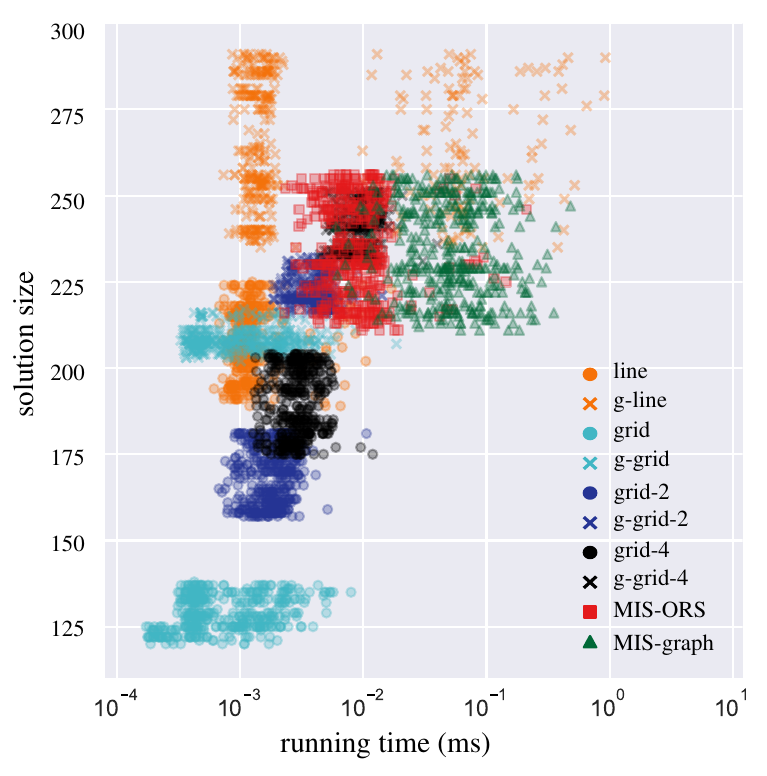}
		\caption{hamlets-CH, 10\% mixed updates}\label{fig:hamlets:appendix}
	\end{subfigure}	
	\caption{Time-quality scatter plots for the large OSM instances. The x-axis (log-scale) shows runtime. The y-axis shows the quality ratio compared to the optimal \maxset solution size.}
	\label{fig:plot-prac-large}
\end{figure}

Let us next consider the largest OSM instance in Figure~\ref{fig:hamlets:appendix}. 
It again reflects the same findings as obtained from the smaller instances. 
The instance consists of $n=4\,326$ hamlets in Switzerland with 10\% mixed updates ($N=433$) and is denser by a factor of about $2.3$ than hotels-CH (see Table~\ref{tab:osm}).
There is quite some overlap of the different algorithms in terms of the solution size, yet the algorithms form the same general ranking pattern as observed before. 
The approach \glinea\ contributes best solutions in most of the rounds.
Moreover, regarding the running time, \glinea\ is again about nearly an order of magnitude faster than the \indset algorithms, except for a few slower outliers.
Comparing the two \indset approaches, \misr\ is significantly faster than \misg.

\begin{figure}[t!]
	\centering
	\begin{subfigure}[t]{0.5\textwidth}
		\centering
		\includegraphics[width=\textwidth]
		{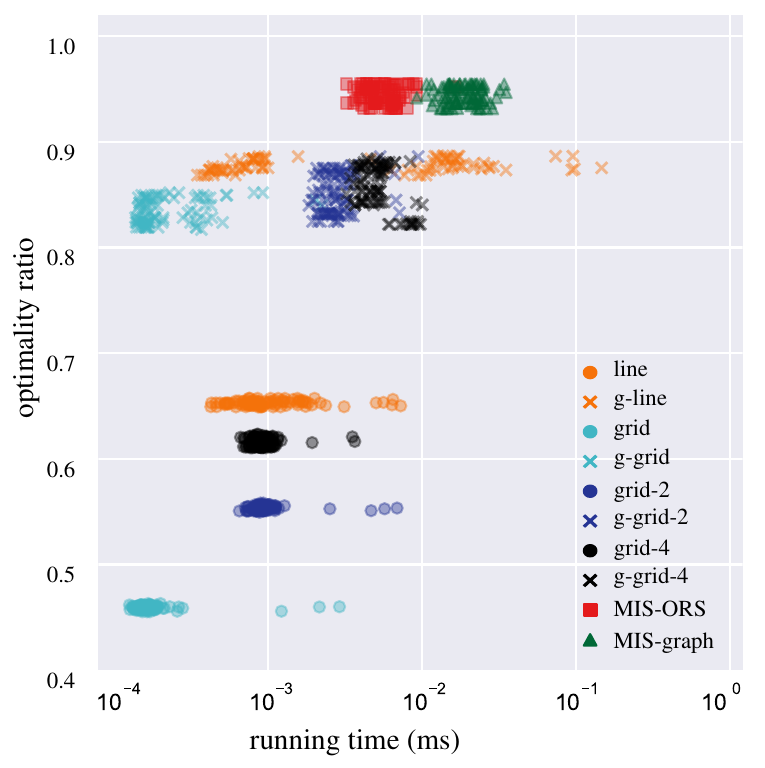}
		\caption{Uniform, $100$ insertions}\label{fig:uni-ins-opt}
	\end{subfigure}%
	\hfill
	\begin{subfigure}[t]{0.5\textwidth}
		\centering
		\includegraphics[width=\textwidth]
		{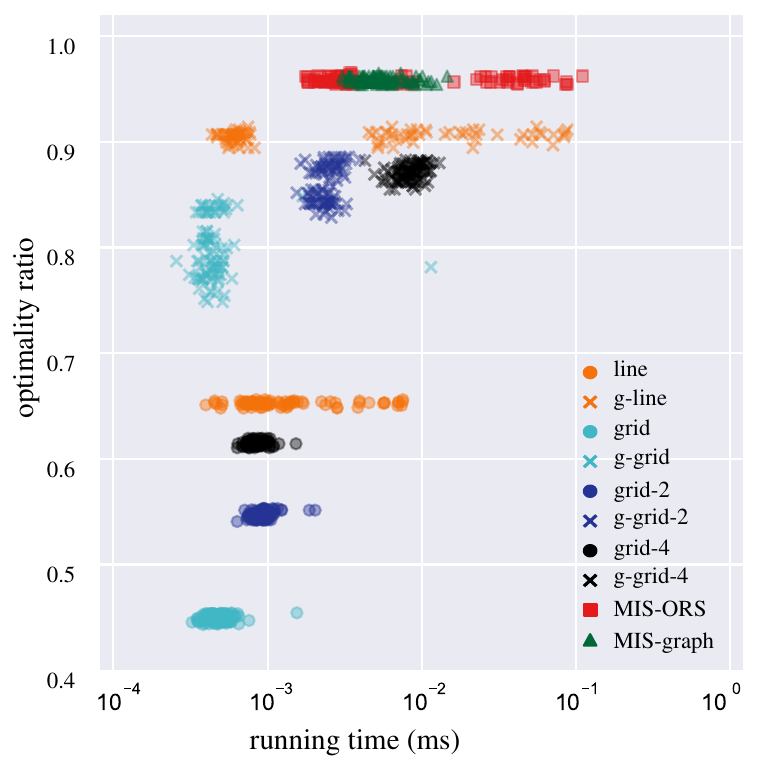}
		\caption{Uniform, $100$ deletions}\label{fig:uni-del-opt}
	\end{subfigure}
	
	\begin{subfigure}[t]{0.5\textwidth}
		\centering
		\includegraphics[width=\textwidth]
		{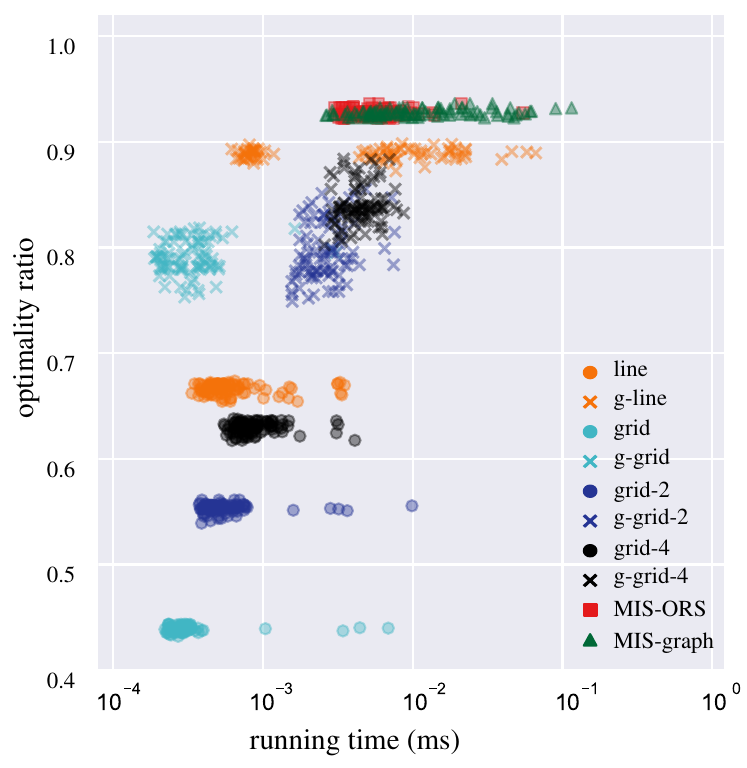}
		\caption{Gaussian, $100$ insertions}\label{fig:gau-ins-opt}
	\end{subfigure}%
	\hfill
	\begin{subfigure}[t]{0.5\textwidth}
		\centering
		\includegraphics[width=\textwidth]
		{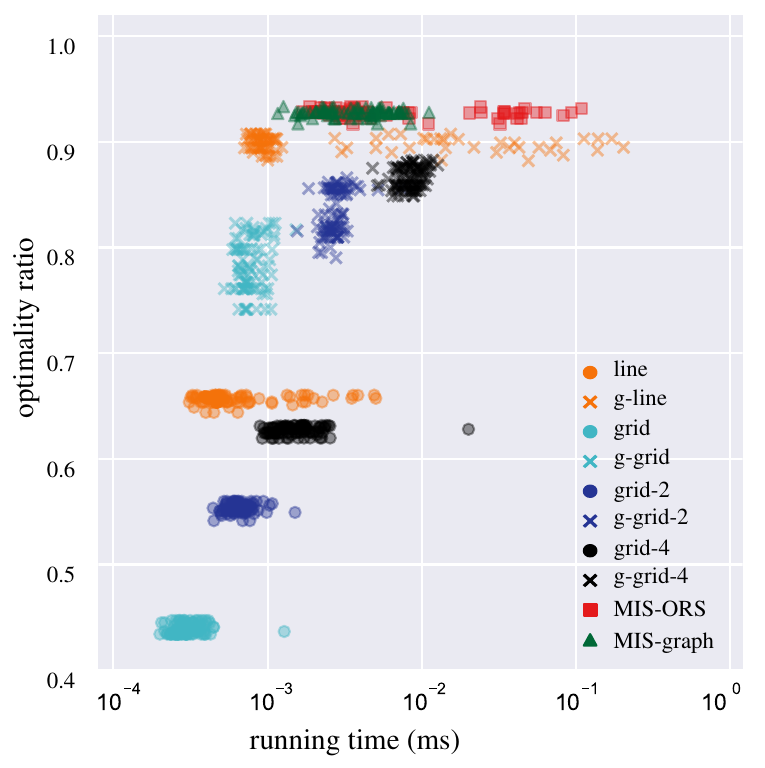}
		\caption{Gaussian, $100$ deletions}\label{fig:gau-del-opt}
	\end{subfigure}

	\caption{Time-quality scatter plots for uniform and Gaussian instances with $n=1\,000$ squares. The x-axis (log-scale) shows runtime. The y-axis shows the quality ratio compared to an optimal \maxset solution.}
	\label{fig:plot-opt-synth}
\end{figure}

Finally, Figure~\ref{fig:plot-opt-synth} shows the optimality ratios of the algorithms for small uniform and Gaussian instances with $n=1\,000$ squares. 
They confirm our earlier observations, but also show that for these small instances, \misg\ and \misr\ are comparable in terms of running time.
This is because the graph size and vertex degrees do not yet influence the running time of \misg\ strongly. Yet, as the next experiment shows, this changes drastically, as the instance size grows.

\FloatBarrier
\subparagraph*{\textbf{Runtimes.}}
In our last experiment, we explore in more detail the scalability of the algorithms for larger instances, both relative to each other and in comparison to the re-computation times of their corresponding static algorithms.
We generated $10$ random instances with $n=1\,000 k$ squares for each $k \in \{1, 2, 4, 8, 16, 32\}$ and measured the average update times over $n/10$ insertions or deletions.
The results for the Gaussian and uniform model are plotted in Figure~\ref{fig:line-plot-all-gaussian} and in Figure~\ref{fig:line-plot-all-uniform} . 
Considering the update times, we confirm the observations from the scatter plots in terms of the performance ranking. 
The running time of most algorithms grows only very slowly as the input size grows larger with the notable exception of \misg, but that was to be expected. 

\begin{figure}[tbp]
	
	\begin{subfigure}[t]{\textwidth}
		\centering
		\includegraphics[width=\textwidth]{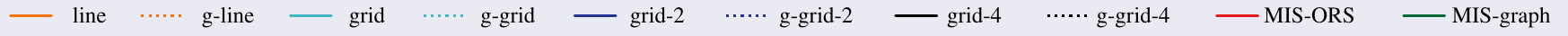}
	\end{subfigure}
	
	\centering
	\begin{subfigure}[t]{0.5\textwidth}
		\centering
		\includegraphics[width=\textwidth]
		{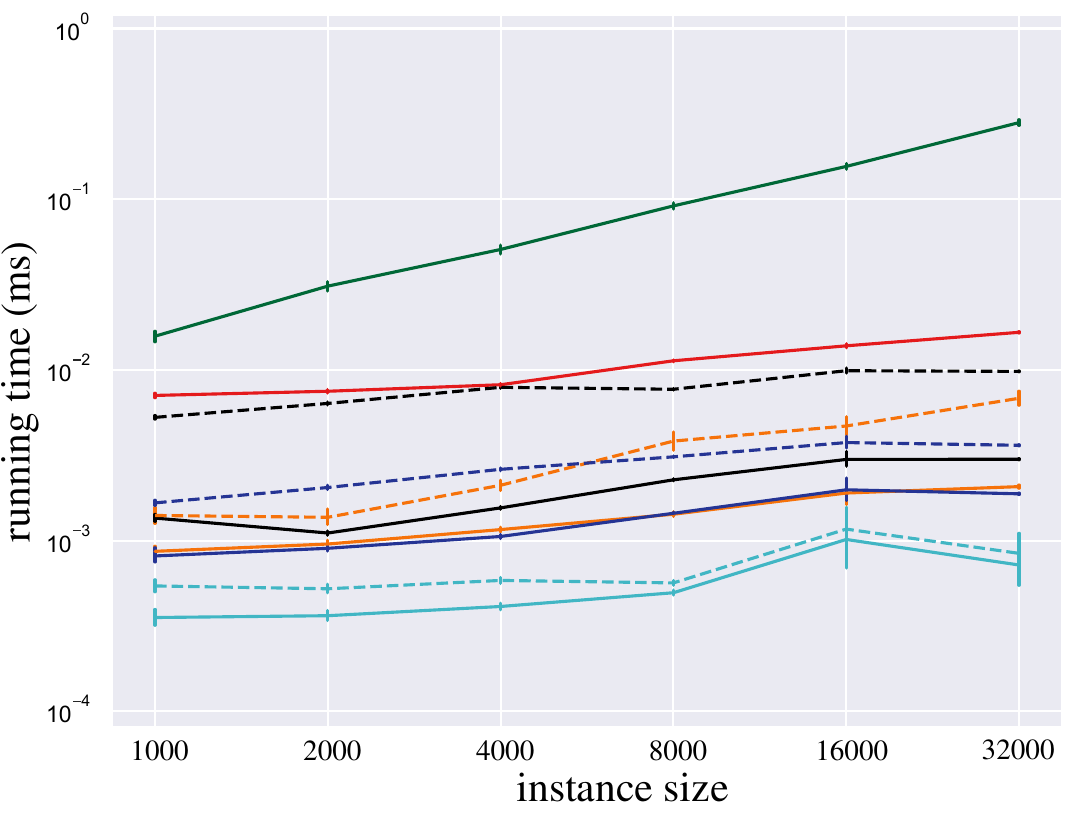}
		\caption{Update times for insertions (Gaussian)}\label{fig:updatetime-ins-gau}
	\end{subfigure}%
	\hfill
	\begin{subfigure}[t]{0.5\textwidth}
		\centering
		\includegraphics[width=\textwidth]
		{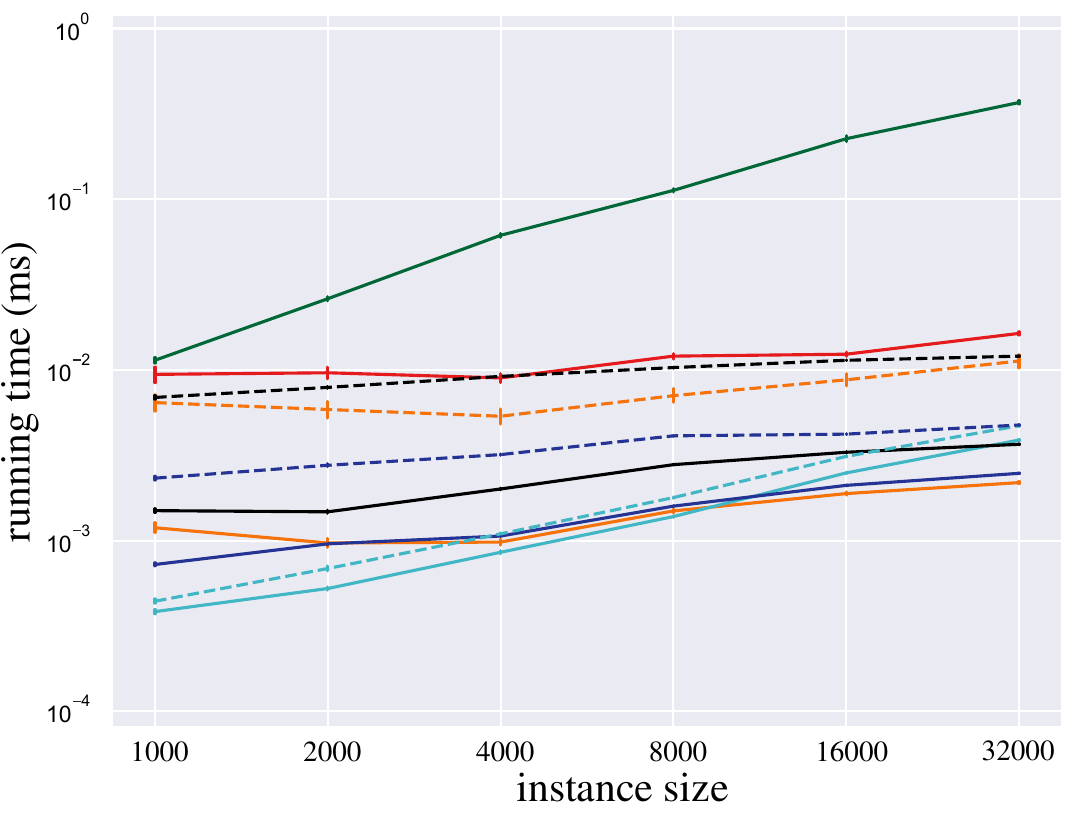}
		\caption{Update times for deletions (Gaussian)}\label{fig:updatetime-del-gau-apx}
	\end{subfigure}
	
	\begin{subfigure}[t]{0.5\textwidth}
		\centering
		\includegraphics[width=\textwidth]
		{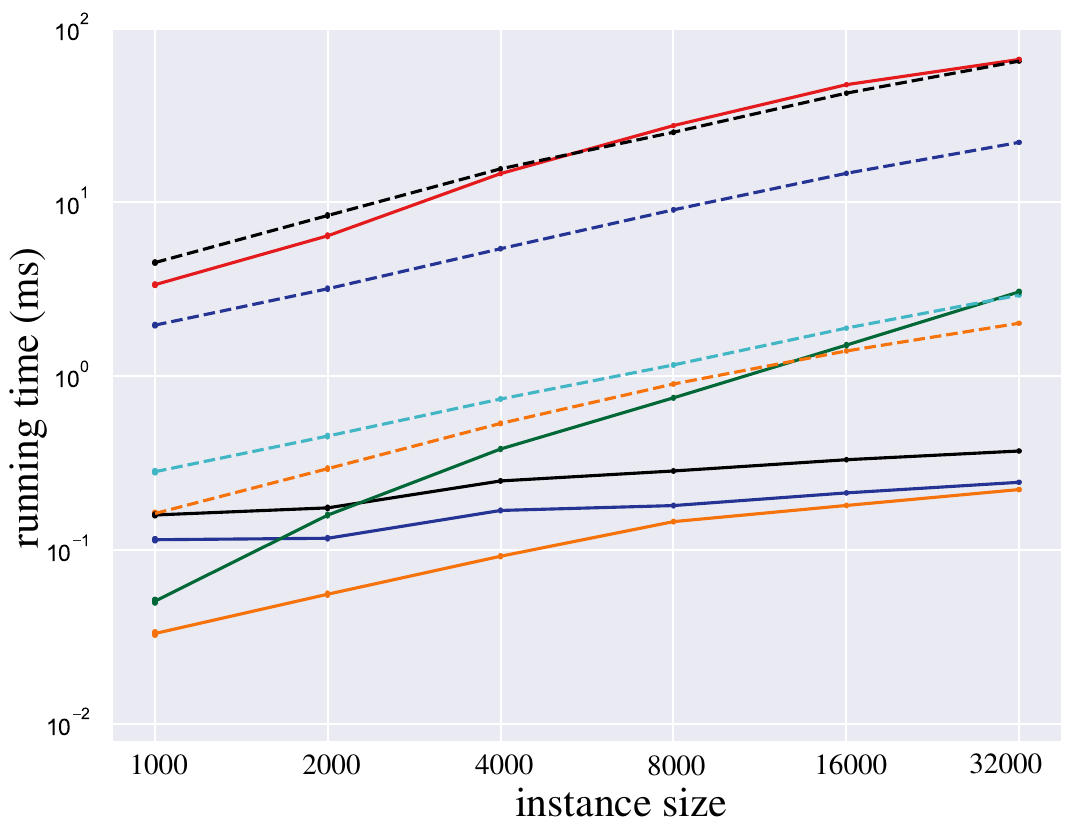}
		\caption{Re-computation times for insertions (Gaussian)}\label{fig:re-time-ins-gau}
	\end{subfigure}%
	\hfill
	\begin{subfigure}[t]{0.5\textwidth}
		\centering
		\includegraphics[width=\textwidth]
		{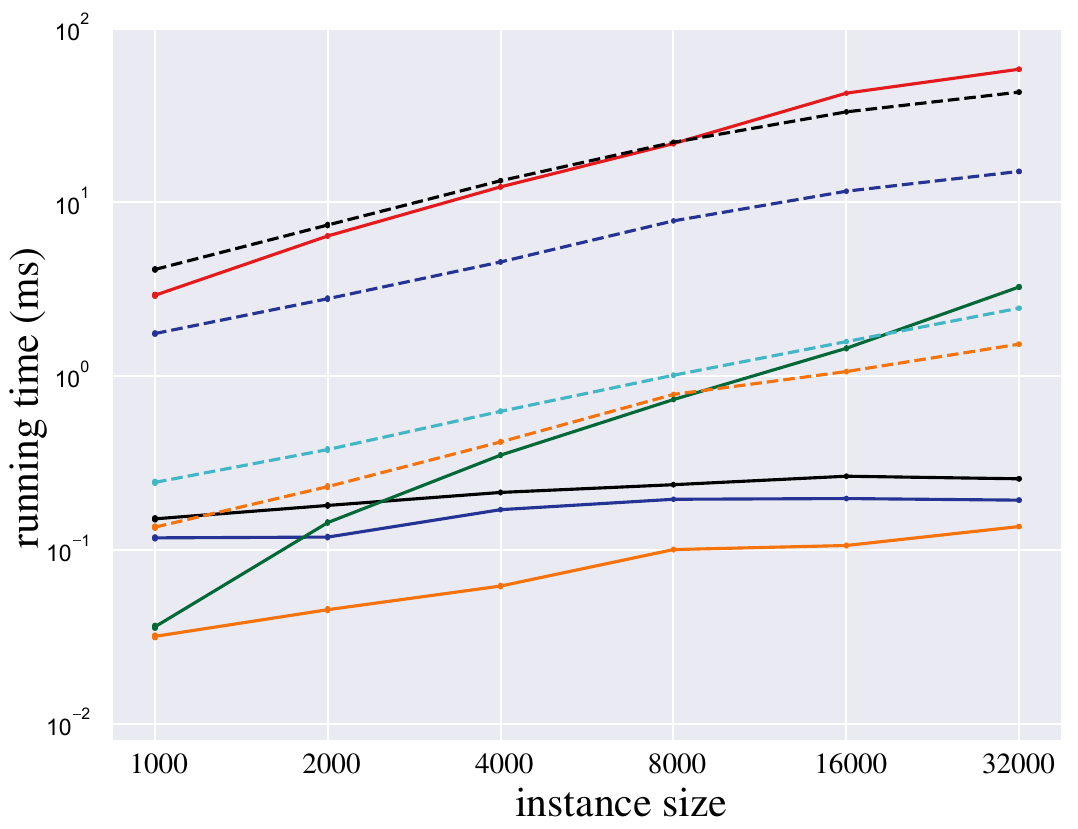}
		\caption{Re-computation times for deletions (Gaussian)}\label{fig:re-time-del-gau-apx}
	\end{subfigure}
	\caption{Log-log runtime plots (notice the different y-offsets) for dynamic updates and re-computation on Gaussian instances of size $n= 1\,000$ to $32\,000$, averaged over $n/10$ updates. Error bars indicate the standard deviation.}
	\label{fig:line-plot-all-gaussian}
\end{figure}

\begin{figure}[tbp]
	
	\begin{subfigure}[t]{\textwidth}
		\centering
		\includegraphics[width=\textwidth]{figs_ESA_dissertation/E5/legend.pdf}
	\end{subfigure}
	
	\centering
	\begin{subfigure}[t]{0.5\textwidth}
		\centering
		\includegraphics[width=\textwidth]
		{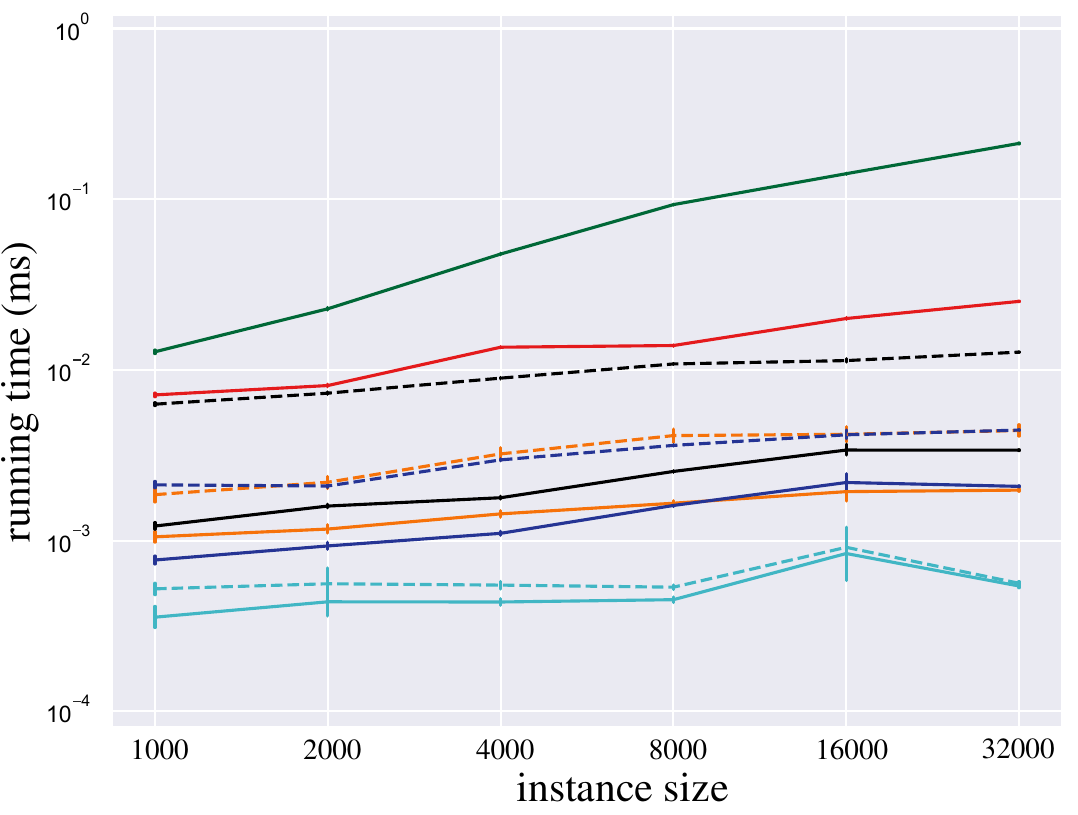}
		\caption{Update times for insertions (uniform)}\label{fig:updatetime-ins}
	\end{subfigure}%
	\hfill
	\begin{subfigure}[t]{0.5\textwidth}
		\centering
		\includegraphics[width=\textwidth]
		{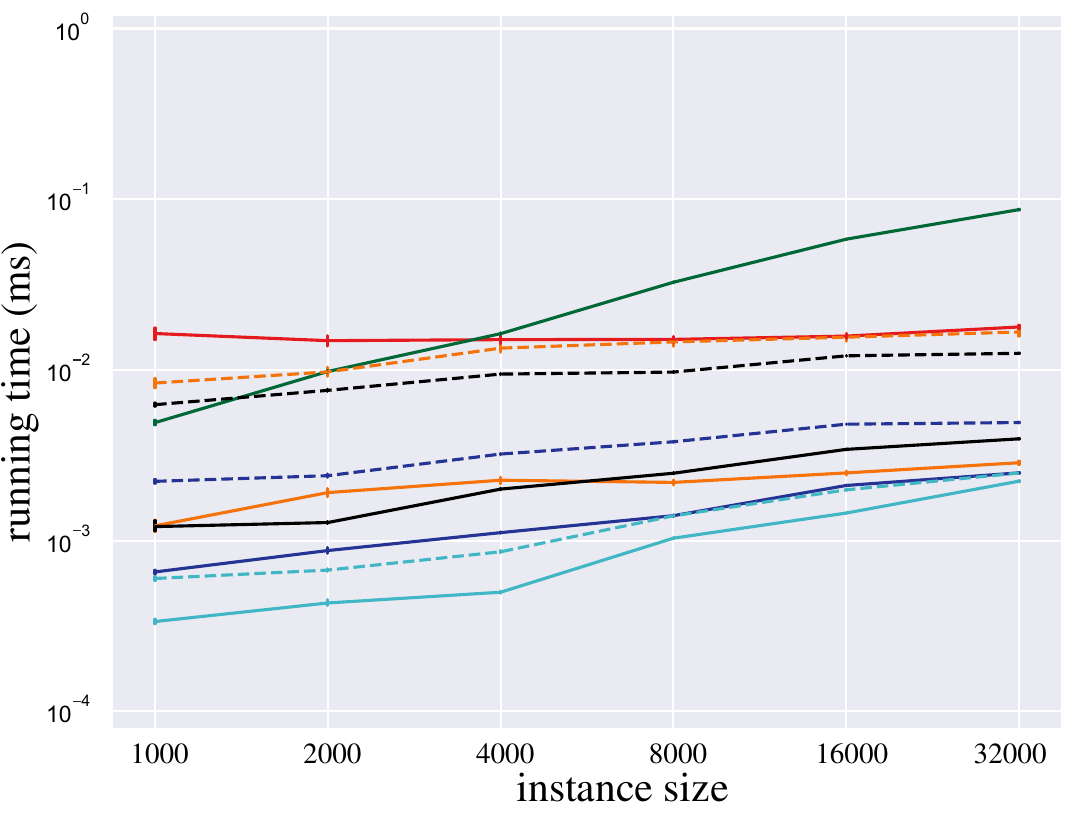}
		\caption{Update times for deletions (uniform)}\label{fig:updatetime-del}
	\end{subfigure}
	
	\begin{subfigure}[t]{0.5\textwidth}
		\centering
		\includegraphics[width=\textwidth]
		{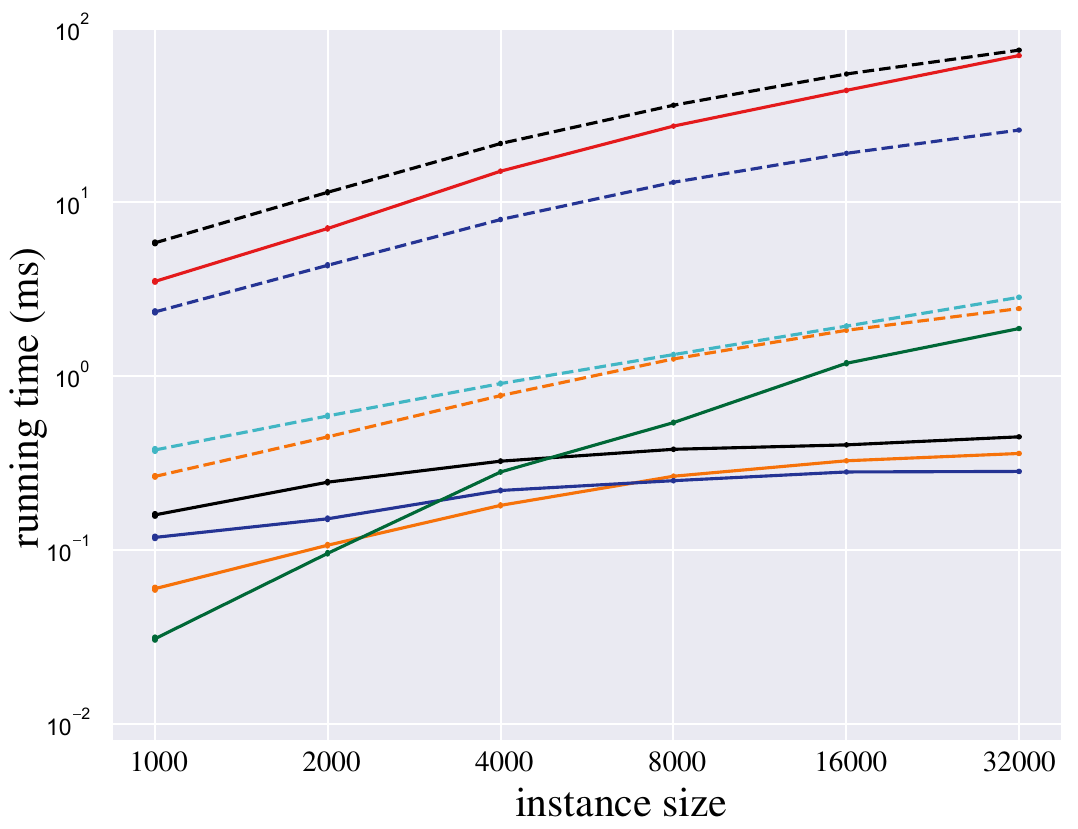}
		\caption{Re-computation times for insertions (uniform)}\label{fig:re-time-ins}
	\end{subfigure}%
	\hfill
	\begin{subfigure}[t]{0.5\textwidth}
		\centering
		\includegraphics[width=\textwidth]
		{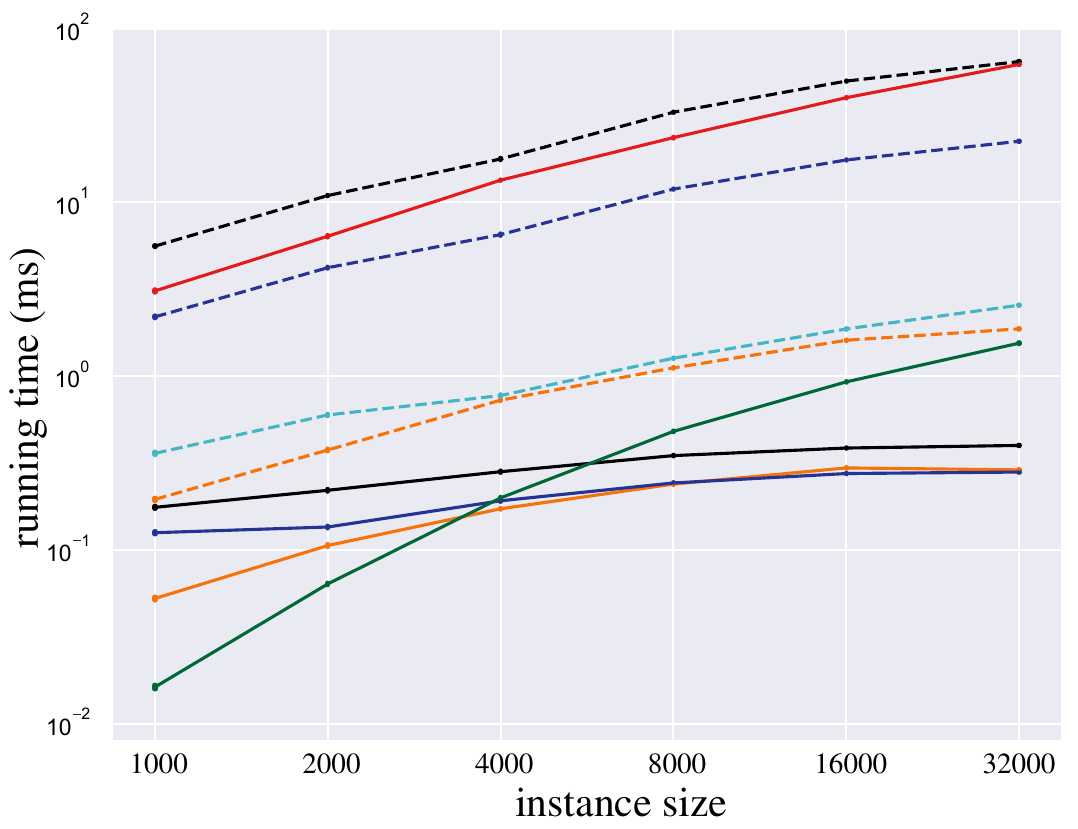}
		\caption{Re-computation times for deletions (uniform)}\label{fig:re-time-del}
	\end{subfigure}
	\caption{Log-log runtime plots (notice the different y-offsets) for dynamic updates and re-computation on uniform instances of size $n= 1\,000$ to $32\,000$, averaged over $n/10$ updates. Error bars indicate the standard deviation.}
	\label{fig:line-plot-all-uniform}
\end{figure}

In the comparison with their non-dynamic versions, i.e., re-computing solutions after each update, the dynamic algorithms indeed show a significant speed-up in practice, already for small instance sizes of $n=1\,000$, and even more so as $n$ grows (notice the different y-offsets). For some algorithms, including \misr\ and \ggridk{4}, this can be as high as 3--4 orders of magnitude for $n=32\,000$. It clearly confirms that the investigation of algorithms for dynamic \indset and \maxset problems for rectangles is well justified also from a practical point of view.
\subparagraph*{\textbf{{Discussion.}}}
Our experimental evaluation provides several interesting insights into the practical performance of the different algorithms.
For the synthetic instances, both \indset-based algorithms \misg\ and \misr\ generally showed the best solution quality in the field, reaching 90\% of the exact \maxset size, where we could compare against optimal solutions.
This is in strong contrast to their factor-4 worst-case approximation guarantee of only 25\%.

Our algorithm \misr\ avoids storing the intersection graph explicitly. Instead, we only store the relevant geometric information in a dynamic data structure and derive edges on demand. Therefore it overcomes the natural barrier of $\Omega(\Delta)$ vertex update in a dynamic graph, where $\Delta$ is the maximum degree in the graph. Instead, it has to find the intersections using the complex range query, which takes $O(\log n \log\log n)$ time.
We observe that \misr\ provides faster update times than \misg\ in general and is more scalable.
Recall that in our implementation, we used the dynamic range searching data structure from CGAL, which does not provide the theoretical worst-case update time of $O(\log n \log \log n)$ 
from Theorem~\ref{thm:mis-rt}.
Exploring how \misr\ can benefit from such a state-of-the-art dynamic data structure in practice remains to be investigated in future work.
Notwithstanding, it remains to state that even with the suboptimal data structure, \misr\  was able to compute its solutions for up to $32\,000$ squares in less than 1ms. %

An expected observation is that while consistently exceeding their theoretical guarantees, the approximation algorithms do not perform too well in practice due to their pigeonhole choice of too strictly separated subinstances.
However, a simple greedy augmentation of the approximate solutions can boost the solution size significantly, and for some algorithms even to almost that of the \indset algorithms. 
Of course, at the same time this increases the runtime of the algorithms. 
We want to point out \glinea, the greedy-augmented version of the 2-approximation algorithm \linea, as it computes very good solutions, even comparable or better than \misr\ and \misg\  for the real-world instances, and at 90\% of the \indset solutions for the synthetic instances. 
At the same time, \glinea\ is still significantly faster than \misr\ and \misg\ and thus turns out to be a well-balanced compromise between time and quality.
It is our recommended method if \misr\ or \misg\ are too slow for an application.

\FloatBarrier
\subsection{Experimental Results for Unit-Height Rectangles}
In this section, we compare our approaches \linea and \misg for unit-height rectangles with the same sets of experiments as in the previous section. 
Recall that the \linea\ approach performs a greedy \maxset approach for interval graphs on each stabbing line where the intervals are sorted with increasing right endpoints.
For fairness, we also provide the ``greedy'' version of \misg\ approach, g-\misg.
In the initialization phase, g-\misg\ sorts the vertices of the conflict graph by vertex degrees incrementally and builds the maximal independent set greedily by iterating the vertices in this new order. 
We expect that this greedy variant would provide larger solutions than \misg\ and the same update time as \misg.
\subparagraph*{\textbf{Time-quality trade-offs.}}
For our first set of experiments, we compare the \linea\ approach with \misg, including their greedy variants, in terms of update time and size of the computed solution. 
Figure~\ref{fig:plot-unifrom-rect} shows scatter plots of runtime vs solution size on uniform and Gaussian rectangle benchmarks. 
For each instance with $n=10\,000$, each algorithm performed $400$ updates, either insertions (Figure~\ref{fig:uni-ins-rect} and Figure~\ref{fig:gauss-ins:appendix-rect}) or deletions (Figure~\ref{fig:uni-del-rect} and  Figure~\ref{fig:gauss-del:appendix-rect}).
All plots show that g-\misg\, the greedy variant of \misg, computes the best solutions and well ahead of the rest. 
Both greedy variants increase the solution size significantly without significant additional runtime.
Note that g-\linea\ is nearly two orders of magnitude faster than graph-based approaches with about $90\%$ of the \indset solution obtained by g-\misg.
It is interesting to see that g-\linea\ gets larger solutions than \misg, which is in contrast with our experimental results for unit squares. 

\begin{figure}[tb]
	\begin{subfigure}[t]{\textwidth}
		\centering
		\includegraphics[width=.5\textwidth]{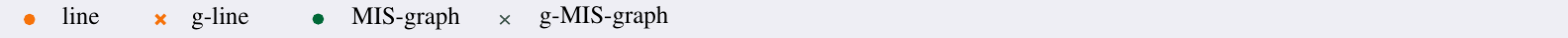}
	\end{subfigure}

	\centering
	\begin{subfigure}[t]{0.5\textwidth}
		\centering
		\includegraphics[width=\textwidth]
		{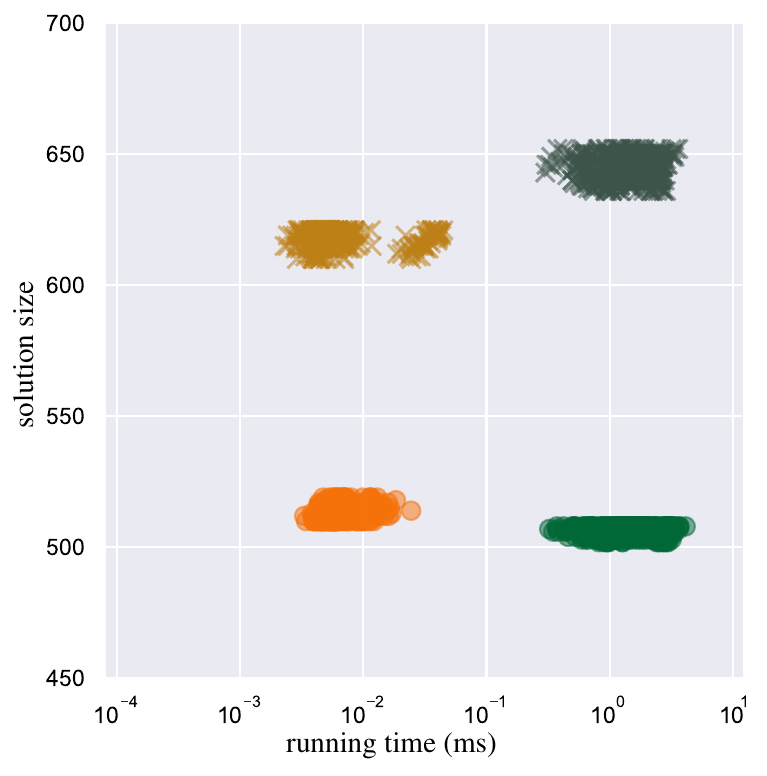}
		\caption{Uniform, $n=10\,000$, $400$ insertions}\label{fig:uni-ins-rect}
	\end{subfigure}%
	\hfill
	\begin{subfigure}[t]{0.5\textwidth}
		\centering
		\includegraphics[width=\textwidth]
		{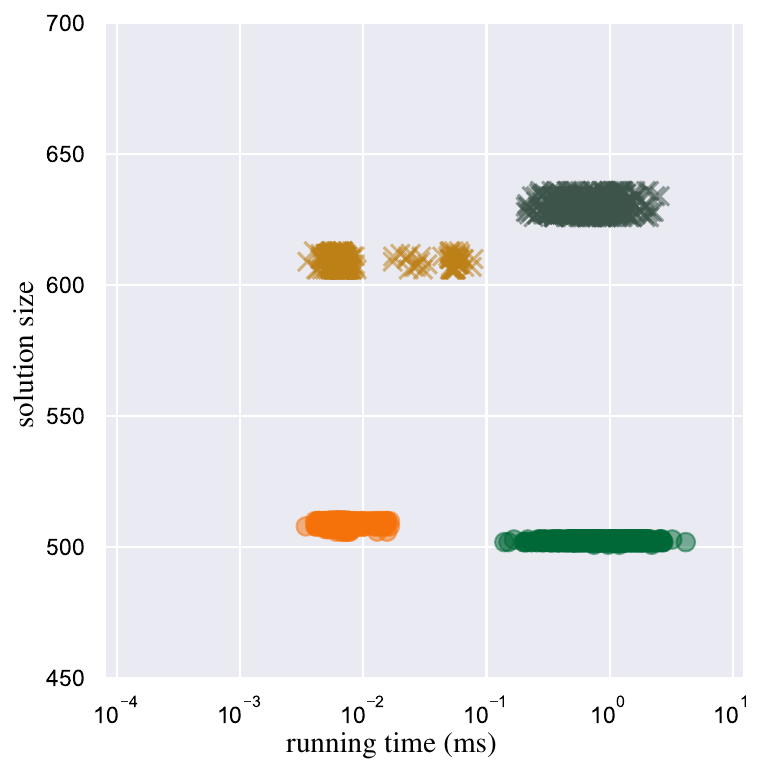}
		\caption{Uniform, $n=10\,000$, $400$ deletions}\label{fig:uni-del-rect}
	\end{subfigure}
	
		\begin{subfigure}[t]{0.5\textwidth}
		\centering
		\includegraphics[width=\textwidth]
		{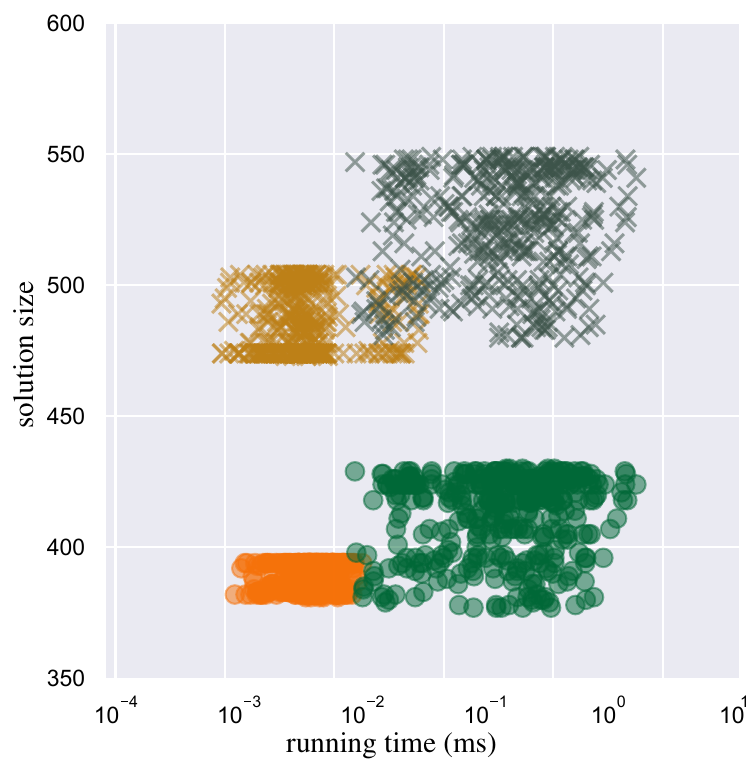}
		\caption{Gaussian, $n=10\,000$, $400$ insertions}\label{fig:gauss-ins:appendix-rect}
	\end{subfigure}%
	\hfill
	\begin{subfigure}[t]{0.5\textwidth}
		\centering
		\includegraphics[width=\textwidth]
		{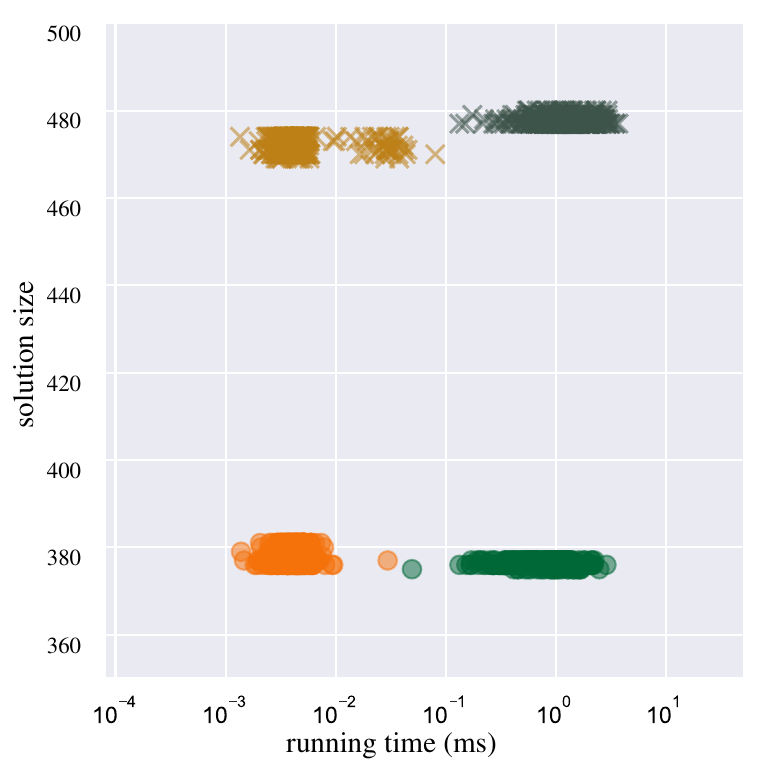}
		\caption{Gaussian, $n=10\,000$, $400$ deletions}\label{fig:gauss-del:appendix-rect}
	\end{subfigure}

\caption{Time-quality scatter plots for synthetic benchmark instances of uniform-height rectangles. The x-axis (log-scale) shows runtime, the y-axis shows the solution size.
}
	\label{fig:plot-unifrom-rect}
\end{figure}

\subparagraph*{\textbf{Optimality Gaps.}}
Next, we explore our approaches on real-world instances and show the results in Figure~\ref{fig:plot-prac-small-rect} and Figure~\ref{fig:plot-prac-large-rect}.
For small instances, we compute \maxset exactly with MaxHS at each round and compare the solution of our presented dynamic approaches with the optimum. 
These plots show that g-\misg\  reaches consistently about $90\%$ of the optimum.
In contrast with the graph-based approaches, algorithms \linea\ and g-\linea\ show a wider range of optimization ratios.
The original approximations are well above their respective worst case ratios, namely around $65\%$ and $80\%$, respectively. The greedy variants of these two approaches push this towards larger solutions with nearly no additional running time.
Note that g-\linea\ reaches between $80\%$ and $85\%$ of the optimum, but faster by one to two orders of magnitude compared to the graph-based approaches. 

Consider the large OSM instances in Figure~\ref{fig:plot-prac-large-rect}.
Here, we also observe a similar pattern as the one from the smaller instance, except that they show larger variances of solution sizes for each algorithm and \misg\ and g-\misg\ get closer to each other in terms of solution size.
One possible reason could be that with the change of input, the ordering of vertices with incremental vertex degrees is not maintained anymore.

\begin{figure}[tbp]	
	\centering
	\begin{subfigure}[t]{0.5\textwidth}
		\centering
		\includegraphics[width=\textwidth]
		{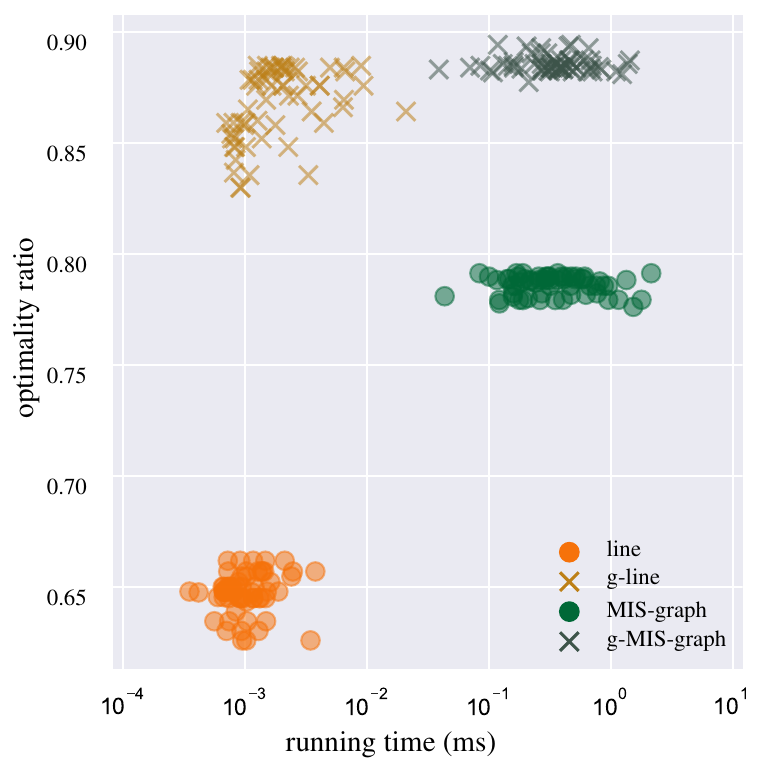}
		\caption{post-CH, 10\% mixed updates}\label{fig:post-ch-rect}
	\end{subfigure}%
	\hfill
	\begin{subfigure}[t]{0.5\textwidth}
		\centering
		\includegraphics[width=\textwidth]
		{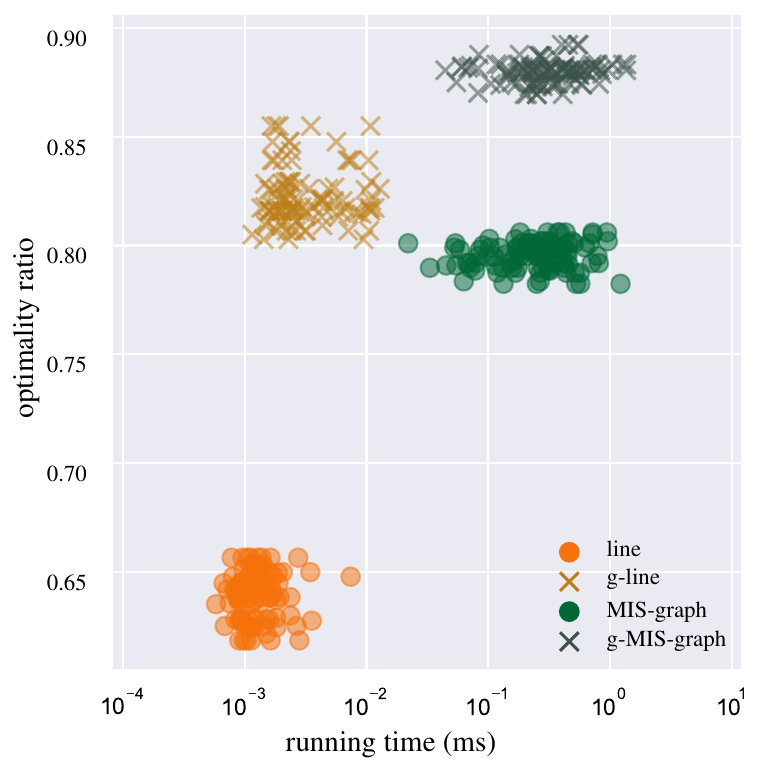}
		\caption{peaks-AT, 10\% mixed updates}\label{fig:peaks-at-rect}
	\end{subfigure}	
	\begin{subfigure}[t]{0.5\textwidth}
		\centering
		\includegraphics[width=\textwidth]
		{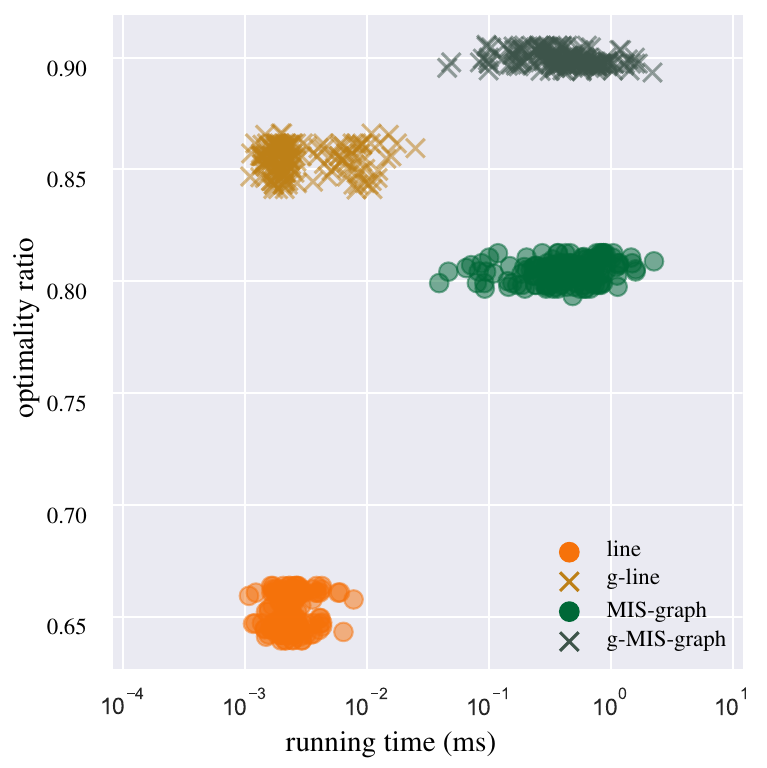}
		\caption{hotels-CH, 10\% mixed updates}
		\label{fig:hotels:appendix-rect}
	\end{subfigure}%
	\hfill
	\begin{subfigure}[t]{0.5\textwidth}
		\centering
		\includegraphics[width=\textwidth]
		{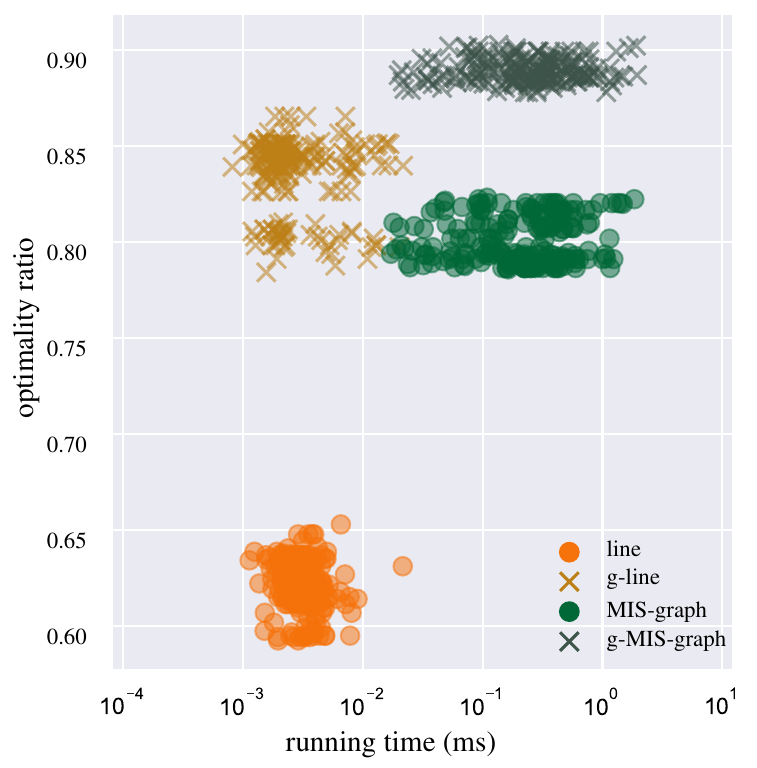}
		\caption{hotels-AT, 10\% mixed updates}\label{fig:hotels-at-rect}
	\end{subfigure}	
	\caption{Time-quality scatter plots for the small OSM uniform-height rectangle instances. The x-axis (log-scale) shows runtime. The y-axis shows the quality ratio compared to an optimal \maxset solution.}
	\label{fig:plot-prac-small-rect}
\end{figure}
\begin{figure}[tbp]	
	\begin{subfigure}[t]{0.5\textwidth}
		\centering
		\includegraphics[width=\textwidth]
		{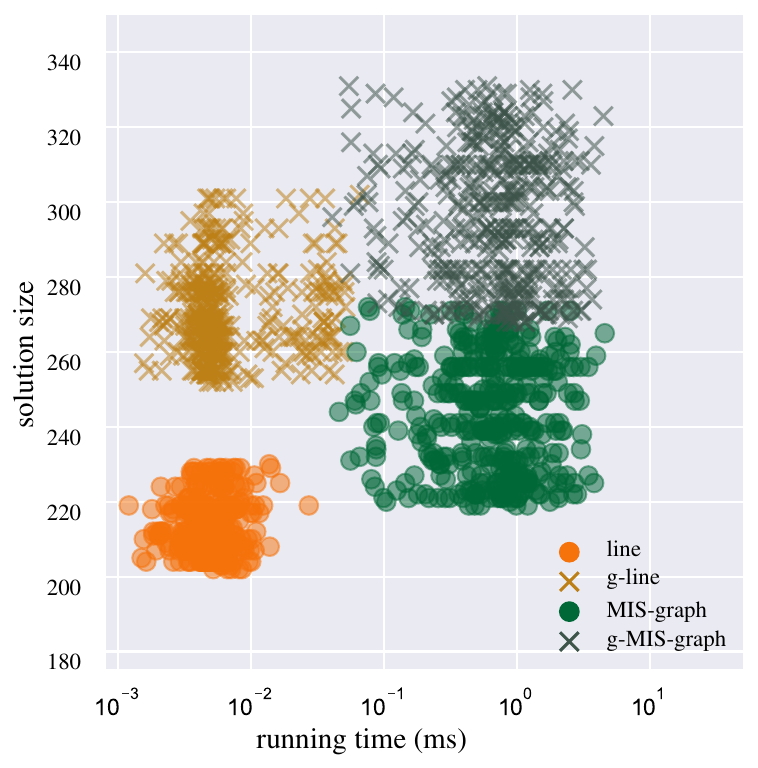}
		\caption{peaks-CH, 10\% mixed updates}\label{fig:peaks-ch-rect}
	\end{subfigure}%
	\hfill
	\begin{subfigure}[t]{0.5\textwidth}
		\centering
		\includegraphics[width=\textwidth]
		{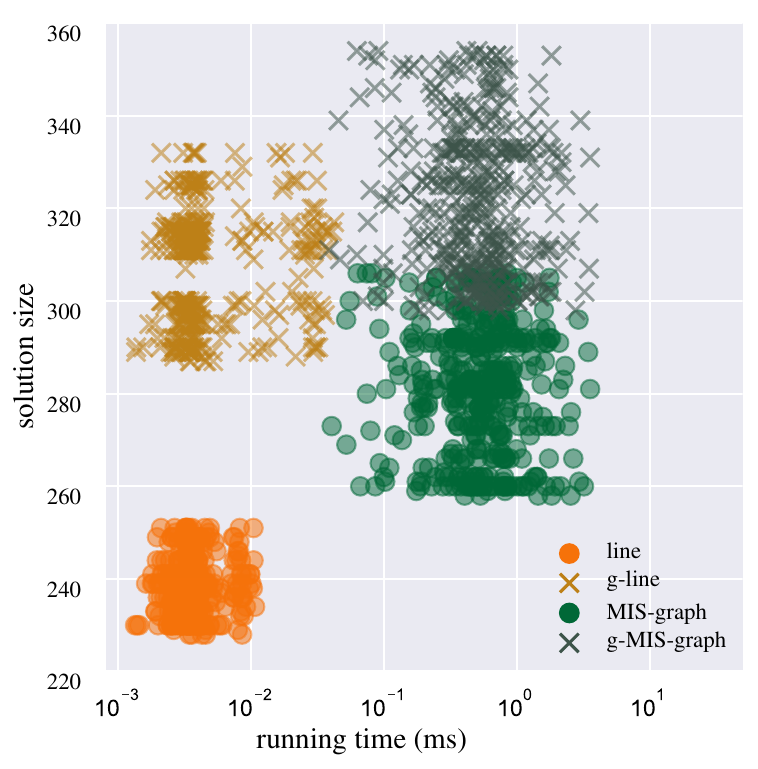}
		\caption{hamlets-CH, 10\% mixed updates}\label{fig:hamlets:appendix-rect}
	\end{subfigure}	
	\caption{Time-quality scatter plots for the large OSM uniform-height rectangle instances. The x-axis (log-scale) shows runtime. The y-axis shows the quality ratio compared to the optimal \maxset solution size.}
	\label{fig:plot-prac-large-rect}
\end{figure}

\begin{figure}[t!]
	\centering
	\begin{subfigure}[t]{0.5\textwidth}
		\centering
		\includegraphics[width=\textwidth]
		{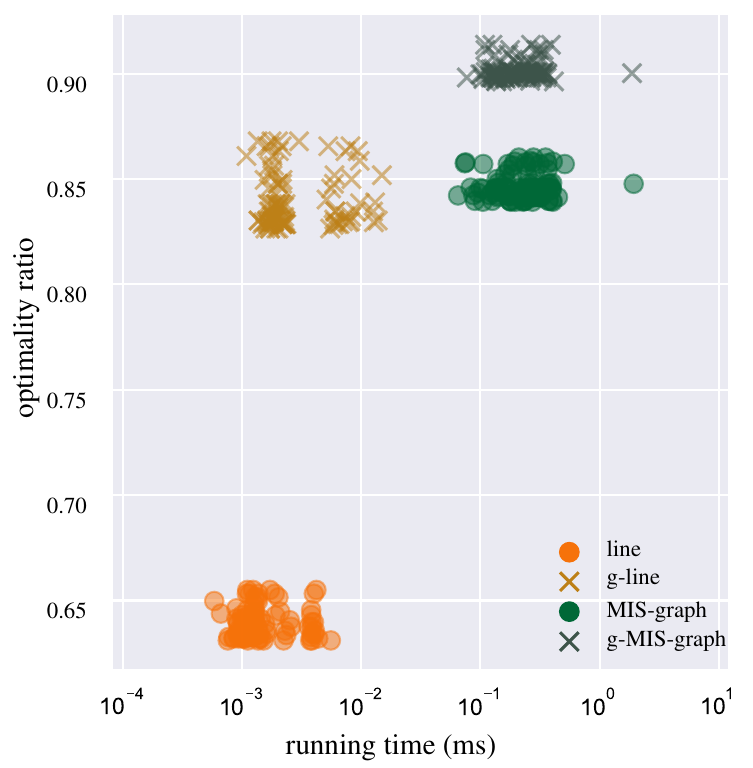}
		\caption{Uniform, $100$ insertions}\label{fig:uni-ins-opt-rect}
	\end{subfigure}%
	\hfill
	\begin{subfigure}[t]{0.5\textwidth}
		\centering
		\includegraphics[width=\textwidth]
		{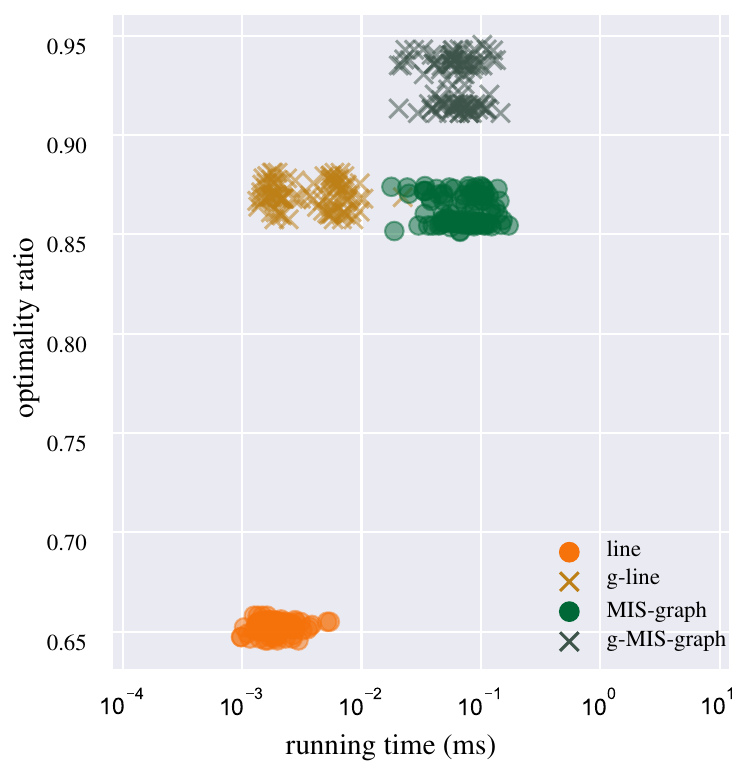}
		\caption{Uniform, $100$ deletions}\label{fig:uni-del-opt-rect}
	\end{subfigure}
	
	\begin{subfigure}[t]{0.5\textwidth}
		\centering
		\includegraphics[width=\textwidth]
		{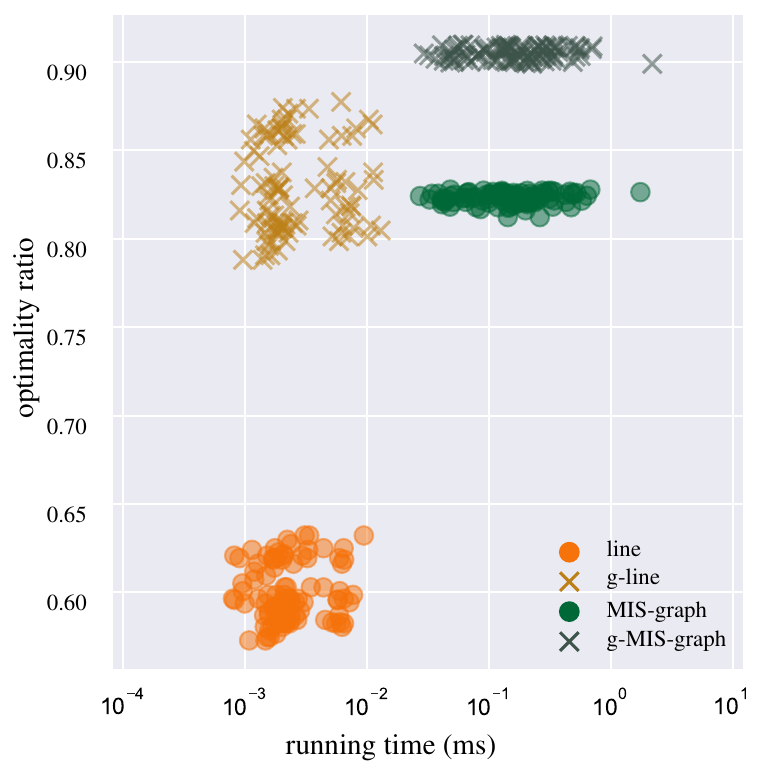}
		\caption{Gaussian, $100$ insertions}\label{fig:gau-ins-opt-rect}
	\end{subfigure}%
	\hfill
	\begin{subfigure}[t]{0.5\textwidth}
		\centering
		\includegraphics[width=\textwidth]
		{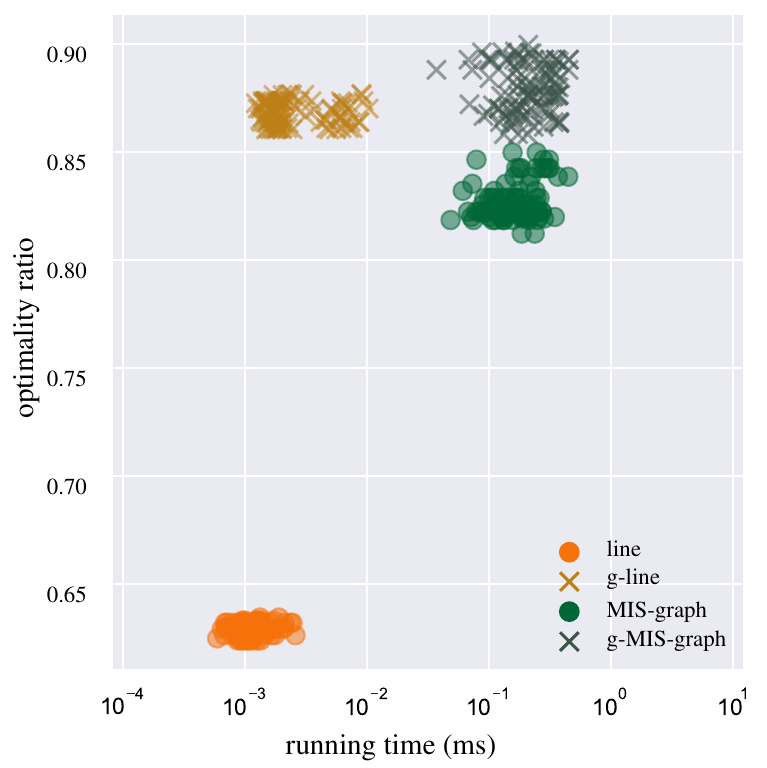}
		\caption{Gaussian, $100$ deletions}\label{fig:gau-del-opt-rect}
	\end{subfigure}

	\caption{Time-quality scatter plots for uniform and Gaussian instances with $n=1\,000$ unit-height rectangles. The x-axis (log-scale) shows runtime. The y-axis shows the quality ratio compared to an optimal \maxset solution.}
	\label{fig:plot-opt-synth-rect}
\end{figure}

Finally, we also compare the solutions of the approximation approaches with the optimum on the small uniform and Gaussian instances with $n=1\,000$ unit-height rectangles; see Figure~\ref{fig:plot-opt-synth-rect}. 
It confirms our observations from the small OSM instances that g-\linea\ provides a very good balance between quality and computation time.

\FloatBarrier

\subparagraph*{\textbf{Runtimes.}}
In the last experiment, we explore the scalability of our presented approaches, both relative to each other and in comparison to the recomputing times.
For each $k\in \{1, 2, 4, 8, 16, 32\}$, we generated 10 random instances with $1000k$ squares and measured the average update time over $100k$ insertion or deletions. 
The results are plotted in Figure~\ref{fig:line-plot-all-gaussian-rect} (for Gaussian instances) and Figure~\ref{fig:line-plot-all-uniform-rect} (for uniform instances).
Note here, since the update procedure of g-\misg\ is exactly the same as the \misg\ approach, so we do not include g-\misg\ explicitly in this set of runtime experiments.

Considering the update time, the plots confirm the observations we had before. 
The g-\linea\ approach is nearly as fast as \linea.
The runtime of  \misg\ shows steeper increase as the instance size increases than the runtime of \linea\ and g-\linea. 

In the comparison with non-dynamic versions, i.e., recomputing the solution after each update, the dynamic approaches show a speed-up by at least one order of magnitude. 

\FloatBarrier
\subparagraph*{\textbf{{Discussion.}}}
Our experimental results for unit-height and arbitrary width rectangles confirm several of our findings obtained in the experiment for unit squares. 
Moreover, all approaches are well above their respective approximation ratios.
A simple greedy variant of the \misg\ approach can significantly boost the solution size and provide the best solution quality. 
Therefore, if the solution size is then priority, then the g-\misg\ approach can be chosen. 
Our greedy augmented version of \linea\ reaches a good balance of solution quality and update time since it is significantly faster than graph-based approaches and computes very good solutions in practice.

\begin{figure}[tbp]
	
	\begin{subfigure}[t]{\textwidth}
		\centering
		\includegraphics[width=.35\textwidth]{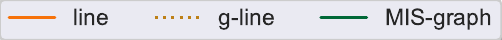}
	\end{subfigure}
	
	\centering
	\begin{subfigure}[t]{0.5\textwidth}
		\centering
		\includegraphics[width=\textwidth]
		{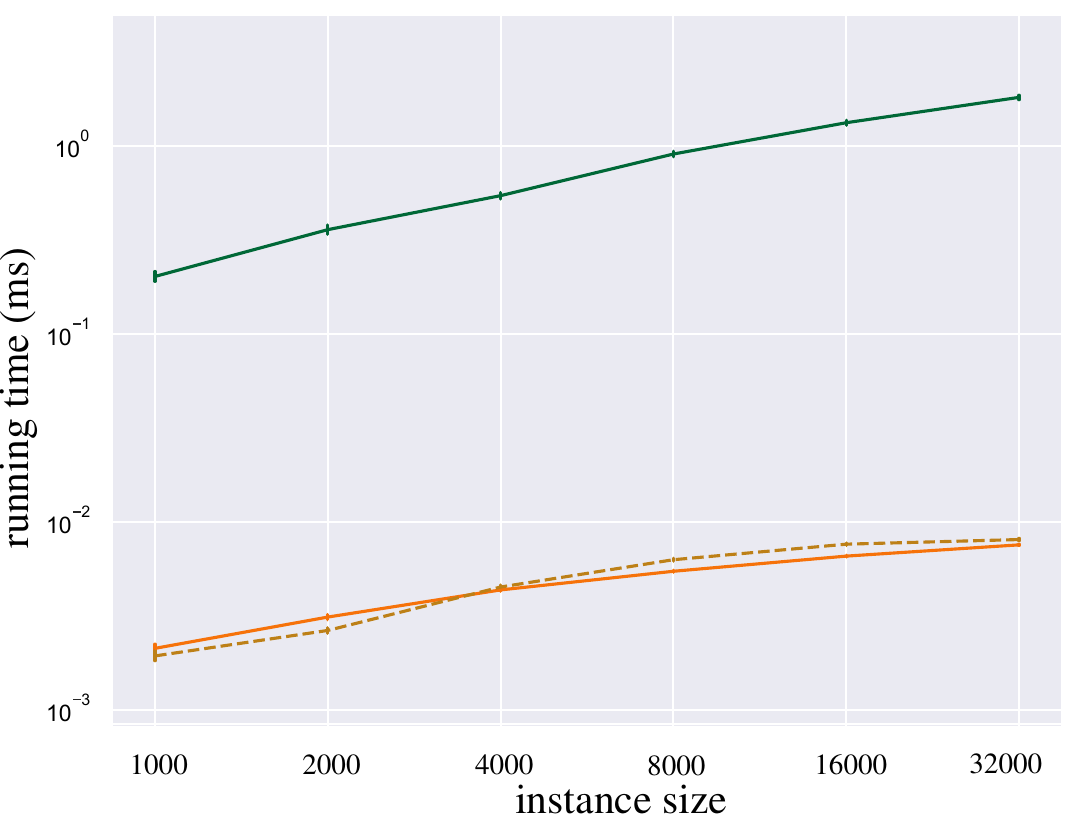}
		\caption{Update times for insertions (Gaussian)}\label{fig:updatetime-ins-gau-rect}
	\end{subfigure}%
	\hfill
	\begin{subfigure}[t]{0.5\textwidth}
		\centering
		\includegraphics[width=\textwidth]
		{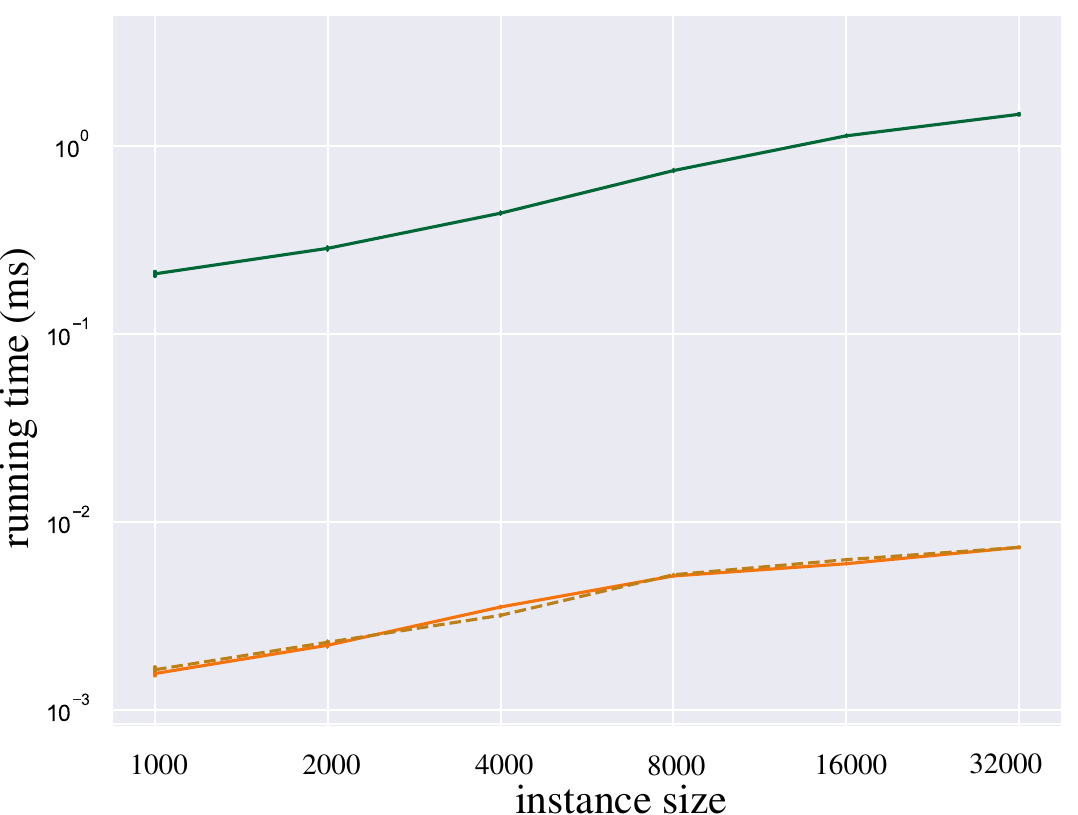}
		\caption{Update times for deletions (Gaussian)}\label{fig:updatetime-del-gau-apx-rect}
	\end{subfigure}
	
	\begin{subfigure}[t]{0.5\textwidth}
		\centering
		\includegraphics[width=\textwidth]
		{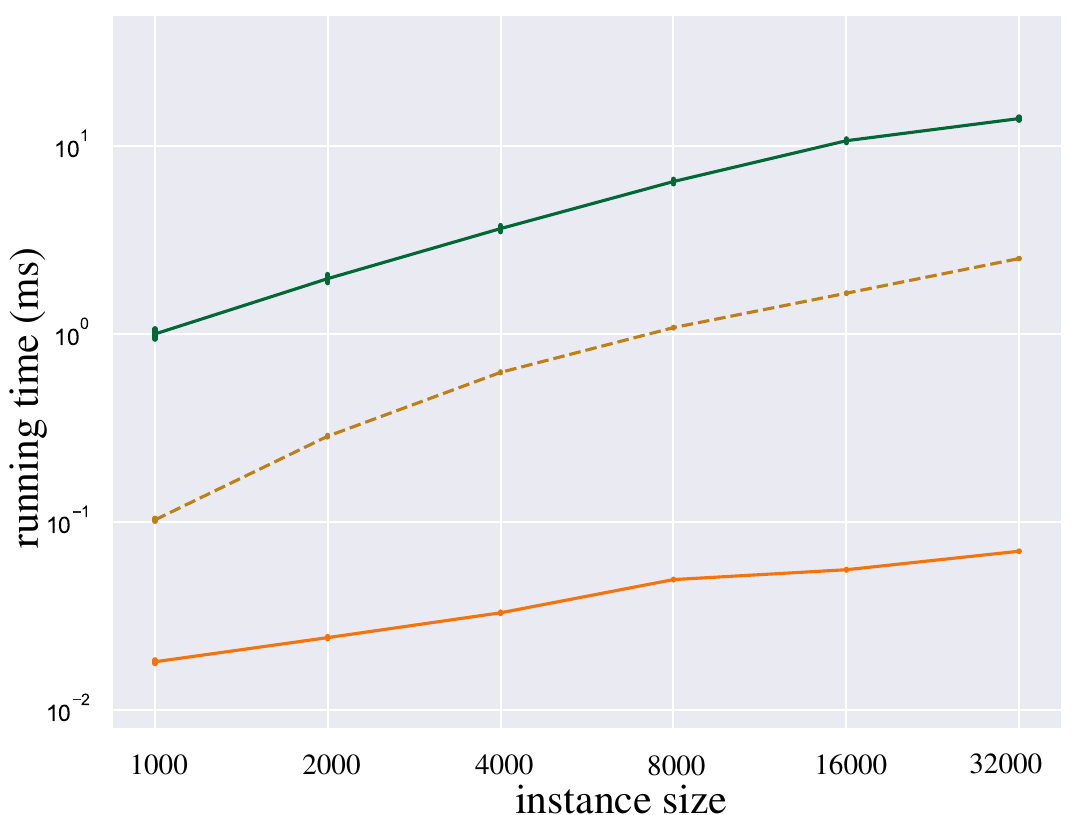}
		\caption{Re-computation times for insertions (Gaussian)}\label{fig:re-time-ins-gau-rect}
	\end{subfigure}%
	\hfill
	\begin{subfigure}[t]{0.5\textwidth}
		\centering
		\includegraphics[width=\textwidth]
		{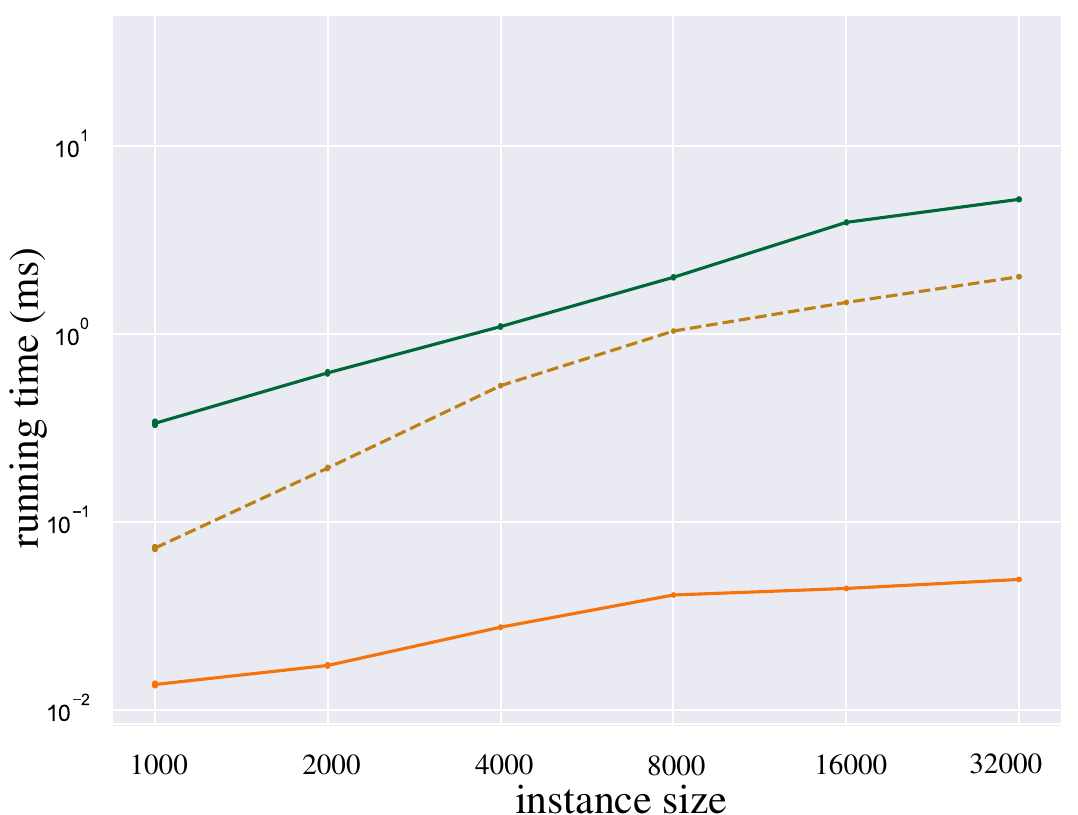}
		\caption{Re-computation times for deletions (Gaussian)}\label{fig:re-time-del-gau-apx-rect}
	\end{subfigure}
	\caption{Log-log runtime plots (notice the different y-offsets) for dynamic updates and re-computation on Gaussian uniform-height rectangle instances of size $n= 1\,000$ to $32\,000$, averaged over $n/10$ updates. Error bars indicate the standard deviation.}
	\label{fig:line-plot-all-gaussian-rect}
\end{figure}

\begin{figure}[tbp]
	
	\begin{subfigure}[t]{\textwidth}
		\centering
		\includegraphics[width=.35\textwidth]{figs_ESA_dissertation/E5-Rect/legend.pdf}
	\end{subfigure}
	
	\centering
	\begin{subfigure}[t]{0.5\textwidth}
		\centering
		\includegraphics[width=\textwidth]
		{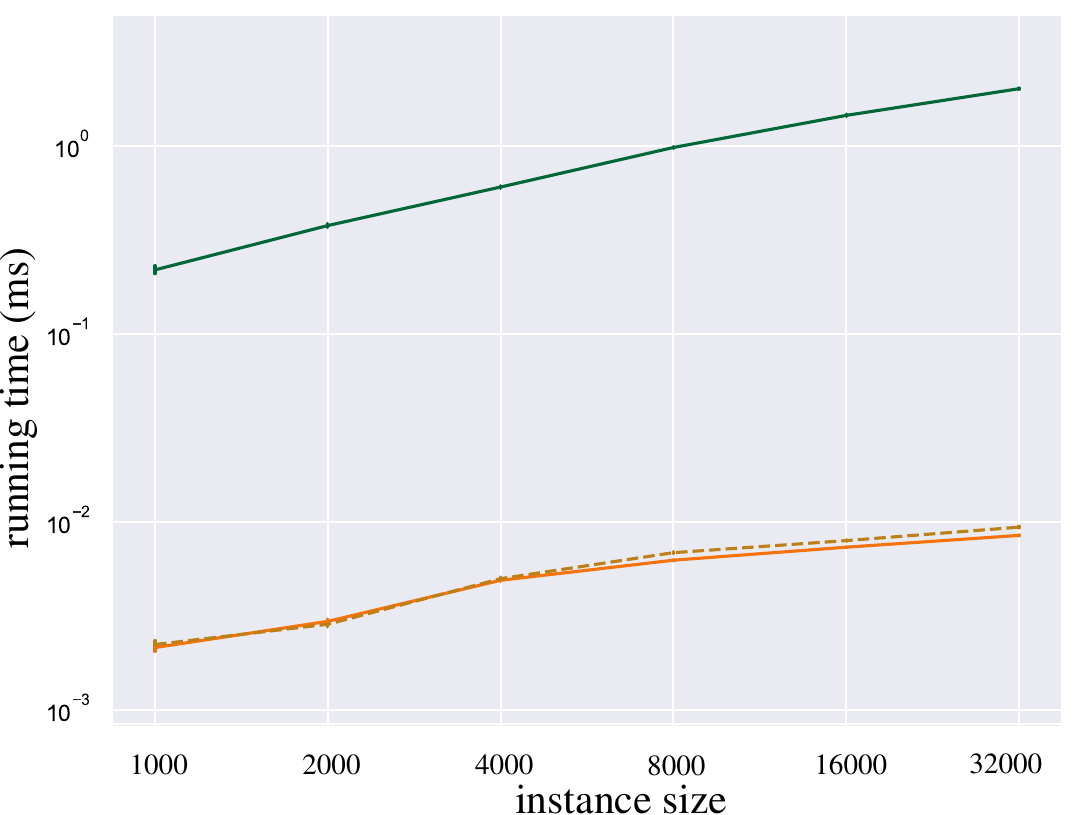}
		\caption{Update times for insertions (uniform)}\label{fig:updatetime-ins-rect}
	\end{subfigure}%
	\hfill
	\begin{subfigure}[t]{0.5\textwidth}
		\centering
		\includegraphics[width=\textwidth]
		{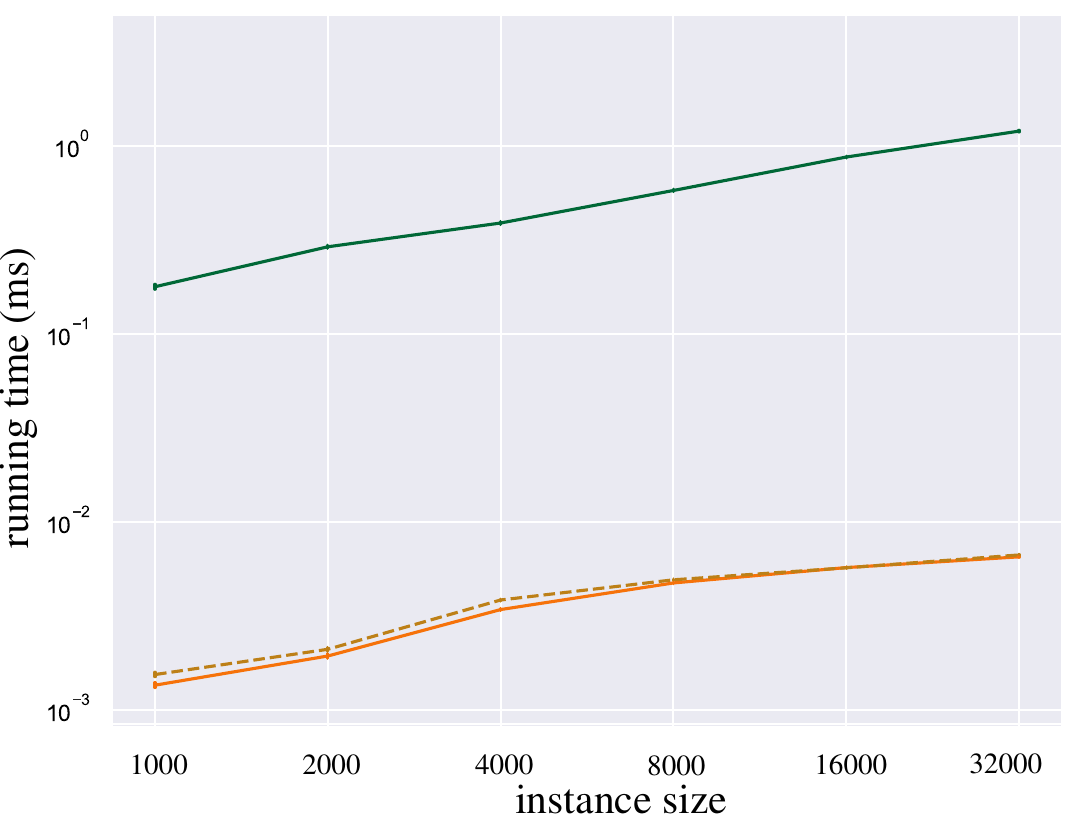}
		\caption{Update times for deletions (uniform)}\label{fig:updatetime-del-rect}
	\end{subfigure}
	
	\begin{subfigure}[t]{0.5\textwidth}
		\centering
		\includegraphics[width=\textwidth]
		{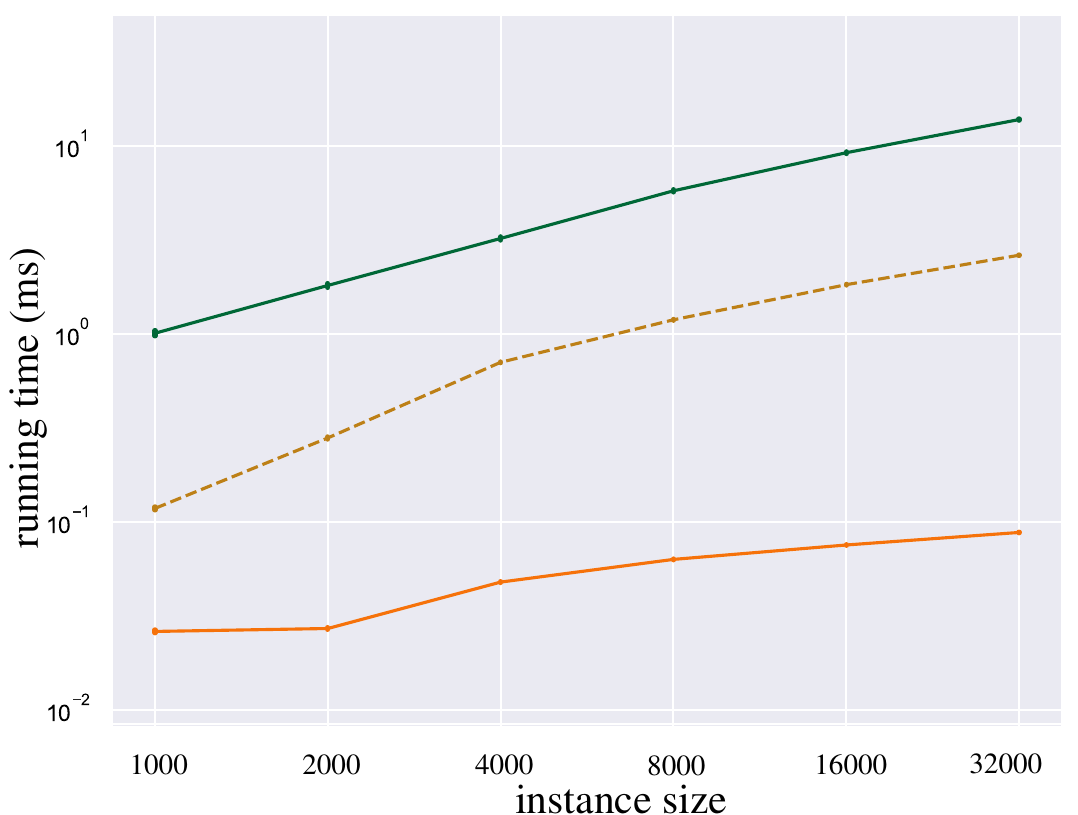}
		\caption{Re-computation times for insertions (uniform)}\label{fig:re-time-ins-rect}
	\end{subfigure}%
	\hfill
	\begin{subfigure}[t]{0.5\textwidth}
		\centering
		\includegraphics[width=\textwidth]
		{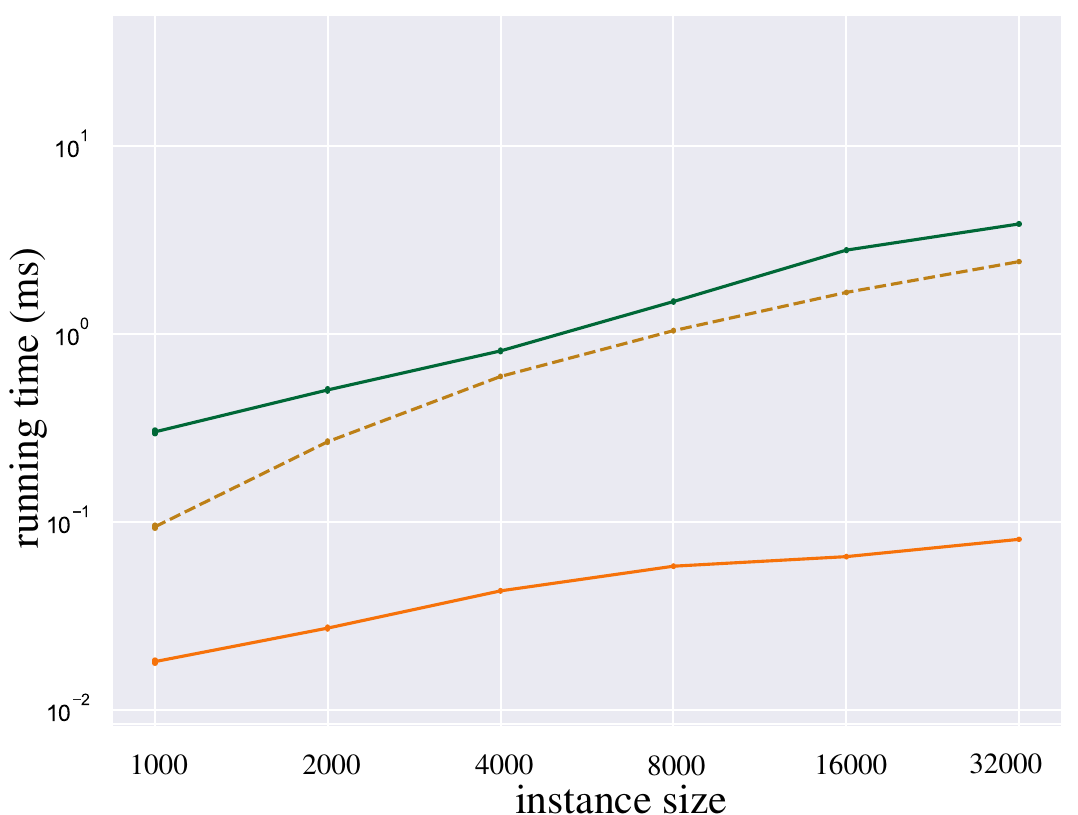}
		\caption{Re-computation times for deletions (uniform)}\label{fig:re-time-del-rect}
	\end{subfigure}
	\caption{Log-log runtime plots (notice the different y-offsets) for dynamic updates and re-computation on uniform-height rectangle instances of size $n= 1\,000$ to $32\,000$, averaged over $n/10$ updates. Error bars indicate the standard deviation.}
	\label{fig:line-plot-all-uniform-rect}
\end{figure}

\section{Conclusions}
We investigated the \indset and \maxset problems on dynamic sets of uniform rectangles and uniform-height rectangles from an algorithm engineering perspective, providing both theoretical results for maintaining a \indset or an approximate \maxset and reporting insights from an experimental study. 
Open problems for future work include (i) finding %
 \maxset sublinear-update-time approximation algorithms for dynamic unit squares with approximation ratio better than $2$, (ii) studying similar questions for dynamic disk graphs, and (iii) implementing improvements such as a faster dynamic range searching data structure to speed-up our algorithm \misr.
 Moreover, it would be interesting to design dynamic approximation schemes for \maxset that maintain 
 stability in a solution.
\clearpage

\bibliographystyle{plainurl}%
\bibliography{MIS.bib}

\begin{thebibliography}{10}

\bibitem{abboud2019dynamic}
Amir Abboud, Raghavendra Addanki, Fabrizio Grandoni, Debmalya Panigrahi, and
  Barna Saha.
\newblock Dynamic set cover: improved algorithms and lower bounds.
\newblock In Moses Charikar and Edith Cohen, editors, {\em Proceedings of the
  51st Annual {ACM} {SIGACT} Symposium on Theory of Computing, {STOC} 2019,
  Phoenix, AZ, USA, June 23-26, 2019}, pages 114--125. {ACM}, 2019.
\newblock \href {https://doi.org/10.1145/3313276.3316376}
  {\path{doi:10.1145/3313276.3316376}}.

\bibitem{aw-asmwir-13}
Anna Adamaszek and Andreas Wiese.
\newblock Approximation schemes for maximum weight independent set of
  rectangles.
\newblock In {\em 54th Annual {IEEE} Symposium on Foundations of Computer
  Science, {FOCS} 2013, 26-29 October, 2013, Berkeley, CA, {USA}}, pages
  400--409. {IEEE} Computer Society, 2013.
\newblock \href {https://doi.org/10.1109/FOCS.2013.50}
  {\path{doi:10.1109/FOCS.2013.50}}.

\bibitem{agarwal1998label}
Pankaj~K Agarwal, Marc Van~Kreveld, and Subhash Suri.
\newblock Label placement by maximum independent set in rectangles.
\newblock {\em Comput. Geom. Theory Appl.}, 11(3-4):209--218, 1998.
\newblock \href {https://doi.org/10.1016/S0925-7721(98)00028-5}
  {\path{doi:10.1016/S0925-7721(98)00028-5}}.

\bibitem{DBLP:conf/stoc/AssadiOSS18}
Sepehr Assadi, Krzysztof Onak, Baruch Schieber, and Shay Solomon.
\newblock Fully dynamic maximal independent set with sublinear in n update
  time.
\newblock In Timothy~M. Chan, editor, {\em Proceedings of the Thirtieth Annual
  {ACM-SIAM} Symposium on Discrete Algorithms, {SODA} 2019, San Diego,
  California, USA, January 6-9, 2019}, pages 1919--1936. {SIAM}, 2019.
\newblock \href {https://doi.org/10.1137/1.9781611975482.116}
  {\path{doi:10.1137/1.9781611975482.116}}.

\bibitem{assadi2019fully}
Sepehr Assadi, Krzysztof Onak, Baruch Schieber, and Shay Solomon.
\newblock Fully dynamic maximal independent set with sublinear in n update
  time.
\newblock In Timothy~M. Chan, editor, {\em Proceedings of the Thirtieth Annual
  {ACM-SIAM} Symposium on Discrete Algorithms, {SODA} 2019, San Diego,
  California, USA, January 6-9, 2019}, pages 1919--1936. {SIAM}, 2019.
\newblock \href {https://doi.org/10.1137/1.9781611975482.116}
  {\path{doi:10.1137/1.9781611975482.116}}.

\bibitem{DBLP:journals/tvcg/BeenDY06}
Ken Been, Eli Daiches, and Chee{-}Keng Yap.
\newblock Dynamic map labeling.
\newblock {\em {IEEE} Trans. Vis. Comput. Graph.}, 12(5):773--780, 2006.
\newblock \href {https://doi.org/10.1109/TVCG.2006.136}
  {\path{doi:10.1109/TVCG.2006.136}}.

\bibitem{bnpw-oarcd-10}
Ken Been, Martin Nöllenburg, Sheung-Hung Poon, and Alexander Wolff.
\newblock Optimizing active ranges for consistent dynamic map labeling.
\newblock {\em Comput. Geom. Theory Appl.}, 43(3):312--328, 2010.
\newblock \href {https://doi.org/10.1016/j.comgeo.2009.03.006}
  {\path{doi:10.1016/j.comgeo.2009.03.006}}.

\bibitem{DBLP:journals/comgeo/BeenNPW10}
Ken Been, Martin N{\"{o}}llenburg, Sheung{-}Hung Poon, and Alexander Wolff.
\newblock Optimizing active ranges for consistent dynamic map labeling.
\newblock {\em Comput. Geom.}, 43(3):312--328, 2010.
\newblock \href {https://doi.org/10.1016/j.comgeo.2009.03.006}
  {\path{doi:10.1016/j.comgeo.2009.03.006}}.

\bibitem{behnezhad2019fully}
Soheil Behnezhad, Mahsa Derakhshan, MohammadTaghi Hajiaghayi, Cliff Stein, and
  Madhu Sudan.
\newblock Fully dynamic maximal independent set with polylogarithmic update
  time.
\newblock In David Zuckerman, editor, {\em 60th {IEEE} Annual Symposium on
  Foundations of Computer Science, {FOCS} 2019, Baltimore, Maryland, USA,
  November 9-12, 2019}, pages 382--405. {IEEE} Computer Society, 2019.
\newblock \href {https://doi.org/10.1109/FOCS.2019.00032}
  {\path{doi:10.1109/FOCS.2019.00032}}.

\bibitem{bernstein2019deamortization}
Aaron Bernstein, Sebastian Forster, and Monika Henzinger.
\newblock A deamortization approach for dynamic spanner and dynamic maximal
  matching.
\newblock In Timothy~M. Chan, editor, {\em Proceedings of the Thirtieth Annual
  {ACM-SIAM} Symposium on Discrete Algorithms, {SODA} 2019, San Diego,
  California, USA, January 6-9, 2019}, pages 1899--1918. {SIAM}, 2019.
\newblock \href {https://doi.org/10.1137/1.9781611975482.115}
  {\path{doi:10.1137/1.9781611975482.115}}.

\bibitem{bhattacharya2017deterministic}
Sayan Bhattacharya, Deeparnab Chakrabarty, and Monika Henzinger.
\newblock Deterministic fully dynamic approximate vertex cover and fractional
  matching in {O(1)} amortized update time.
\newblock In Friedrich Eisenbrand and Jochen K{\"{o}}nemann, editors, {\em
  Integer Programming and Combinatorial Optimization - 19th International
  Conference, {IPCO} 2017, Waterloo, ON, Canada, June 26-28, 2017,
  Proceedings}, volume 10328 of {\em Lecture Notes in Computer Science}, pages
  86--98. Springer, 2017.
\newblock \href {https://doi.org/10.1007/978-3-319-59250-3\_8}
  {\path{doi:10.1007/978-3-319-59250-3\_8}}.

\bibitem{bhattacharya2018dynamic}
Sayan Bhattacharya, Deeparnab Chakrabarty, Monika Henzinger, and Danupon
  Nanongkai.
\newblock Dynamic algorithms for graph coloring.
\newblock In Artur Czumaj, editor, {\em Proceedings of the Twenty-Ninth Annual
  {ACM-SIAM} Symposium on Discrete Algorithms, {SODA} 2018, New Orleans, LA,
  USA, January 7-10, 2018}, pages 1--20. {SIAM}, 2018.
\newblock \href {https://doi.org/10.1137/1.9781611975031.1}
  {\path{doi:10.1137/1.9781611975031.1}}.

\bibitem{DBLP:journals/corr/abs-2007-08643}
Sujoy Bhore, Jean Cardinal, John Iacono, and Grigorios Koumoutsos.
\newblock Dynamic geometric independent set.
\newblock {\em CoRR}, abs/2007.08643, 2020.
\newblock URL: \url{https://arxiv.org/abs/2007.08643}, \href
  {http://arxiv.org/abs/2007.08643} {\path{arXiv:2007.08643}}.

\bibitem{chalermsook2009maximum}
Parinya Chalermsook and Julia Chuzhoy.
\newblock Maximum independent set of rectangles.
\newblock In Claire Mathieu, editor, {\em Proceedings of the Twentieth Annual
  {ACM-SIAM} Symposium on Discrete Algorithms, {SODA} 2009, New York, NY, USA,
  January 4-6, 2009}, pages 892--901. {SIAM}, 2009.
\newblock \href {https://doi.org/10.1137/1.9781611973068.97}
  {\path{doi:10.1137/1.9781611973068.97}}.

\bibitem{chan2012approximation}
Timothy~M Chan and Sariel Har-Peled.
\newblock Approximation algorithms for maximum independent set of pseudo-disks.
\newblock {\em Discrete \& Computational Geometry}, 48(2):373--392, 2012.
\newblock \href {https://doi.org/10.1007/s00454-012-9417-5}
  {\path{doi:10.1007/s00454-012-9417-5}}.

\bibitem{ct-dorsr-17}
Timothy~M. Chan and Konstantinos Tsakalidis.
\newblock Dynamic orthogonal range searching on the {RAM}, revisited.
\newblock In Boris Aronov and Matthew~J. Katz, editors, {\em Computational
  Geometry (SoCG'17)}, volume~77 of {\em LIPIcs}, pages 28:1--28:13. Schloss
  Dagstuhl -- Leibniz-Zentrum für Informatik, 2017.
\newblock \href {https://doi.org/10.4230/LIPIcs.SoCG.2017.28}
  {\path{doi:10.4230/LIPIcs.SoCG.2017.28}}.

\bibitem{chechik2019fully}
Shiri Chechik and Tianyi Zhang.
\newblock Fully dynamic maximal independent set in expected poly-log update
  time.
\newblock In David Zuckerman, editor, {\em 60th {IEEE} Annual Symposium on
  Foundations of Computer Science, {FOCS} 2019, Baltimore, Maryland, USA,
  November 9-12, 2019}, pages 370--381. {IEEE} Computer Society, 2019.
\newblock \href {https://doi.org/10.1109/FOCS.2019.00031}
  {\path{doi:10.1109/FOCS.2019.00031}}.

\bibitem{DBLP:journals/tog/ChristensenMS95}
Jon Christensen, Joe Marks, and Stuart~M. Shieber.
\newblock An empirical study of algorithms for point-feature label placement.
\newblock {\em {ACM} Trans. Graph.}, 14(3):203--232, 1995.
\newblock \href {https://doi.org/10.1145/212332.212334}
  {\path{doi:10.1145/212332.212334}}.

\bibitem{ce-amir-16}
Julia Chuzhoy and Alina Ene.
\newblock On approximating maximum independent set of rectangles.
\newblock In Irit Dinur, editor, {\em {IEEE} 57th Annual Symposium on
  Foundations of Computer Science, {FOCS} 2016, 9-11 October 2016, Hyatt
  Regency, New Brunswick, New Jersey, {USA}}, pages 820--829. {IEEE} Computer
  Society, 2016.
\newblock \href {https://doi.org/10.1109/FOCS.2016.92}
  {\path{doi:10.1109/FOCS.2016.92}}.

\bibitem{DBLP:conf/icalp/CormodeDK19}
Graham Cormode, Jacques Dark, and Christian Konrad.
\newblock Independent sets in vertex-arrival streams.
\newblock In Christel Baier, Ioannis Chatzigiannakis, Paola Flocchini, and
  Stefano Leonardi, editors, {\em International Colloquium on Automata,
  Languages, and Programming (ICALP'19)}, volume 132 of {\em LIPIcs}, pages
  45:1--45:14. Schloss Dagstuhl - Leibniz-Zentrum fuer Informatik, 2019.
\newblock \href {https://doi.org/10.4230/LIPIcs.ICALP.2019.45}
  {\path{doi:10.4230/LIPIcs.ICALP.2019.45}}.

\bibitem{DBLP:conf/esa/BergG13}
Mark de~Berg and Dirk H.~P. Gerrits.
\newblock Labeling moving points with a trade-off between label speed and label
  overlap.
\newblock In Hans~L. Bodlaender and Giuseppe~F. Italiano, editors, {\em
  Algorithms - {ESA} 2013 - 21st Annual European Symposium, Sophia Antipolis,
  France, September 2-4, 2013. Proceedings}, volume 8125 of {\em Lecture Notes
  in Computer Science}, pages 373--384. Springer, 2013.
\newblock \href {https://doi.org/10.1007/978-3-642-40450-4\_32}
  {\path{doi:10.1007/978-3-642-40450-4\_32}}.

\bibitem{DBLP:journals/vlc/NascimentoE08}
Hugo A.~D. do~Nascimento and Peter Eades.
\newblock User hints for map labeling.
\newblock {\em J. Vis. Lang. Comput.}, 19(1):39--74, 2008.
\newblock \href {https://doi.org/10.1016/j.jvlc.2006.03.004}
  {\path{doi:10.1016/j.jvlc.2006.03.004}}.

\bibitem{EppGalIta-ATCH-99}
David Eppstein, Zvi Galil, and Giuseppe~F. Italiano.
\newblock {Dynamic graph algorithms}.
\newblock In Mikhail~J. Atallah, editor, {\em Algorithms and Theory of
  Computation Handbook}, chapter~8. CRC Press, 1999.
\newblock \href {https://doi.org/10.1.1.43.8372} {\path{doi:10.1.1.43.8372}}.

\bibitem{erlebach2005polynomial}
Thomas Erlebach, Klaus Jansen, and Eike Seidel.
\newblock Polynomial-time approximation schemes for geometric intersection
  graphs.
\newblock {\em SIAM Journal on Computing}, 34(6):1302--1323, 2005.
\newblock \href {https://doi.org/10.1137/s0097539702402676}
  {\path{doi:10.1137/s0097539702402676}}.

\bibitem{fw-ppwalm-91}
Michael Formann and Frank Wagner.
\newblock A packing problem with applications to lettering of maps.
\newblock In Robert L.~Scot Drysdale, editor, {\em Proceedings of the Seventh
  Annual Symposium on Computational Geometry, North Conway, NH, USA, , June
  10-12, 1991}, pages 281--288. {ACM}, 1991.
\newblock \href {https://doi.org/10.1145/109648.109680}
  {\path{doi:10.1145/109648.109680}}.

\bibitem{DBLP:conf/compgeom/FormannW91}
Michael Formann and Frank Wagner.
\newblock A packing problem with applications to lettering of maps.
\newblock In Robert L.~Scot Drysdale, editor, {\em Proceedings of the Seventh
  Annual Symposium on Computational Geometry, North Conway, NH, USA, , June
  10-12, 1991}, pages 281--288. {ACM}, 1991.
\newblock \href {https://doi.org/10.1145/109648.109680}
  {\path{doi:10.1145/109648.109680}}.

\bibitem{DBLP:journals/ipl/FowlerPT81}
Robert~J. Fowler, Mike Paterson, and Steven~L. Tanimoto.
\newblock Optimal packing and covering in the plane are \textsf{NP}-complete.
\newblock {\em Inf. Process. Lett.}, 12(3):133--137, 1981.
\newblock \href {https://doi.org/10.1016/0020-0190(81)90111-3}
  {\path{doi:10.1016/0020-0190(81)90111-3}}.

\bibitem{Gabriel2015}
Edith Gabriel.
\newblock Spatio-temporal point pattern analysis and modeling.
\newblock In Shashi Shekhar, Hui Xiong, and Xun Zhou, editors, {\em
  Encyclopedia of GIS}, pages 1--8. Springer, 2015.
\newblock \href {https://doi.org/10.1007/978-3-319-23519-6_1646-1}
  {\path{doi:10.1007/978-3-319-23519-6_1646-1}}.

\bibitem{DBLP:journals/corr/abs-2106-00623}
Waldo G{\'{a}}lvez, Arindam Khan, Mathieu Mari, Tobias M{\"{o}}mke,
  Madhusudhan~Reddy Pittu, and Andreas Wiese.
\newblock A 4-approximation algorithm for maximum independent set of
  rectangles.
\newblock {\em CoRR}, abs/2106.00623, 2021.
\newblock URL: \url{https://arxiv.org/abs/2106.00623}, \href
  {http://arxiv.org/abs/2106.00623} {\path{arXiv:2106.00623}}.

\bibitem{gamlath2019online}
Buddhima Gamlath, Michael Kapralov, Andreas Maggiori, Ola Svensson, and David
  Wajc.
\newblock Online matching with general arrivals.
\newblock In David Zuckerman, editor, {\em 60th {IEEE} Annual Symposium on
  Foundations of Computer Science, {FOCS} 2019, Baltimore, Maryland, USA,
  November 9-12, 2019}, pages 26--37. {IEEE} Computer Society, 2019.
\newblock \href {https://doi.org/10.1109/FOCS.2019.00011}
  {\path{doi:10.1109/FOCS.2019.00011}}.

\bibitem{gavruskin2015dynamic}
Alexander Gavruskin, Bakhadyr Khoussainov, Mikhail Kokho, and Jiamou Liu.
\newblock Dynamic algorithms for monotonic interval scheduling problem.
\newblock {\em Theoretical Computer Science}, 562:227--242, 2015.
\newblock \href {https://doi.org/10.1016/j.tcs.2014.09.046}
  {\path{doi:10.1016/j.tcs.2014.09.046}}.

\bibitem{gnr-clrm-16}
Andreas Gemsa, Martin Nöllenburg, and Ignaz Rutter.
\newblock Consistent labeling of rotating maps.
\newblock {\em J. Computational Geometry}, 7(1):308--331, 2016.
\newblock \href {https://doi.org/10.20382/jocg.v7i1a15}
  {\path{doi:10.20382/jocg.v7i1a15}}.

\bibitem{DBLP:journals/jea/GemsaNR16}
Andreas Gemsa, Martin N{\"{o}}llenburg, and Ignaz Rutter.
\newblock Evaluation of labeling strategies for rotating maps.
\newblock {\em {ACM} J. Exp. Algorithmics}, 21(1):1.4:1--1.4:21, 2016.
\newblock \href {https://doi.org/10.1145/2851493} {\path{doi:10.1145/2851493}}.

\bibitem{DBLP:journals/networks/GuptaLL82}
U.~I. Gupta, D.~T. Lee, and Joseph~Y.{-}T. Leung.
\newblock Efficient algorithms for interval graphs and circular-arc graphs.
\newblock {\em Networks}, 12(4):459--467, 1982.
\newblock \href {https://doi.org/10.1002/net.3230120410}
  {\path{doi:10.1002/net.3230120410}}.

\bibitem{henzinger_et_al:LIPIcs:2020:12209}
Monika Henzinger, Stefan Neumann, and Andreas Wiese.
\newblock Dynamic approximate maximum independent set of intervals, hypercubes
  and hyperrectangles.
\newblock In Sergio Cabello and Danny~Z. Chen, editors, {\em Symposium on
  Computational Geometry (SoCG 2020)}, volume 164 of {\em LIPIcs}, pages
  51:1--51:14. Schloss Dagstuhl--Leibniz-Zentrum f{\"u}r Informatik, 2020.
\newblock \href {https://doi.org/10.4230/LIPIcs.SoCG.2020.51}
  {\path{doi:10.4230/LIPIcs.SoCG.2020.51}}.

\bibitem{hm-avlsi-85}
Dorit~S. Hochbaum and Wolfgang Maass.
\newblock Approximation schemes for covering and packing problems in image
  processing and vlsi.
\newblock {\em J. ACM}, 32(1):130--136, 1985.
\newblock \href {https://doi.org/10.1145/2455.214106}
  {\path{doi:10.1145/2455.214106}}.

\bibitem{hopcroft1973n}
John~E Hopcroft and Richard~M Karp.
\newblock An $n^{5/2}$ algorithm for maximum matchings in bipartite graphs.
\newblock {\em SIAM Journal on Computing}, 2(4):225--231, 1973.
\newblock \href {https://doi.org/10.1137/0202019} {\path{doi:10.1137/0202019}}.

\bibitem{k-racp-72}
Richard~M. Karp.
\newblock Reducibility among combinatorial problems.
\newblock In R.~E. Miller, J.~W. Thatcher, and J.~D. Bohlinger, editors, {\em
  Complexity of Computer Computations}, pages 85--103, 1972.
\newblock \href {https://doi.org/10.1007/978-1-4684-2001-2_9}
  {\path{doi:10.1007/978-1-4684-2001-2_9}}.

\bibitem{kllns-esl-19}
Fabian Klute, Guangping Li, Raphael L{\"{o}}ffler, Martin N{\"{o}}llenburg, and
  Manuela Schmidt.
\newblock Exploring semi-automatic map labeling.
\newblock In Farnoush~Banaei Kashani, Goce Trajcevski, Ralf~Hartmut
  G{\"{u}}ting, Lars Kulik, and Shawn~D. Newsam, editors, {\em Proceedings of
  the 27th {ACM} {SIGSPATIAL} International Conference on Advances in
  Geographic Information Systems, {SIGSPATIAL} 2019, Chicago, IL, USA, November
  5-8, 2019}, pages 13--22. {ACM}, 2019.
\newblock \href {https://doi.org/10.1145/3347146.3359359}
  {\path{doi:10.1145/3347146.3359359}}.

\bibitem{linial1987distributive}
Nathan Linial.
\newblock Distributive graph algorithms-global solutions from local data.
\newblock In {\em 28th Annual Symposium on Foundations of Computer Science
  {(FOCS} 1987), Los Angeles, California, USA, 27-29 October 1987}, pages
  331--335. {IEEE} Computer Society, 1987.
\newblock \href {https://doi.org/10.1109/SFCS.1987.20}
  {\path{doi:10.1109/SFCS.1987.20}}.

\bibitem{DBLP:conf/ieeevast/MacEachrenJRPSMZB11}
Alan~M. MacEachren, Anuj~R. Jaiswal, Anthony~C. Robinson, Scott Pezanowski,
  Alexander Savelyev, Prasenjit Mitra, Xiao Zhang, and Justine~I. Blanford.
\newblock Senseplace2: Geotwitter analytics support for situational awareness.
\newblock In {\em Visual Analytics Science and Technology (VAST'11)}, pages
  181--190. {IEEE}, 2011.

\bibitem{mn-dfc-90}
Kurt Mehlhorn and Stefan Näher.
\newblock Dynamic fractional cascading.
\newblock {\em Algorithmica}, 5(1--4):215--241, 1990.
\newblock \href {https://doi.org/10.1007/BF01840386}
  {\path{doi:10.1007/BF01840386}}.

\bibitem{mn-lpcgc-99}
Kurt Mehlhorn and Stefan Näher.
\newblock {\em The {LEDA} Platform of Combinatorial and Geometric Computing}.
\newblock Cambridge University Press, 1999.
\newblock \href {https://doi.org/10.1145/204865.204889}
  {\path{doi:10.1145/204865.204889}}.

\bibitem{DBLP:journals/corr/abs-2101-00326}
Joseph S.~B. Mitchell.
\newblock Approximating maximum independent set for rectangles in the plane.
\newblock {\em CoRR}, abs/2101.00326, 2021.
\newblock URL: \url{https://arxiv.org/abs/2101.00326}, \href
  {http://arxiv.org/abs/2101.00326} {\path{arXiv:2101.00326}}.

\bibitem{nguyen2008constant}
Huy~N. Nguyen and Krzysztof Onak.
\newblock Constant-time approximation algorithms via local improvements.
\newblock In {\em 49th Annual {IEEE} Symposium on Foundations of Computer
  Science, {FOCS} 2008, October 25-28, 2008, Philadelphia, PA, {USA}}, pages
  327--336. {IEEE} Computer Society, 2008.
\newblock \href {https://doi.org/10.1109/FOCS.2008.81}
  {\path{doi:10.1109/FOCS.2008.81}}.

\bibitem{pardalos1994maximum}
Panos~M Pardalos and Jue Xue.
\newblock The maximum clique problem.
\newblock {\em Journal of Global Optimization}, 4(3):301--328, 1994.
\newblock \href {https://doi.org/10.1007/BF01098364}
  {\path{doi:10.1007/BF01098364}}.

\bibitem{DBLP:journals/cartographica/RylovR14}
Maxim~A. Rylov and Andreas~W. Reimer.
\newblock A comprehensive multi-criteria model for high cartographic quality
  point-feature label placement.
\newblock {\em Cartogr. Int. J. Geogr. Inf. Geovisualization}, 49(1):52--68,
  2014.
\newblock \href {https://doi.org/10.3138/carto.49.1.2137}
  {\path{doi:10.3138/carto.49.1.2137}}.

\bibitem{DBLP:journals/tog/SanderNCH08}
Pedro~V. Sander, Diego Nehab, Eden Chlamtac, and Hugues Hoppe.
\newblock Efficient traversal of mesh edges using adjacency primitives.
\newblock {\em {ACM} Trans. Graph.}, 27(5):144, 2008.
\newblock \href {https://doi.org/10.1145/1409060.1409097}
  {\path{doi:10.1145/1409060.1409097}}.

\bibitem{DBLP:conf/apvis/ThomBKWE12}
Dennis Thom, Harald Bosch, Steffen Koch, Michael W{\"{o}}rner, and Thomas Ertl.
\newblock Spatiotemporal anomaly detection through visual analysis of
  geolocated twitter messages.
\newblock In {\em Pacific Visualization (PacificVis'12)}, pages 41--48. {IEEE},
  2012.

\bibitem{vanBevern2015}
Ren{\'e} van Bevern, Matthias Mnich, Rolf Niedermeier, and Mathias Weller.
\newblock Interval scheduling and colorful independent sets.
\newblock {\em Journal of Scheduling}, 18(5):449--469, 2015.
\newblock \href {https://doi.org/10.1007/s10951-014-0398-5}
  {\path{doi:10.1007/s10951-014-0398-5}}.

\bibitem{DBLP:conf/compgeom/KreveldSW98}
Marc~J. van Kreveld, Tycho Strijk, and Alexander Wolff.
\newblock Point set labeling with sliding labels.
\newblock In Ravi Janardan, editor, {\em Proceedings of the Fourteenth Annual
  Symposium on Computational Geometry, Minneapolis, Minnesota, USA, June 7-10,
  1998}, pages 337--346. {ACM}, 1998.
\newblock \href {https://doi.org/10.1145/276884.276922}
  {\path{doi:10.1145/276884.276922}}.

\bibitem{ww-pla-95}
Frank Wagner and Alexander Wolff.
\newblock A practical map labeling algorithm.
\newblock {\em Comput. Geom. Theory Appl.}, 7:387--404, 1997.
\newblock \href {https://doi.org/10.1016/S0925-7721(96)00007-7}
  {\path{doi:10.1016/S0925-7721(96)00007-7}}.

\bibitem{DBLP:journals/comgeo/WelzlWW97}
Frank Wagner and Alexander Wolff.
\newblock A practical map labeling algorithm.
\newblock {\em Comput. Geom.}, 7:387--404, 1997.
\newblock \href {https://doi.org/10.1016/S0925-7721(96)00007-7}
  {\path{doi:10.1016/S0925-7721(96)00007-7}}.

\bibitem{DBLP:journals/jacm/WillardL85}
Dan~E. Willard and George~S. Lueker.
\newblock Adding range restriction capability to dynamic data structures.
\newblock {\em J. {ACM}}, 32(3):597--617, 1985.
\newblock \href {https://doi.org/10.1145/3828.3839}
  {\path{doi:10.1145/3828.3839}}.

\bibitem{DBLP:journals/toc/Zuckerman07}
David Zuckerman.
\newblock Linear degree extractors and the inapproximability of max clique and
  chromatic number.
\newblock {\em Theory Comput.}, 3(1):103--128, 2007.
\newblock \href {https://doi.org/10.4086/toc.2007.v003a006}
  {\path{doi:10.4086/toc.2007.v003a006}}.

\end{thebibliography}

\end{document}